\def\llncs{0}%For LLNCS. Note that this may cause errors in the bibliography.
\def\lipics{0} %For LIPIcs. 

\def\sigconf{0}  %ACM SIG conference
\ifnum\sigconf=0
    \RequirePackage{etex}
\fi

\def\cryptology{0} % For journal of cryptology

\def\bigfont{0}  % For using a huge font

\def\masterthesis{0}
\def\preprint{1}

\def\quantumjournal{0}%cannot turn any of the previous modes when quantum journal mode is on.

\def\draft{0}
\def\anonymous{0}
\def\shownomenclature{1} %nomenclature is incompatible with sigconf
\def\toc{1} %Show TOC. Should be set to 1 if paper is long (say, above 30 pages), and for all types of thesis. NOTE THAT TOC is incompatible with sigconf

\newif\ifab 
\abfalse

\newif\iffv 
\fvtrue
\newif\ifnv
\nvtrue
\pdfoutput=1
\ifnum
    \ifnum\preprint=1 1\else\ifnum\cryptology=1 1\else 0\fi\fi
=1 %If preprint or cryptology
    \ifnum\draft=1
       \documentclass[11pt,draft,a4paper]{article}
       \pdfoutput=1
    \else
        \documentclass[11pt,final,a4paper]{article}
        \pdfoutput=1
    \fi
    \usepackage{lmodern}
\fi
\ifnum\masterthesis=1
    \documentclass[11pt,final,a4paper,titlepage]{article}
    \usepackage{lmodern}
\fi
\ifnum\sigconf=1
        \documentclass[sigconf,final]{acmart}
\fi
\ifnum\bigfont=1
    \documentclass[final,17pt]{extarticle} \usepackage{fullpage}
\fi
\ifnum\llncs=1
    \documentclass[runningheads]{llncs}
\fi
\ifnum\lipics=1
    \input{lipics_header}
\fi
\ifnum\quantumjournal=1
  \documentclass[a4paper,onecolumn,11pt %,accepted=YYYY-MM-DD
  ]
  {quantumarticle}
  \pdfoutput=1

\fi

\ifnum\toc=1\ifnum\lipics=0
    \usepackage[nottoc]{tocbibind} %This adds the bibliography to the table of contents (of course, only if it exists). There is no need to add the table of contents page to the table of contents, hence the option "nottoc".
\fi\fi
\ifnum\lipics=0\ifnum\sigconf=0\usepackage{authblk}\fi\fi 
\ifnum\sigconf=0
    \usepackage{silence} %authbulk must loaded before the silence package!!
    \usepackage{amssymb} %incompatible with sigconf
\fi
\ifnum\lipics=0 % LIPIcs has its own import for line numbering
    \usepackage[modulo]{lineno}
\fi

\ifnum\llncs=1

\fi \usepackage{amsthm}

\usepackage{amsfonts,latexsym,color,paralist,url,ifdraft,mathrsfs,thm-restate,comment,bbold,booktabs,ifthen,wrapfig}

\usepackage{footnote}

\usepackage[utf8]{inputenc} \usepackage[autostyle=false, style=english]{csquotes} \MakeOuterQuote{"}

\usepackage[T1]{fontenc}
\ifnum\sigconf=0
    \ifnum \shownomenclature=1
        \usepackage[refpage,prefix]{nomencl}
        \newcommand*{\nom}[3]{#2\nomenclature[#1]{#2}{#3}} 
        \makenomenclature
  \makeatletter
        \def\thenomenclature{%
          \providecommand*{\listofnlsname}{\nomname}%

              \section{\nomname}
              \label{sec:nomenclature}    
              \if@intoc\addcontentsline{toc}{section}{\nomname}\fi%

          \nompreamble
          \if@nomentbl
            \let\itemOrig=\item
            \def\item{\gdef\item{\\}}%
            \expandafter\longtable\expandafter{\@nomtableformat}
          \else
            \list{}{%
              \labelwidth\nom@tempdim
              \leftmargin\labelwidth
              \advance\leftmargin\labelsep
              \itemsep\nomitemsep
              \let\makelabel\nomlabel}%
          \fi
        }
        \makeatother 
    \fi
    \ifnum\lipics=0    
        \usepackage[draft=false,pageanchor]{hyperref}
        \usepackage[final]{graphicx}
    \fi
    \usepackage{doi}
\fi

\ifdraft{\usepackage{showlabels}}{} %must be after use package{hyperref}.

\usepackage [ lambda,advantage , operators , sets , adversary , landau , probability , notions , logic ,ff, mm,primitives , events , complexity , asymptotics , keys]{cryptocode}

\renewcommand{\verify}{\ensuremath{\mathsf{verify}}}

\ifdraft{\newcommand{\authnote}[3]{{\color{#3} {\bf  #1:} #2}}}{\newcommand{\authnote}[3]{}}
\newcommand{\bra}[1]{\langle #1 \vert}
\newcommand{\ket}[1]{\vert #1 \rangle}
\newcommand{\ketbra}[1]{\vert #1 \rangle \langle #1 \vert}
\newcommand{\braket}[2]{\langle #1 \vert #2 \rangle}

\newcommand{\tensor}{\otimes}
\DeclareMathOperator{\tr}{Tr}

\ifdraft{\linenumbers}{}

\usepackage{algorithm,algorithmicx,algpseudocode}
\usepackage[normalem]{ulem}

\usepackage{tikz}
\usetikzlibrary{arrows,chains,matrix,positioning,scopes,shapes}

\makeatletter
\def\squarecorner#1{
    \pgf@x=\the\wd\pgfnodeparttextbox%
    \pgfmathsetlength\pgf@xc{\pgfkeysvalueof{/pgf/inner xsep}}%
    \advance\pgf@x by 2\pgf@xc%
    \pgfmathsetlength\pgf@xb{\pgfkeysvalueof{/pgf/minimum width}}%
    \ifdim\pgf@x<\pgf@xb%
        \pgf@x=\pgf@xb%
    \fi%
    \pgf@y=\ht\pgfnodeparttextbox%
    \advance\pgf@y by\dp\pgfnodeparttextbox%
    \pgfmathsetlength\pgf@yc{\pgfkeysvalueof{/pgf/inner ysep}}%
    \advance\pgf@y by 2\pgf@yc%
    \pgfmathsetlength\pgf@yb{\pgfkeysvalueof{/pgf/minimum height}}%
    \ifdim\pgf@y<\pgf@yb%
        \pgf@y=\pgf@yb%
    \fi%
    \ifdim\pgf@x<\pgf@y%
        \pgf@x=\pgf@y%
    \else
        \pgf@y=\pgf@x%
    \fi
    \pgf@x=#1.5\pgf@x%
    \advance\pgf@x by.5\wd\pgfnodeparttextbox%
    \pgfmathsetlength\pgf@xa{\pgfkeysvalueof{/pgf/outer xsep}}%
    \advance\pgf@x by#1\pgf@xa%
    \pgf@y=#1.5\pgf@y%
    \advance\pgf@y by-.5\dp\pgfnodeparttextbox%
    \advance\pgf@y by.5\ht\pgfnodeparttextbox%
    \pgfmathsetlength\pgf@ya{\pgfkeysvalueof{/pgf/outer ysep}}%
    \advance\pgf@y by#1\pgf@ya%
}
\makeatother

\pgfdeclareshape{square}{
    \savedanchor\northeast{\squarecorner{}}
    \savedanchor\southwest{\squarecorner{-}}

    \foreach \x in {east,west} \foreach \y in {north,mid,base,south} {
        \inheritanchor[from=rectangle]{\y\space\x}
    }
    \foreach \x in {east,west,north,mid,base,south,center,text} {
        \inheritanchor[from=rectangle]{\x}
    }
    \inheritanchorborder[from=rectangle]
    \inheritbackgroundpath[from=rectangle]
}

\makeatletter
\tikzset{join/.code=\tikzset{after node path={%
\ifx\tikzchainprevious\pgfutil@empty\else(\tikzchainprevious)%
edge[every join]#1(\tikzchaincurrent)\fi}}}
\makeatother
\tikzset{>=stealth',every on chain/.append style={join},
         every join/.style={->}}
\tikzstyle{labeled}=[execute at begin node=$\scriptstyle,
   execute at end node=$]
\usepackage[capitalise]{cleveref} %Incompatible with showonlyrefs
\crefname{Game}{Game}{Games}

\usepackage{autonum}%must be loaded after cleveref. Sometimes causes bugs. The alternative is incomaptible with cleveref:
\ifthenelse{
\(\equal{\cryptology}{1}\OR
\(\equal{\preprint}{1} \OR \equal{\masterthesis}{1}\) \OR \(\equal{\bigfont}{1} \OR \equal{\quantumjournal}{1} \)
\) %Overleaf's linter incorrectly thinks there's a problem in this line
}
{
    \newtheorem{theorem}{Theorem} 
    \newtheorem{lemma}{Lemma}
    \newtheorem{corollary}{Corollary}
    \newtheorem{proposition}{Proposition}
    \newtheorem{observation}{Observation}  %%added externally. Was not in macros.
    \newtheorem{definition}{Definition}

    \newtheorem{claim}{Claim}
    \newtheorem{fact}{Fact}
    \newtheorem*{theorem*}{Theorem}
    \newtheorem*{lemma*}{Lemma}
    \newtheorem*{corollary*}{Corollary}
    \newtheorem*{proposition*}{Proposition}
    \newtheorem*{claim*}{Claim}
    \theoremstyle{definition}
    \newtheorem{openproblem}{Open Problem}
    \theoremstyle{remark}
    \newtheorem{remark}[theorem]{Remark}

    \theoremstyle{plain}
    
}{}  
\ifnum\sigconf=1

  \newtheorem*{theorem*}{Theorem}
  \newtheorem*{lemma*}{Lemma}
  \newtheorem*{corollary*}{Corollary}
  \newtheorem*{proposition*}{Proposition}
  \newtheorem*{observation}{Observation} %%added externally. Was not in macros.
  \newtheorem*{claim*}{Claim}
  \theoremstyle{definition}
  \newtheorem{openproblem}{Open Problem}

\fi
\theoremstyle{plain}
\newcounter{Game} 
\newenvironment{game}[1][htb]
    {
    \floatname{algorithm}{Game}% Update algorithm name
    \refstepcounter{Game}%to call the Game pointer. Currently the problem is that caption command by default calls the algorithm pointer. For the moment we need to use Game~~\ref{blah}
   \begin{algorithm}[#1]%
  }{\end{algorithm}}

\expandafter\let\expandafter\savedflalignstar\csname flalign*\endcsname
\expandafter\let\expandafter\savedendflalignstar\csname endflalign*\endcsname
\AtBeginDocument{%
  \expandafter\let\csname flalign*\endcsname\savedflalignstar
  \expandafter\let\csname endflalign*\endcsname\savedendflalignstar
}

\newcommand{\ab}[1]{}  %%

\ifnum\masterthesis=1
    \usepackage{setspace}
    \usepackage{fancyhdr}
    \newcommand{\longunderline}{\underline{\ \ \ \ \ \ \ \ \ \ \ \ \ \ \ \ \ \ \ \ }\ }
    
    \newenvironment{changemargin}[2]{%
        \begin{list}{}{%
        \setlength{\topsep}{0pt}%
        \setlength{\leftmargin}{#1}%
        \setlength{\rightmargin}{#2}%
        \setlength{\listparindent}{\parindent}%
        \setlength{\itemindent}{\parindent}%
        \setlength{\parsep}{\parskip}%
        }%
    \item[]}{\end{list}}

\fi
\ifnum\masterthesis=1
\usepackage{pdfpages}%for_the hebrew title
\fi
\newcommand{\onote}[1]{\authnote{Or}{#1}{blue}}
\newcommand{\anote}[1]{\authnote{Amit}{#1}{red}}

\usepackage{multirow}
\usepackage{booktabs}
\usepackage{adjustbox}
\usepackage{algorithm,algorithmicx,algpseudocode}
\makeatother  
\usepackage{xspace}
\usepackage{textcomp} \newcommand{\cent}{\text{\textcent}}
\newcommand{\PRS}{\textsf{PRS}\xspace}
\newcommand{\prs}{\text{pseudorandom states}\xspace}
\newcommand{\io}{\textsf{iO}\xspace}
\newcommand{\mill}{\includegraphics[scale=.03]{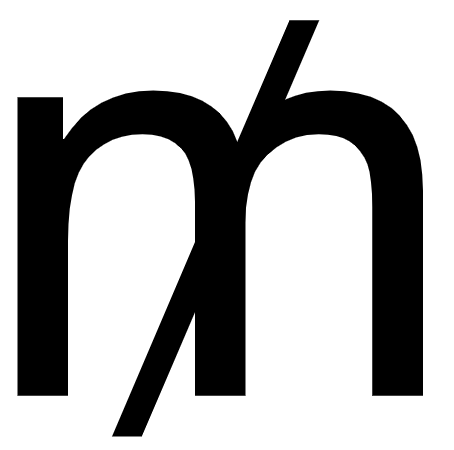}}
\newcommand{\keygen}{\ensuremath{\mathsf{key\textit{-}gen}}}
\newcommand{\prkeygen}{\ensuremath{\mathsf{private\textit{-}key\textit{-}gen}}}
\newcommand{\pkkeygen}{\ensuremath{\mathsf{public\textit{-}key\textit{-}gen}}}
\renewcommand{\verify}{\ensuremath{\mathsf{verify}}} \newcommand{\accept}{\ensuremath{\mathsf{accept}}}

\newcommand{\bank}{\ensuremath{\mathsf{mint}}}
\newcommand{\Count}{\ensuremath{\mathsf{Count}}}
\newcommand{\Ima}{\ensuremath{\text{Im}}}

\newcommand{\EE}{\mathbb{E}}

\newcommand{\pkqc}{\ensuremath{\mathsf{Pk\textit{-}QC}}}
\newcommand{\prqc}{\ensuremath{\mathsf{Pr\textit{-}QC}}}

\newcommand{\nrsp}{\textsf{nonadaptive\--secure\--against\--public\--sabotage}\xspace}
\newcommand{\nrrsp}{\textsf{nonadaptive\--rationally\--secure\--against\--public\--sabotage}\xspace}
\newcommand{\mnrrsp}{\textsf{multiverifier\--nonadaptive\--rationally\--secure\--against\--public\--sabotage}\xspace}
\newcommand{\ut}{\textsf{untraceable}\xspace}

\newcommand{\rnauf}{\textsf{nonadaptive\--rational\-ly\--un\-forgeable}\xspace}
\newcommand*{\arnauf}{\textsf{all\--or\--nothing\--nonadaptive\--rational\-ly\--un\-forgeable}\xspace}
\newcommand*{\frnauf}{\textsf{flexible\--nonadaptive\--rational\-ly\--un\-forgeable}\xspace}

\newcommand*{\frauf}{\textsf{flexible\--adaptive\--rational\-ly\--un\-forgeable}\xspace}

\newcommand*{\ufla}{U_\textsf{flexible}^{\textsf{adapt}}\xspace}
\newcommand*{\uflna}{U_\textsf{flexible}^{\textsf{nonadapt}}\xspace}
\newcommand*{\uana}{U_\textsf{all\--or\--nothing}^{\textsf{adapt}}\xspace}
\newcommand*{\uanna}{U_\textsf{all\--or\--nothing}^{\textsf{nonadapt}}\xspace}
\newcommand*{\uanmna}{U_\textsf{all\--or\--nothing}^{\textsf{multiver\--nonadapt}}\xspace}
\newcommand*{\uflmna}{U_\textsf{flex}^{\textsf{multiver\--nonadapt}}\xspace}
\newcommand*{\lanna}{L_\textsf{public-sabotage}^{\textsf{nonadapt}}\xspace}

\newcommand{\narpf}{\textsf{nonadaptive\--rationally\--secure\--against\--private\--sabotage}\xspace}
\newcommand{\narf}{\textsf{nonadaptive\--rationally\--secure\--against\--sabotage}\xspace}

\newcommand{\nars}{\textsf{nonadaptive\--rationally\--secure}\xspace}

\newcommand{\uf}{\textsf{unforgeable}\xspace}

\newcommand{\suf}{\textsf{standard\--unforgeable}\xspace}
\newcommand{\ruf}{\textsf{rationally\--unforgeable}\xspace}

\newcommand{\napf}{\textsf{nonadaptive\--secure\--against\--private\--sabotage}\xspace}

\newcommand*{\anauf}{\textsf{all\--or\--nothing\--nonadaptive\--un\-forgeable}\xspace}
\newcommand*{\fnauf}{\textsf{flexible\--nonadaptive\--un\-forgeable}\xspace}

\newcommand*{\fauf}{\textsf{flexible\--adaptive\--un\-forgeable}\xspace}

\newcommand{\nauf}{\textsf{nonadaptive\--unforgeable}\xspace}
\newcommand{\auf}{\textsf{adaptive\--unforgeable}\xspace}

\newcommand{\nauuf}{\textsf{nonadaptive\--unconditionally\--unforgeable}\xspace}
\newcommand{\auuf}{\textsf{adaptive\--unconditionally\--unforgeable}\xspace}

\newcommand{\mnauf}{\textsf{multiverifier\--nonadaptive\--unforgeable}\xspace}
\newcommand{\mnaruf}{\textsf{multiverifier\--nonadaptive\--rational\--unforgeable}\xspace}
\newcommand{\mnapf}{\textsf{multiverifier\--nonadaptive\--secure\--against\--private\--sabotage}\xspace}
\newcommand{\mnarpf}{\textsf{multiverifier\--nonadaptive\--rational\--secure\--against\--\--private\--sabotage}\xspace}
\newcommand{\mnars}{\textsf{multiverifier\--nonadaptive\--rational\--secure}\xspace}

\newcommand{\MS}{\ensuremath{\mathcal{M}}}
\newcommand{\Sym}[1]{\ensuremath{\mathbb{Sym}^{#1}}}
\newcommand{\BasisSym}[1]{\ensuremath{Sym}^{#1}}
\newcommand{\SymTilde}[1]{\ensuremath{\mathbb{\widetilde{Sym}}^{#1}}}
\newcommand{\BasisSymTilde}[1]{\ensuremath{\widetilde{Sym}}^{#1}}
\newcommand{\Good}[2]{\ensuremath{\mathbb{Good}^{{#1}, {#2}}}}

\newcommand{\GoodTilde}[2]{\ensuremath{\widetilde{\mathbb{Good}}^{{#1}, {#2}}}}
\newcommand{\Bad}[2]{\ensuremath{\mathbb{Bad}}^{{#1}, {#2}}}
\newcommand{\BadTilde}[2]{\ensuremath{\widetilde{\mathbb{Bad}}^{{#1}, {#2}}}}
\newcommand{\GoodSym}[2]{\ensuremath{\mathbb{GoodSym}^{{#1},{#2}}}}
\newcommand{\BasisGoodSym}[2]{\ensuremath{GoodSym^{{#1},{#2}}}}
\newcommand{\GoodSymTilde}[2]{\ensuremath{\mathbb{\widetilde{GoodSym}}^{{#1},{#2}}}}
\newcommand{\BasisGoodSymTilde}[2]{\ensuremath{\widetilde{GoodSym}^{{#1},{#2}}}}
\newcommand{\HTilde}[1]{\ensuremath{\mathbb{\widetilde{H}}}^{#1}}
\newcommand{\Qoperator}{\ensuremath{Q}}
\newcommand{\Poperator}[2]{\ensuremath{P}_{{#1}, {#2}}}
\usepackage{float}
\floatstyle{boxed}
\newfloat{boxx}{thb}{lob}
\floatname{boxx}{Game}

\begin{document}
\ifnum\masterthesis=0
    \title{Almost Public Quantum Coins
    %\ifdraft{\\(working draft)}
    }
\fi%%%%till this point.
% \title{Almost Public Quantum Coins
% %\ifdraft{\\(working draft)}
% }

%%%%%%%%%Old-code%%%%%%%%%
% \ifnum\llncs=1
%     \institute{}
%     \ifnum\anonymous=0
%     \author[1]{Amit Behera\inst{1}, \and Or Sattath\inst{1} 
%     \institute{Computer Science Department, Ben-Gurion University}}
%     \fi
% \else
%     \ifnum\anonymous=0
%         \ifnum\masterthesis=0
%             \ifnum\sigconf=0
%                 \author[1]{Amit Behera}
%                 \author[1]{Or Sattath}
%                 \affil[1]{Computer Science Department, Ben-Gurion University}
%             \else
%                 \author{Or Sattath}
%                 \affiliation{%
%                 \institution{Computer Science Department, Ben-Gurion University}
%                 \country{Israel}}
%             \fi
%         \fi
%     \else
%         \ifnum\sigconf=0
%             \author{}
%         \fi
%     \fi

% \fi

% \ifnum\masterthesis=0
%     \ifnum\sigconf=0
%         \maketitle
%     \fi
% \fi

%%%old-stuff complete

%%%updated-from-the-template%%%

\ifnum\anonymous=0
    \ifnum\preprint=1
        \author[1]{Amit Behera}
        \author[1]{Or Sattath}
        \affil[1]{Computer Science Department, Ben-Gurion University of the Negev}
    \fi
    \ifnum\cryptology=1
        \author{Amit Behera}
        \affil{Department of Computer Science, Ben Gurion University of the Negev, Beersheba, Israel\\
        behera@post.bgu.ac.il}
        \author{Or Sattath}
        \affil{Department of Computer Science, Ben Gurion University of the Negev, Beersheba, Israel\\
        sattath@bgu.ac.il}
    \fi
    \ifnum\sigconf=1
        \author{Or Sattath}
            \affiliation{%
              \institution{Department of Computer Science, Ben Gurion University of the Negev}
              \city{Beersheba}
              \country{Israel}}
        \email{sattath@bgu.ac.il}

        \author{Amit Behera}
            \affiliation{%
              \institution{Department of Computer Science, Ben Gurion University of the Negev}
              \city{Beersheba}
              \country{Israel}}
        \email{behera@post.bgu.ac.il}
    \fi
    \ifnum\llncs=1
        \author{Amit Behera\inst{1}\orcidID{0009-0008-1462-7222} \and
        Or Sattath\inst{1}\orcidID{0000-0001-7567-3822} }
        \authorrunning{A. Behera, O. Sattath}
        \institute{Department of Computer Science, Ben Gurion University of the Negev, Beersheba, Israel 
        \email{sattath@bgu.ac.il} \and
        Department of Computer Science, Ben Gurion University of the Negev, Beersheba, Israel\\
        \email{behera@post.bgu.ac.il}}
    \fi
    \ifnum\quantumjournal=1
        \author[1]{Amit Behera}
        \orcid[1]{0009-0008-1462-7222}
        \email{behera@post.bgu.ac.il}
        \author[2]{Or Sattath}
        \orcid[2]{0000-0001-7567-3822}
        \email{sattath@bgu.ac.il}

        \affil{Computer Science Department, Ben-Gurion University of the Negev}
        
    \fi
\else
    \ifnum\lipics=0
        \author{}
    \fi
\fi

\ifthenelse{\equal{\masterthesis}{0} \AND \equal{\sigconf}{0}}{
    \maketitle
}

%%%%%updated-content ends

% \ifnum\masterthesis=0
%     \iffv
    
%     \fi
% \fi

\iffv
\ifnum\masterthesis=1
    
    \begin{titlepage}
        \centering
        { Ben-Gurion University of the Negev}
        
        {The Faculty of Natural Sciences}
        
        {\small The Department of Computer Science}
        
        \vspace{2cm}
        
        {\Large \bfseries Almost Public Quantum Coins}
        
        \vspace{2cm}
        
        {\small Thesis submitted in partial fulfillment of the requirements for the Master of Sciences degree}
        
        \vspace{1cm}
        
        {\bfseries Amit Behera}
        
        {Under the supervision of Dr. Or Sattath}
        
        \vspace{2cm}
        
        \today
    \end{titlepage}
    
    \begin{titlepage}
        \centering
        { Ben-Gurion University of the Negev}
        
        {The Faculty of Natural Sciences}
        
        {\small The Department of Computer Science}
        
        \vspace{2cm}
        
        {\Large \bfseries Almost Public Quantum Coins}
        
        \vspace{2cm}
        
        {\small Thesis submitted in partial fulfillment of the requirements for the Master of Sciences degree}
        
        \vspace{1cm}
        
        {\bfseries Amit Behera}
        
        {Under the supervision of Dr. Or Sattath}
        
        \vspace{1cm}
        
        {\small Signature of student: \longunderline Date: \longunderline}
        
        \vspace{0.5cm}
        
        {\small Signature of supervisor: \longunderline Date: \longunderline}
        
        \vspace{0.5cm}
        
        \begin{changemargin}{-1cm}{-1cm}
        \centering
            {\small Signature of the committee for graduate studies: \longunderline Date: \longunderline}
        \end{changemargin}
        \vspace{2cm}
        
        \today
    \end{titlepage}
\fi

\ifnum\masterthesis=1
    \pagenumbering{roman}
    \begin{center}
        {\large \bfseries Almost Public Quantum Coins}
        
        \vspace{0.5cm}
        
        {\bfseries Amit Behera}
        
        \vspace{0.5cm}
        
        Thesis submitted in partial fulfillment of the requirements for the Master of Sciences degree
        
        \vspace{0.25cm}
        
        {Ben-Gurion University of the Negev}
        
        \vspace{0.25cm}
        
        \today
        
        \vspace{2cm}
        
        {\bfseries Abstract}
        
        \vspace{0.5cm}
    \end{center}
\else    
    \begin{abstract}
\fi

In a quantum money scheme, a bank can issue money that users cannot counterfeit. Similar to bills of paper money, most quantum money schemes assign a unique serial number to each money state, thus potentially compromising the privacy of the users of quantum money.
However in a quantum coins scheme, just like the traditional currency coin scheme, all the money states are exact copies of each other, providing a better level of privacy for the users.

A quantum money scheme can be private, i.e., only the bank can verify the money states, or public, meaning anyone can verify. In this work, we propose a way to lift any private quantum coin  scheme---which is known to exist based on Pseudorandom States~\cite{JLS18}%Ji, Liu, and Song (CRYPTO'18) 
---to a scheme that closely resembles a public quantum coin scheme. Verification of a new coin is done by comparing it to the coins the user already possesses, by using a projector on to the symmetric subspace. No public coin scheme was known prior to this work. It is also the first construction that is close to a public quantum money scheme and is provably secure based on standard assumptions.
Finally, the lifting technique, when instantiated with the private quantum coins scheme~\cite{MS10}, gives rise to the first construction that is close to an inefficient unconditionally secure public quantum money scheme.% This is the first example of a scheme that is close to a public quantum scheme and unconditionally secure.
\ifnum\masterthesis=0
    \end{abstract}
\fi

\fi
\ifnum\sigconf=1
    \maketitle

\fi

\ifnum\masterthesis=1
    \pagebreak
    \pagenumbering{arabic}
\fi

\ifnum\masterthesis=1
    \subsection*{Acknowledgments}
    I would like to express my deepest gratitude to my supervisor Dr. Or Sattath for her continued assistance and guidance throughout the research period and the writing of this thesis. It has been an absolute pleasure learning from and working alongside him, and I hope to continue doing so in the future. I would also like to thank all my friends, teachers and mentors both in Israel and India for their throughout love and support. A special thanks to Himalaya Senapati for guiding me like an elder brother at all times. Finally, I would like to thank my parents: Ajay Kumar Behera and Arati Behera for their unparalleled support, dedication and love towards me. I don't think it is ever possible for me to repay the debt of gratitude I owe them. This thesis, like almost every other thing in my life, owes them a lot and I dedicate it to them.
\fi

\ifnum\masterthesis=1
    \setcounter{tocdepth}{3}
    \tableofcontents
    
    \listoffigures
\fi %%%this part about abstract and acknowledgement is still old. Needs to be updated as per template.
%The main results and the contributions are summarized in the Main results (\cref{pg:main-results}) and Scientific Contributions (\cref{pg:sceintific_constributions}) subsections of the introduction, on \cpageref{pg:main-results} and  \cpageref{pg:sceintific_constributions}, respectively. 
\ifnum\toc=1
    \ifnum\llncs=1
        \setcounter{secnumdepth}{3}
        \setcounter{tocdepth}{3}
    \fi
    \tableofcontents 
\fi

\ifab

\fi

%\maketitle

%KEYWORDS: Quantum money, Quantum Coins, Public Quantum Money
\ifab
\vspace{-20pt}
\fi
 
\iffv

\section{Introduction}

\fi 

\iffv An analog of the traditional monetary system, quantum money comprises of quantum money states that are issued by the bank and that are used for transactions. \fi 
A quantum money scheme can be \emph{private} or \emph{public}. In the private scenario, only the bank, the entity that issued the money, can verify its authenticity, whereas in the public scenario, the bank generates a public key that anyone can use to verify the quantum money. \iffv While public quantum money is suitable for use in a setting such as our current cash system, private quantum money is applicable in settings such as travel tickets, wherein we do not expect users to transact with anyone other than the ticket issuer. \fi The second characterization refers to whether quantum money users are given \emph{bills} or \emph{coins}. In a bill scheme, each quantum money state is unique and is usually associated, with a distinct classical serial number. On the other hand, in a coin scheme, all quantum states are exact copies\iffv, and therefore supposed to be indistinguishable from each other\fi.  \ifab A bill, unlike a coin, is marked with a distinct serial number, which can be used to track the bill, compromising the users' privacy. Coins are indistinguishable from one another and hence, cannot be tracked. Therefore, coins provide better privacy than bills. \fi  In both variants, the quantum money scheme overall is said to be secure if the users cannot counterfeit the quantum money state.  The quantum setting is well suited to prove such unforgeability property due to the quantum no-cloning theorem~\cite{WZ82,Par70,Die82}. % It has been extremely difficult to construct public quantum money let alone public quantum coins. After a long line of research including many failed constructions, we have constructions of public quantum bills from strong assumptions~\cite{Zha21,Shm22} the necessity of which is unclear. We have partial impossibility results showing that certain assumptions are insufficient to build public quantum money. Currently, constructing public quantum money from minimal assumptions or understanding the minimal assumptions to construct public quantum money stands as one of the major open problems in quantum cryptography. 
%We have partial results both in terms of upper bounds, i.e., constructions as well as lower bounds, i.e., impossibility results showing that certain assumptions are not sufficient to build public quantum money. Currently, constructing public quantum money from minimal assumptions or understanding the minimal assumptions to construct public quantum money stands as one of the major open problems in quantum cryptography. 

\iffv
\ifab \vspace{-15pt} \fi \subsection{Related works}\label{pg:related-works}
Quantum money has been studied extensively in investigations of private bills~\cite{Wie83,BBB+83,TOI03,Gav12,GK15}, public bills~\cite{Aar09,FGH+12,Zha21,Shm22,Kan18,LMZ23}, and private coins~\cite{MS10,JLS18,AMR20}. 

The security of the private schemes is generally solid, and some schemes, such as that of Wiesner, are unconditionally secure~\cite{MVW12,PYJ+12}. 
Mosca and Stebila constructed an \emph{inefficient} 
\iffv
(see \cref{definition:inefficient quantum money})
\fi 
private coin scheme, in which the coin is an $n$ qubit state sampled uniformly from the Haar measure~\cite{MS10}. The construction for private coins\footnote{\cite{JLS18} do not make the distinction between coins and bills and hence do not call their construction a quantum coin scheme, even though it clearly is.} by Ji, Liu, and Song is based on Pseudorandom States (\PRS)~\cite{JLS18}. The construction of \PRS was later simplified by~\cite{BS19}. \anote{I think Brakrski Shmueli should be removed from here:  The construction was later simplified by \cite{BS19}. This is arguably one of the weakest computational assumption possible in quantum cryptography.}
Another recent work by Alagic, Majenz and Russell~\cite{AMR20} provides a \emph{stateful}\footnote{i.e., the bank maintains a (quantum) state which is used and updated with every mint and verify query.} construction for private quantum coins by simulating Haar random states. Their construction is unconditionally secure and shares many of the properties of the work by Mosca and Stebila while still being efficient. Of course, the obvious disadvantage is the statefulness of this scheme.

In contrast to the private scenario, it has been extremely difficult to construct public quantum money schemes, let alone public quantum coins. Several constructions of public quantum money were broken. Aaronson's scheme~\cite{Aar09} was broken in Ref.~\cite{LAF+10}. Aaronson and Christiano's scheme~\cite{AC13} was broken in Refs.~\cite{PFP15,Aar16,BS16a,PDF+19} and a fix using quantum-secure Indistinguishability Obfuscation ($\io$) was suggested in~\cite{BS16a} and proved to be secure in~\cite{Zha21}. There have been other constructions of public quantum money based on \io, see~\cite{CLL+21,Shm22}. Unfortunately, quantum secure \io is still not a standard assumption, as we still lack sufficient understanding of it. In particular, the existing constructions are either based on pre-quantum assumptions~\cite{JLS20}, or else they are candidate constructions that we do not know how to instantiate from standard assumptions, see~\cite{GGH15,BGM+18,BDGM20,WW20}. Moreover, it is unclear why \io should be necessary to construct public quantum money.  %Similarly, in~\cite{LMZ23}, it was shown that a certain class of lattice-based constructions of public quantum money is insecure. Currently, constructing public quantum money from minimal assumptions or understanding the minimal assumptions to construct public quantum money stands as one of the major open problems in quantum cryptography. 
%As the authors are not aware of \io schemes that explicitly claim to be quantum secure, this \io based construction cannot be instantiated at this point. 

Another construction by Zhandry~\cite{Zha21}, called quantum lightning, is based on a non-standard hardness assumption and was broken for the most possible choice of the parameters~\cite{Rob19}. Farhi et al.~\cite{FGH+12} constructed a quantum money scheme using elegant techniques from knot theory, but their construction only has a partial security reduction~\cite{Lut11} to a non-standard hardness assumption. Some of the ideas in~\cite{FGH+12} have been abstracted into a framework in~\cite{LMZ23} which also shows other instantiations of the framework based on non-standard assumptions. In~\cite{Kan18}, Kane proposed a new technique to construct a class of public quantum money schemes and showed that a general sub-exponential attack (black-box attack) against such quantum money schemes is impossible. Further, she argues that instantiating such a technique with modular forms could yield a secure public money scheme and provides arguments supporting her claim, but it still lacks security proof at this point. %In a recent talk at the Simon's institute (see~\url{https://youtu.be/8fzLByTn8Xk}), Peter Shor discussed his ongoing (unpublished) work regarding a new construction of public quantum money, based on lattices. In the talk, he argues that the scheme is secured based on post-quantum lattice based assumptions, namely the shortest vector problem. To summarize, even though several public quantum money schemes are known, none  of the existing schemes have a security proof based on standard hardness assumptions.

 In terms of impossibility results, it was shown in~\cite{LMZ23} that a certain class of lattice-based constructions of public quantum money is insecure. Similarly, partial impossibility results show that a black-box construction of public quantum money from collision-resistant hash functions where the verification algorithm makes only classical queries to the hash function is insecure, see~\cite{AHY23}. Our work constructs an intermediate primitive between public and private quantum money schemes but from \prs, a weaker assumption than not only collision-resistant hash but even one-way functions, see Ref~\cite{Kre21}, thus complementing the impossibility result of~\cite{AHY23}. Similarly, in subsequent work, \cite{RZ21} proposes ``Franchised Quantum Money,'' which is another intermediate notion between private and public quantum money that works in a model with a bounded number of users. \cite{RZ21} constructs franchised quantum money based on digital signatures, which can be based on one-way functions.

%\paragraph{Subsequent Works} 
Subsequent to our work, there has been significant progress in the study of pseudorandom states, which is the only assumption used in the construction of our quantum money scheme. In particular, it was shown in~\cite{Kre21} that pseudorandom states are strictly weaker primitive than one-way functions. In light of this result, our work shows that a meaningful form of public quantum money scheme can be constructed based on an assumption strictly weaker than one-way functions. Furthermore, there has been an extensive study of all the applications of pseudorandom states, their variants, and the primitives implied by them, see~\cite{AQY21,MY22a,AGQY22,ALY24}. The collection of these cryptographic primitives that are (potentially) weaker than one-way functions is called ``Microcrypt.'' Some of the notable examples of cryptographic primitives in Microcrypt include the fundamental primitives such as quantum public-key encryption with quantum public keys~\cite{BGH+23,KMN+23}, one-time digital signatures with quantum public keys~\cite{MY22a}, bit-commitments and multi-party computations with quantum communication~\cite{BCK+21,AQY21,MY22a}, etc. The current work adds to this list by showing a weaker yet meaningful form of public quantum money can be achieved in Microcrypt.

\fi

\iffv
\subsection{Coins vs. bills: What difference does it make?} %Similar to the authentication procedures that are in place for traditional currency, the characterization of quantum bills is based on a serial numbering system that renders each bill easily distinguishable from all others.

A currency bill, unlike a coin, is marked with a distinct serial number, which can be used to track the bill and may compromise the users' privacy. \iffv Indeed, one needs to look no further than the police, whose investigations sometimes use ``marked bills'', \fi \ifab For example, police uses marked bill for investigation, \fi a technique that can also be easily exploited by others (e.g., a business may try to learn the identity of its competitor's customers, etc.). On the other hand, coins are supposed to be identical copies of each other, and hence, should be untraceable. Therefore, intuitively coins provide better privacy than bills. \iffv It should be noted that \ifab even then \fi \iffv in reality even though coins are supposed to be identical copies of each other, still \fi  there can be attacks to violate the indistinguishability of coins and therefore violating the privacy of the users. For example, the attacker might use color ink to mark one of the coins and later identify the coin using the mark. The notion of untraceability for quantum money was defined only recently in~\cite{AMR20} and there seems to be a few issues with the definition that we discuss in \cref{appendix:untraceability}.
\fi

\iffv
Similar notions of privacy have been extensively studied in the classical setting. For example, Chaum's ECash~\cite{Cha82} provided anonymity using the notion of blind signatures. The Bitcoin~\cite{Nak08} system stores all the transactions in a public ledger hence making it pseudonymous. Even though it provides pseudonymity, various studies have shown heuristics and approaches which can be used to reveal all the different addresses that belong to the same person~\cite{RS13,RS14,MPJ+16}. The raison d'être of several crypto-currencies and protocols, like CoinJoin~\cite{Max13}, Monero~\cite{Sab13}, Zcoin~\cite{MGGR13}, Zcash~\cite{BCG+14} and Mimblewimble~\cite{Poe16}, is to provide better privacy.

In the quantum setting, quantum bills do not require transactions to be recorded like the block-chain based classical crypto-currencies and hence better in that sense. However, quantum bills rely on serial numbers and thereby prone to privacy threats, since these serial numbers could be recorded by the parties involved in the transactions. 
% Indeed quantum bills do not satisfy the untraceability property for quantum money defined in a recent work~\cite{AMR20} by Alagic et al. In contrast, quantum coins, just like their classical counterparts, are better for the same reason and might satisfy the untraceability definition. In terms of one's privacy, therefore, we view quantum coins as the preferred quantum money format.
\fi

%Usually the currency bills are monitored by central authorities who have the power to track people down

%From the point of view of privacy, this is better than the different kinds of crypto-currencies and ecash systems currently in use. For example, crypto-currencies like bitcoin. Similar to today's cash system, quantum bills allow greater anonymity as we do not need to record the transactions. There have been attempts to address these privacy issues by using privacy-oriented crypto-currencies like Monero, Zcoin, Zcash, etc. A reliance on serial numbers, however, which can be tracked, may compromise user privacy. Indeed, one need look no further than the police, whose investigations sometimes use ``marked bills'', a technique that can also be easily exploited by others (e.g., a business may try to learn the identity of its competitor's customers, etc.). Like their traditional currency counterparts, quantum bills also do not facilitate complete privacy. In contrast, quantum coins of the same denomination are, ideally, identical copies of each other, and hence, they are indistinguishable. In terms of one's security and privacy, therefore, we view quantum coins as the preferred quantum money format, insofar as they provide users with greater privacy protections. 

From a theoretical perspective, quantum coins can be thought of as no-cloning on steroids: the no-cloning theorem guarantees that copying a quantum state is impossible. More formally, given an arbitrary quantum state $\ket{\psi}$, it is impossible to create a two register state $\ket{\phi}$ such that the states $\ket{\phi}$ and $\ket{\psi}\tensor \ket{\psi}$ have high fidelity. This property provides the motivation for quantum money, wherein the use of a variant of the no-cloning theorem precludes counterfeiting quantum bills. Yet the ``na\"ive'' no-cloning theorem cannot guarantee that, given $n$ copies of the same state, one cannot generate $n+1$ copies with high fidelity. In other words, the unforgeability of quantum bills resembles or extends results regarding,  $1 \rightarrow 2$ optimal cloners while that of quantum coins requires that one understand the properties of $n\rightarrow n+1$ optimal cloners---see~\cite{DEM98} and references therein.
\iffv
We stress that a quantum coin forger does not imply a universal cloner for three main reasons:\footnote{In other words, an impossibility result regarding universal $n \rightarrow n+1$ cloning is not sufficient to prove unforgeability of quantum coins.} (a) A coin forger only needs to succeed in cloning the set of coins generated by the scheme; as its name suggests, a \emph{universal} cloner guarantees the fidelity for \emph{all} quantum states. (b) In the unforgeability game, the forger who receives $n$ coins can try to successfully verify many more states than it receives, whereas a universal cloner must succeed on exactly $n+1$ states. (c) In the \emph{adaptive} unforgeability game, the forger learns the outcomes of the verifications one by one and can exploit that knowledge.  
\fi

\subsection{Our approach}
\ifab Banknote counters are often used to detect counterfeiting of cash. \fi \iffv
The current ``gold standard'' for detecting counterfeit cash bills is to use a banknote counter.  Equipped with dedicated hardware, banknote counters can verify the built-in security features (e.g., ultraviolet ink, magnetic ink, etc.) of a given cash bill. This approach, however, depends on the target currency and requires tailor made technologies. \fi \ifab However, there \else There \fi is an alternative approach. Consider the following scenario: you travel to a foreign country and withdraw some cash from an ATM. Later you execute a monetary transaction in which you receive money from an untrusted source. How will you verify the authenticity of this money?  You could compare it to the money that you withdrew from the bank's ATM\iffv, and therefore trust\fi. If it does not look the same, you would not accept it\iffv, and you might even revert the transaction\fi. We call this method \emph{comparison-based verification}\ifab . \fi \iffv, and we use the money of a foreign country as an example to emphasize the fact that \else Note that, \fi this approach works even when the specific security features of the money are not known to the verifier\ifab, as shown in the example\fi.

In this work, we propose a novel way to lift any quantum private coin scheme to a scheme which, up to some restrictions, is a public quantum coin scheme, by using an approach inspired from comparison-based verification. 
\iffv
A user can verify a coin that she receives by comparing it to the coins she already has.

In this case, we do not need the bank to run the verification to authenticate money, thus rendering the scheme \emph{public}. \fi Since the comparison is to money the user already has, it is crucial that the money states \iffv of each denomination \fi be exact copies, i.e., this approach only works for \iffv private quantum \fi coins \iffv and not for private quantum bills\fi. \iffv Technically, in the quantum scenario, the comparison is achieved by testing whether the new money and the money that we already have are in the symmetric subspace. The verifier, therefore, must have at least one original coin to validate the authenticity of new coins.
This is similar to the setting given in the example above, where we compare the money issued to us by a (trusted) bank to the new money.

\fi

\subsection{Main obstacles and our solution} \label{pg:main_obstacles_solutions}
As mentioned above, our approach to lift a private quantum coin scheme to a public scheme, is based on comparison-based verification. In the classical setting, comparison-based verification is achieved by testing whether two money bills are identical. The quantum analog of this classical approach, which is also known as SWAP test is to project on to the symmetric subspace of two registers.  As far as the authors are aware, this approach of using projective measurement into the symmetric subspace for verification, was never used in a cryptographic protocol. This is perhaps why it seems hard to construct public quantum coins using this approach: \iffv
\begin{itemize}
    \item \textbf{(Naive $0$ to $1$ cloning)} The quantum SWAP test has slightly different properties~\cite{BCW+01}, when compared to the classical approach. In particular, a state that is the tensor product of two orthogonal states, such as $\ket{0}\tensor \ket{1}$, is rejected (only) with probability $\frac{1}{2}$. As a result, an adversary without any coins can pass a single verification with probability at least $\frac{1}{2}$, see the paragraph \emph{Proof idea} on \cpageref{pg:proof_idea}, for more details.
    \item \textbf{(Sabotage attacks)} Since the verification of a new coin is done using coins from the wallet, an honest user can lose or destroy her own coins due to a transaction with an adversarial merchant.
    \item \textbf{(Refund)} Even if money is not destroyed due to verification, we need a way to recover our own coins after a failed verification.
    % \item \textbf{(Traceability attacks)} It is not clear if this approach would guarantee untraceability -- the intuition is that an adversary could change a valid coin to some other state, and later use this for tracing.
\end{itemize}
\else
 The quantum SWAP test has slightly different properties~\cite{BCW+01}, when compared to the classical approach. In particular, a state that is the tensor product of two orthogonal states, such as $\ket{0}\tensor \ket{1}$, is rejected (only) with probability $\frac{1}{2}$. As a result, an adversary without any coins can pass a single verification with probability $\frac{1}{2}$, see the  paragraph \emph{Proof idea} on \cpageref{pg:proof_idea} for more details. Since the verification of a new coin is done using coins from the wallet, an honest user can lose or destroy her own coins due to a transaction with an adversarial merchant. Even if money is not destroyed, we need a way to recover our own coins after a failed verification. 
%  Lastly, it is not clear if this approach would guarantee untraceability -- the intuition is that an adversary could change a valid coin to some other state, and later use this for tracing.
\fi

In our work, we come up with a construction that manages to get around these issues. This is done by finding weaker variants of the existing security definitions, which are still meaningful, see the  paragraph \emph{Notions of security} on \cpageref{pg:notions_of_security}, for more details. For example, we show that our scheme \emph{is} forgeable, but we manage to prove rational unforgeability, see the  paragraph \emph{Proof idea} on \cpageref{pg:proof_idea}. We provide a user-manual with a few limitations that we discuss next\iffv (also see \cref{subsec:user_manual}) \fi, which gives one way to use the money in a secure fashion. Making all these variations to the definitions, and proving them requires quite a bit of lengthy calculations; in few cases, it requires tweaking the analysis of known constructions and showing that they satisfy a related security notion needed for our purposes. The main effort is finding the right (and usually, weak) definitions to work with and the right way to stitch these together to a meaningful procedure (see the user-manual \iffv in \cref{subsec:user_manual}\fi)---the proofs are mostly straightforward and use standard linear algebraic arguments and some basic combinatorics.

 We find it interesting that despite the simplicity and how primitive our approach is, we can eventually guarantee meaningful notions of security, which are all based on very weak hardness assumptions. We leave several open questions, for which an affirmative answer could be used to lift some of the restrictions we currently have in the user-manual.

\paragraph*{User-manual} The quantum coin scheme that we construct has slightly different properties compared to a true public coin. Our money scheme uses the following user-manual. A user needs a fresh coin received directly from the bank, in order to verify each transaction she receives. However, she can spend as much money as she wants from her wallet, including the ones that she received from others. Due to this restriction, we call our construction an \emph{almost public quantum coin} scheme, which also explains the title of the paper. We elaborately discuss the user manual in \cref{subsec:user_manual}.
\paragraph*{Comparison to a true public quantum coin scheme}The only inconvenience in our scheme compared to a truly public money scheme, is that a user cannot receive more transactions than the number of fresh coins from the bank she started with. Hence, after sometime, she might have to go to the bank for refund. However, this is mostly a theoretical inconvenience. Practically, the user will never have to go to the bank since money initially withdrawn would be enough for all her transactions, for a long period of time. For example, if a user takes $100$ dollars from the bank, and suppose each public coin is worth a cent, then she essentially has $10000$ cents/public coins, and is eligible to receive $10000$ transactions. Moreover, she can spend the money that she received from others and therefore there is no limit on spending. This is practically enough for all her transactions, all the year round, and therefore, would have to go to the bank only once a year.  We can also increase the limit of the number of receivable transactions by devaluing the worth of our coins. For instance, in the previous example, we can instead declare the worth of our public coin to be a micro-dollar, and hence if a merchant starts with $100$ dollars she would now have $10^8$ public coins and would be eligible to receive $10^8$ new transactions. We show that our construction, under the user manual discussed above is rationally secure against forgeability and sabotage attacks, i.e., on average the gain of a cheating adversary and the net loss of an honest user is at most negligible. Since in real life, the users are rational entities, rational unforgeability and rational security against sabotage ensures that nobody will try to forge money or sabotage others.

Hence, in all practicality, a user using our coins scheme will never feel any difference from a truly public coin scheme. Therefore, we believe our construction is practically equivalent to a public coin scheme even though there are some theoretical differences.

\subsection{Notions of security}\label{pg:notions_of_security}
As mentioned before, we cannot achieve the standard notions of security for our construction, which forces us to consider alternate weaker notions. We consider two forms of weakening of the standard security notions and then consider a form of strengthening for our purposes. 

Firstly,  we had to shift to a weaker notion called rational unforgeability to show our construction is unforgeable. In the standard definition, we say a quantum money scheme is unforgeable if an adversary starting with $n$ money states cannot pass more than $n$ verifications (except with negligible probability).  However, in rational unforgeability, we require that the expected gain of any cheating adversary should be, at most, negligible. We show that our construction is not standard unforgeable, and hence such a relaxation is necessary. Further, we would like to emphasize that even though rational unforgeability is a weaker notion, it is still effective and meaningful in real life.
Consider a money scheme where a forger can counterfeit a money state with probability $\frac{1}{2}$, but \iffv only if she risks two of her own money states, i.e., \fi if she fails, then she has to lose her two money states. Even if she has a non-negligible probability of forging, in expectation, she loses one money state while trying to forge. Therefore, a rational forger would not try to forge and would instead stick to the protocol. Such a scheme will be deemed forgeable under standard notions, although it is secure in some sense. %In the old Byzantine sense of security such a scheme would be deemed insecure but in reality as long people are rational, the scheme can be used securely and satisfies what we define as \emph{rational unforgeability}. 
%Although the byzantine sense of unforgeability has been the most popular and extensively used notion in the literature, it seems too strong to model certain scenarios.
%In practice, the users are more rational, hence instead of viewing parties submitting money states as adversaries, we can see them as rational parties who want to maximize their own benefits. Therefore we need to take into account the penalty that the user might have to pay in case it tries to forge and fails. 
This leads to the definition of rational \iffv unforgeability (see \cref{definition:rational_unforgeability}), \else unforgeability, \fi where we only require that \emph{in expectation} any forger can have at best negligible advantage due to forging. Our construction is rational \iffv unforgeable---and we argue that this provides a meaningful notion of security. \else unforgeable. \fi

Secondly, we also deviate from the usual notion of defining the gain of the adversary. Usually, the gain or utility of the adversary is defined as the difference between the number of times she passed verification and the number of coins she took from the bank. We call this the \emph{flexible utility}. However, we do not know if our construction is rationally unforgeable under this definition, and hence, we switch to a restricted definition called \emph{all-or-nothing utility}, 
% of adversarial utility in the nonadaptive model. Hence, we consider a more restricted notion of adversarial utility,
wherein 
% if even one coin submitted by the adversary fails verification, we reject all her coins, i.e., 
we either count all or none of her coins. The motivation for all-or-nothing utility comes from scenarios in which only one kind of coins are used for a particular item. Suppose a person goes to buy a TV from an honest seller but is allowed to buy only one TV. He puts all the money on the table according to the worth of the TV she plans to buy. The seller either approves the transaction and gives a TV or rejects and simply says no to the user but does not return the money back to the buyer. Even if one of the money states fail verification, the seller does not approve the transaction. 
We show that our construction is \emph{all-or-nothing rationally unforgeable}. An elaborate discussion on these different notions of utility is given in \cref{subsec:unforgeability and security}.

Next, we consider a strengthening of the security notion by considering multiple users or verifiers that the adversary can attack simultaneously. We call this \emph{ multiverifier all-or-nothing rational unforgeability}, which is still \emph{nonadaptive} in the sense that post-verified money is never returned to the adversary, but it manages to encapsulate the threat model in the user manual (given in \cref{subsec:user_manual}), thus providing a meaningful notion of security to the users. A detailed discussion on the equivalence of the multiverifier threat model and that under the user manual is provided in \cref{appendix:user_manual-multi-ver-equivalence}.
We prove multiverifier all-or-nothing rational unforgeability for our construction, see \Cref{thm:unconditional_secure}. The proof is quite long and tedious, and therefore, we first discuss a weaker notion called \emph{single-verifier all-or-nothing rational unforgeability}, which has a relatively elegant proof for our construction. Even though the single-verifier all-or-nothing rational unforgeability does not capture the most generic threat model as per the usual manual, it is an interesting simplification of the multiverifier unforgeability notion.

%and in fact, satisfies a stronger unforgeability notion, which is exhibited by the user manual.
%We discuss both adaptive and non-adaptive threat models (see \cref{subsec:unforgeability and security}) as well as a multiverifier version of the nonadaptive threat model that captures all kinds of attack in the user manual. 
The main security guarantees for our construction in the different threat models are listed down in \cref{table:unforgeability_results}.
\renewcommand\arraystretch{1.5}

\begin{table}[htp]
\centering
\begin{adjustbox}{max width=\textwidth}
\begin{tabular}{llll}
\toprule                                          %& Utility 
& Standard Unforgeability                      & Rational Unforgeability                                                                                              \\
\midrule {Single-verifier}             %& Flexible       & No (\cref{alg:opt_attack}) & Unknown                                                                                                              \\
                                                                 %& All-or-nothing 
                                                                 & No (\cref{alg:opt_attack}) & Yes (\cref{thm:unconditional_secure})                                                              \\
%\hline \multirow{2}{*}{Adaptive}                  & Flexible       & No (\cref{alg:opt_attack}) & No (\cref{remark:adaptive_Attack})                                                                 \\
                                                                 %& All-or-nothing & No (\cref{alg:opt_attack}) & Yes (\cref{remark:adapt-multiver_nonadapt} and \cref{thm:multi-unconditional_secure})  \\
\hline {Multiverifier} %& Flexible       & No (\cref{alg:opt_attack}) & Unknown                                                                                                              \\
                                                                 %& All-or-nothing 
                                                                 & No (\cref{alg:opt_attack}) & Yes (\cref{thm:multi-unconditional_secure})\\ \bottomrule   
\end{tabular}
\end{adjustbox}
\caption{Unforgeability of our construction}
\label{table:unforgeability_results}
\end{table}

\iffv

 Our construction also achieves another security notion other than unforgeability called the security against sabotage (a notion which was first defined in in~\cite{BS16a}), which essentially ensures that no adversary can sabotage or harm honest users. We only discuss rational security against sabotage because the standard notion does not hold for our construction by the same attack (\cref{alg:opt_attack}) as in the case of unforgeability. All results regarding security against sabotage are given in \cref{table:fairness_results}.

 \begin{table}[htp]
\centering
\begin{adjustbox}{max width=\textwidth}
\begin{tabular}{llll}
\toprule                                          %& Loss        
& Private sabotage                      & Public sabotage                                                                                              \\
\midrule {Single-verifier}             %& Flexible       & Unknown & Unknown                                                                                                              \\
                                                                 %& All-or-nothing 
                                                                 & Yes (\cref{thm:rational priv fair}) & Yes (\cref{thm:rational public fair})                                                              \\
%\hline \multirow{2}{*}{Adaptive}                  & Flexible       & No (\cref{remark:flex-adapt-priv-sabotage}) & Unknown                                                                 \\
                                                                 %& All-or-nothing & Yes (\cref{remark:all-or-nothing-adpt-to-multi-sec_against_sabotage} and \cref{cor:multi_ver-rational-fair} ) & Yes (\cref{remark:all-or-nothing-adpt-to-multi-sec_against_sabotage} and \cref{remark:single-multi-rational-fair-public})  \\
\hline {Multiverifier} %& Flexible       & Unknown & Unknown                                                                                                              \\
                                                                 %& All-or-nothing 
                                                                 & Yes (\cref{prop:sing-multi-rational-fair} and \cref{thm:rational priv fair}) & Yes (\cref{remark:single-multi-rational-fair-public} and \cref{thm:rational public fair})\\ \bottomrule   
\end{tabular}
\end{adjustbox}
\caption{Rational security against sabotage of our construction}
\label{table:fairness_results}
\end{table} 
We discuss security against sabotage in detail, including the definition of private and public sabotage in \cref{appendix:fairness}.

We are not aware of rational notions of unforgeability in cryptography. More broadly, there is a line of work that is called rational cryptography~\cite{GKM+13,ACH11,PH10,FKN10,KN08,ADG+06,HT04}, which is quite different than ours: it discusses protocols consisting of multiple competing and potentially corrupt parties who might deviate from the protocol and use different strategies in order to maximize their gain.
It should be noted that our work is different from the notion of rational cryptography. In the rational cryptography setting, the analysis is about the equilibrium arising out of the multiple competing parties, whereas here, the problem is to optimize the strategy of the (single) adversary trying to maximize her gain against an honest bank. Our work is somewhat closer to the notion of rational prover discussed in~\cite{MN18}. In~\cite{MN18}, the authors discuss quantum delegation in the setting where the verifier also gives a reward to the prover, the value of which depends on how she cooperated. Hence, the prover aims to maximize her expected reward rather than cheat the honest verifier. The authors show that for a particular reward function they constructed, the optimal strategy for the prover is to cooperate honestly.
 %because instead of having multiple adversarial players trying to maximize their payoff, here we have a single adversarial money user trying to maximize her gain against an honest bank.
\fi 

  %\ifab The related results are proved later in the appendix but our discussion regarding the user manual captures how strong these threat models are. \fi  

%Usually a quantum money scheme is analyzed only on the basis of unforgeability but ideally we should also consider the sabotage attacks (see also~\cite{BS16a}) intended to hurt honest users (note that sabotage attacks are not ruled out by unforgeability). Since we are discussing a coin scheme, there is also the issue of untraceability as defined in~\cite{AMR20}. In our work, the notion of security against sabotage have been formalized in detail and our construction achieves all the above mentioned security notions under some restrictions that are discussed more in the drawbacks.

\subsection{Main result}\label{pg:main-results}
Our main contribution is the lifting result, which can lift any private coin scheme to an almost public coin scheme---see \cref{prop: multi_ver-unforge}. By lifting the result in Ref.~\cite{JLS18}\anote{ (\iffv or the \fi simplified version in~\cite{BS19})}, we get the following result:
\begin{theorem}[Informal Main Result]
Assuming \prs exists, there is a public quantum coin scheme that is rationally secure against multi-verifier forgery attacks.
\label{thm:informal_main}
\end{theorem}
Similarly, by lifting the results in Ref.~\cite{MS10}, we show an \emph{inefficient} scheme with the same properties as in \cref{thm:informal_main} above, which is secure even against \emph{computationally unbounded} adversaries. In~\cite{AMR20}, the authors construct a stateful yet efficient version of the money scheme in~\cite{MS10}, using a stateful mechanism to approximately simulate truly random states. Hence, by lifting the results in Ref.~\cite{AMR20}, we show a stateful public coin scheme that satisfies all the security notions mentioned in~\cref{thm:informal_main} unconditionally. \iffv The formal result is provided in \cref{thm:multi-unconditional_secure}. \fi
 
%Consider a money scheme where one can forge a money state with overwhelming probability but only if she risks her own money that she looses out in the advent of a failure. Hence, in expectation she might actually lose her money trying to forge. Therefore, a rational user would not try to forge and stick to the protocol. In the old Byzantine sense of security such a scheme would be deemed unusable but in reality as long people are rational, the scheme can be used securely and satisfies rational unforgeability. 
 \subsection{Technical overview} \label{pg:proof_idea}As discussed above, our construction is an analogical extension of the classical \emph{comparison-based verification}. \iffv So the first attempt to lift a private quantum coin scheme is that we give users a private coin $\ket{\cent}$ as a public coin, such as the private coin construction by Ji, Liu and Song~\cite{JLS18}. In order to verify a given coin, the verifier should compare it with a valid coin from her wallet, i.e., perform symmetric subspace verification on the two registers (the wallet coin and the new coin)\footnote{This special case when the number of registers is $2$ is also known as the swap-test.}. \fi \ifab %The quantum analogue of testing two strings are equal is the SWAP test.
 Hence, the first natural attempt to lift a private quantum coin scheme such as~\cite{JLS18}, is as follows. Use the private coin $\ket{\cent}$, itself as a public coin; to verify an alleged coin, compare it to a fresh coin from the wallet, by doing a symmetric subspace verification on the  two registers\footnote{This special case when the number of registers is $2$ is also known as the swap-test.}. \fi  Unfortunately, the scheme that we just described is rationally forgeable. For example, suppose an adversary $\adv$ who does not have a valid coin to start with, submits a coin $\ket{\psi}$ to the verifier. Since the private quantum scheme is secure, \iffv $\ket{\psi}$ and $\ket{\cent}$ would be almost orthogonal, i.e., \fi $\ket{\psi} \approx \ket{\cent^{\perp}}$ for some $\ket{\cent^{\perp}}$ in the orthogonal space of $\ket{\cent}$. %Hence, if we think of $\ket{\cent} = \ket{0}$ then $\ket{\psi}$ is like $\ket{1}$.  
The combined state of the registers $\ket{\cent}\tensor\ket{\cent^{\perp}}$ pass the symmetric subspace verification with probability $\frac{1}{2}$. So, the probability of successful forging (and the \emph{expected utility} of the forger) is $\frac{1}{2}$. \iffv Hence, this scheme is forgeable.  \fi

%To reduce this error, we use an amplification technique. 
In order to bypass this problem, we will use a form of amplification. \iffv A public coin will instead consist of multiple private coins, suppose $\kappa$ many. The private coins cannot be transacted discretely of their own but only in sets of $\kappa$ as a public coin; just like \emph{cents} (denoted by $\cent$) and \emph{mills} (denoted by $\mill$) in the current world---cent coins, each of which is equivalent to ten mills, are used in transactions, but mill coins are not used even though as a unit they exist. Hence, we use $\ket{\cent}$ to denote a public coin and $\ket{\mill}$ to denote a private coin in order to show the resemblance to the analogy of cents and mills. 
Therefore, we define a public coin as $\ket{\cent}:=\ket{\mill}^{\tensor \kappa}$, a collection of $\kappa$ private coins ($\ket{\mill}$). \fi \ifab We define a public coin $\ket{\cent}$, as a collection of $\kappa$ private coins ($\ket{\mill}$), which cannot be transacted otherwise, as an individual currency. \fi \iffv Again, to check the authenticity of the coin, the verifier uses the valid public coin in her wallet (which,  for the sake of simplicity, contains only one public coin) \fi \ifab To verify an alleged coin, take a  fresh coin from the wallet, \fi and perform symmetric subspace verification on all the $2\kappa$ registers. In this setting, if an adversary $\adv$ with $0$ public coins, produces an alleged public coin $\ket{\psi}$, then since the private scheme is unforgeable, none of the $\kappa$ registers should pass the verification of the private scheme. Hence, $\ket{\psi}$ must have a large overlap with the span of the states, that can be written as a tensor product of states orthogonal to $\ket{\mill}$\iffv (since every state in the orthogonal space of the subspace mentioned above, will pass the private verification of at least one of its $\kappa$ registers)\fi. %(If $\ket{\psi}$ has a large enough overlap with any basis state other than the state $\ket{\mill^{\perp}}^{\tensor \kappa}$, then the probability that at least one of the $\kappa$ registers of $\ket{\psi}$ will pass the private verification is large. This probability cannot be large because the private coin scheme is unforgeable.) 
The combined state of the registers hence, \iffv is 
\[\ket{\psi}\tensor\ket{\cent}\approx (\ket{\mill_1^{\perp}}\tensor \ldots \ket{\mill_{\kappa}^\perp})\tensor \ket{\mill}^{\tensor \kappa},\] which \fi has a squared overlap\iffv\footnote{squared overlap means the modulus square of the projection in to the subspace.}\fi of $\frac{1}{\binom{2\kappa}{\kappa}}$ with the symmetric subspace of the $2\kappa$ registers. \iffv The term $\frac{1}{\binom{2\kappa}{\kappa}}$ is bounded above by $< \left(\frac{1}{2}\right)^\kappa$. \fi Therefore,  \iffv by choosing $\kappa = \lambda$ or even $\kappa=\log^\alpha(\lambda)$ for $\alpha>1$, \fi $\adv$'s forging probability in this attack, is \iffv $\negl$ (where $\secpar$ is the security parameter). \else $\negl[\kappa]$. \fi  \iffv A similar use of symmetric subspace operations was used in a recent work on private quantum coins~\cite{JLS18}. \fi \ifab Unfortunately, if an adversary having $n$ coins tries to pass $(n+1)$ coins, the optimal success probability becomes inverse polynomially close to $1$, for $n$ large enough. \fi  \iffv Unfortunately, when more number of coins are submitted, suppose an adversary having $n$ coins tries to pass $(n+1)$ coins, the optimal success probability becomes $\frac{\binom{(n+1)\kappa}{n\kappa}}{\binom{(n+2)\kappa}{(n+1)\kappa}}$, which is inverse polynomially close to $1$, for $n$ large enough.   
Hence, it is hard for any adversary to produce two coins from one coins but it is easy to produce $(n+1)$ from $n$ coins. \fi
However, a simple examination shows that \iffv when taking into account that the money is lost in case of a failed verification, \fi the expected utility of any \iffv polynomial-time \fi adversary is \iffv negligible if not negative. \else negligible. \fi

%Therefore, in expectation she can only have negligible gain by forging and hence a rational user won't try to forge.

\iffv  %Therefore, she needs to risk her $n$ coins to forge one coin because if the verification fails the honest verifier will not return the coins as a part of the protocol. Hence, in expectation the adversary can only have negligible gain and therefore rational unforgeable.
This is what motivated us to define rational unforgeability, where we want the expected gain of any adversary should be at best negligible. 
\fi
%Simply using their technique to lift a secure private quantum coin scheme to a public coin scheme gives only an inverse polynomial advantage against forging. We use a slightly different approach (similar to that described above) to obtain an  inverse exponential advantage.

\ifab \vspace{-10pt} \fi %\paragraph{Notions of unforgeability} One of our main contributions in this work are the formal definitions of the different kinds of forging attack models that are relevant to the quantum money setting. %
%Primarily, we have a nonadaptive model wherein the forger obtains polynomially many copies of the quantum money state, executes an algorithm and submits some coins, but all of the states at once. The forger succeeds if a greater number of the quantum money states it submits pass verification than the number of quantum states it actually obtained from the bank. In contrast, in an adaptive model, the forger has oracle access to different verifications that  it can query adaptively. In a public quantum money scheme with a public quantum key, the protocol is stateful and not stateless unlike the quantum money schemes that use a classical key. Moreover, because the  verification procedures of comparison-based verification schemes involve the user's own money, it could be stolen. To address these issues, we need more subtle notions. We also need to distinguish between the different kinds of attack models, even in the adaptive setting, depending on the different kinds of verification oracles.
%The differences between the models are evident in the literature for private quantum schemes in the single-verifier setting (e.g., different kinds of adaptive attacks~\cite{lut10, NSBU16, Aar09} on Wiesner's scheme~\cite{Wie83}), but they were never formalized to this extent. 
\iffv
\subsection{Comparison to previous works}
 In comparison to other public quantum money constructions, our scheme has two main advantages. 
Firstly, we propose a construction (described in \cref{alg:ts}) of an almost public quantum coin scheme, which is \emph{rationally secure} (see \cref{thm:unconditional_secure}) based on \prs, which is a standard assumption even weaker than one-way functions. On the contrary, as discussed in the \emph{Related works} paragraph on \cpageref{pg:related-works}, all known public money constructions known so far are either not provably secure or are based on non-standard assumptions. Moreover, \cite{AHY23} gives evidence that some of the standard and generic assumptions, such as collision-resistant hash functions, are not enough to construct public quantum money.  Given the proven difficulties and failed attempts in constructing public quantum money, we investigate the possibility of constructing a weaker yet meaningful form of public quantum money from possibly weaker assumptions. In that vein, our work constructs an intermediate primitive between public and private quantum money schemes but from \prs, a weaker assumption than not only collision-resistant hash but even one-way functions, see Ref~\cite{Kre21}, thus complementing the impossibility result of~\cite{AHY23}. %\anote{ Should I add this considering it is a subsequent work? : One exception to this is \cite{RZ21}, which constructs franchised quantum money based on one-way functions. We believe~\cite{RZ21} is in the same direction as our work, i.e., constructing a meaningful intermediate notion between public and private quantum money but from weaker and more standard assumptions, but achieves incomparable goals relative to each other. It should be noted that a recent work~\cite{CM24} shows that digital signatures, which is the underlying primitive in~\cite{RZ21}, is either a strictly stronger or an incomparable primitive compared to \prs, the underlying primitive in our construction. }

%Next unlike other construction,   
%In fact, some of these non-standard assumptions do not have candidate constructions, for example - quantum secure indistinguishability obfuscation.

Next, in all the efficient public quantum money schemes that the authors are aware of, the number of qubits required in a money state and the running time of all the algorithms (keygen, mint, and verify) are polynomial in the security parameter. 
However, in our construction, by choosing the underlying constructions carefully, we can make the time and space complexity of our constructions to be polylogarithmic in the security parameter (see \cref{subsec:complexity}). 
\fi

The main drawback of our construction (see \cref{alg:ts}) is that we could only prove a weaker notion of security in the following sense. Firstly, as mentioned earlier, we consider rational unforgeability, which is weaker than standard unforgeability. Secondly, we consider a non-standard notion of utility or adversarial gain in the unforgeability game, called the all-or-nothing utility, which is weaker than the standard notion of utility called the flexible utility (see Notions of security on \cpageref{pg:notions_of_security}). Moreover, we consider a non-adaptive model, where a post-verified money state is never returned back to the adversary. It is natural to ask whether all these restrictions can be lifted. 
Unfortunately, at least some of the restrictions are necessary. For instance, we show that standard unforgeability does not hold for our construction, even if we consider nonadaptive adversaries along with the all-or-nothing utility, see \cref{alg:opt_attack}. Moreover, if we consider the flexible utility instead and allow the adversaries multiple adaptive attempts or queries to the verification oracle, then our construction is not rationally unforgeable (see the discussion after \cref{definition:nonadapt_all-or-nothing_unforge} and \cref{remark:adaptive_Attack}). Of course, this does not rule out other options of potential strengthening, such as considering rational unforgeability under flexible utility and nonadaptive adversaries\footnote{We do not know the answer both in the single-verifier and the multiverifier cases.},  the answers to which are unknown. We enforce the nonadaptive threat model by publishing a user manual for the users that instructs the user to use a fresh coin from the bank to perform a comparison-based verification of a received transaction---see \cref{subsec:user_manual}.

\ifab \vspace{-1pt} \fi \subsection{Scientific contribution} \label{pg:sceintific_constributions}
\begin{enumerate}%\ifab \vspace{-10pt} \fi
    \item \textbf{(Weak and generic \iffv computational \fi assumptions)} \iffv As far as the authors are aware, the \else Our \fi construction comes closest to a provably secure quantum public money based on \prs, which is a generic hardness assumption\iffv, weaker than quantum-secure one-way functions. \else, with polylogarithmic time and space complexities. \fi There is no known construction of public quantum money from standard assumptions, and it seems unlikely~\cite{AHY23} to have a public quantum money construction that is provably secure based on generic standard hardness assumptions such as one-way functions or collision-resistant hash functions.   \iffv Moreover, our construction is comparatively more efficient because unlike most efficient money schemes, both running time for all algorithms and the number of qubits required for each coin is only polylogarithmic (see \cref{subsec:complexity} and the discussion in the last paragraph in \cref{subsec:user_manual}) instead of being polynomial in the security parameter. 
    \fi 
    \item \textbf{(Almost public coin scheme)} Our construction comes close to satisfy the features of both public verifications and coins. There is currently no other public coin scheme.
    \item \textbf{(Quantum public key)} Our construction closely resembles a \emph{public quantum money scheme with a quantum public key},  a topic  that has not been studied. By itself, such a scheme is quite interesting as it may evade the known impossibility result \iffv(see \cref{quote: impossibility}), \fi that unconditionally secure public quantum money schemes (with a classical public key), cannot exist. In fact, we managed to partially circumvent the impossibility \iffv result  (see Appendix \cref{thm:unconditional_secure}), \else result, \fi  by constructing an \ifab unconditionally secure \fi inefficient almost public quantum money scheme, based on a previous work~\cite{MS10}\iffv. \fi\ifab, thus getting closer to answering the open question: Can unconditionally secure public quantum money, with a quantum public key exist? \fi  %by proving the existence of an inefficient but $\nruuf$ public quantum scheme with comparison-based verification based on a prior work~\cite{MS10}\iffv -- `see \cref{cor:unconditional} \fi. 
 \iffv This brings us closer to answering the open question: Can unconditionally secure public quantum money, with a quantum public key exist? \fi  %We managed to construct an inefficient almost public quantum money scheme which is unconditionally secure. Thus we come close to
\item \textbf{(Rational unforgeability)} We also put forward a new notion of rational unforgeability, 
which is weaker than the usual notion of security but still has strong guarantees. \iffv This might open up new possibilities for new constructions, or as a way to circumvent impossibility results. This notion is relevant in most of the cases as in reality users are rational parties rather than adversarial madmen. The authors are not aware of any such notion in the context of  quantum money. \fi  
\item \textbf{(Modularity)} The lifting technique used in our work \iffv lifts any private quantum coin to an almost public quantum coin, preserving the main notions of security against forging and sabotage. 
Our techniques are \else is \fi fairly general, and we hope that \iffv they \else it \fi  could be used to lift other cryptographic protocols such as quantum copy-protection. We discuss this in detail in the future work section\ifab. \fi\iffv, see \cref{sec:future_work}.\fi
\end{enumerate}

\iffv 

\paragraph*{Paper organization}
\onote{Reorder stuff, and explain that the main text contains only the single-verifier unforgeability, which is a weaker notion of security; the stronger result regarding multi-verifier security, that encapsulates the real threat model under the user manual (which is currently not mentioned), is given in the appendix.}\anote{Add the equivalence result.}
\cref{sec:prelim-definitions} contains notations used in the proofs of our results (\cref{subsec:notations}), preliminaries (\cref{subsec:prelims}) and definitions (\cref{subsec:definitions}). \cref{subsec:unforgeability and security} contains the single-verifier unforgeability definitions, which are a weaker simplification of the multiverifier unforgeability notions given in \cref{appendix:multi-ver-unforge} that encapsulates the real threat model under the user manual. In \cref{sec:constructions}, we describe our construction (\cref{alg:ts}) and state our main result \cref{thm:multi-unconditional_secure}. 
%\cref{subsec:notations} contain the notations used in the proofs of our results which are given in the appendix.
 We discuss the complexity and possible implementation of our construction in \cref{subsec:complexity}, followed by the user manual to use our construction (\cref{subsec:user_manual}). Next, in \cref{sec:attack}, we discuss optimal forgery attack (\cref{subsec:attack_description}, \cref{subsec:optimality} and \cref{subsec:analysis_of_attack}) on our construction, and use it to prove the lifting theorem (\cref{prop: rational-unforge}) in \cref{subsec:completeness_security proofs}. 
In \cref{sec:future_work}, we discuss a few open questions relevant to our work and the scope of future work in order to further improve our construction.
In \cref{appendix:fairness}, we discuss the security against sabotage attacks for our scheme. \cref{appendix:untraceability} contains a discussion on the untraceability definitions for quantum money. \cref{appendix:multi-ver-attacks} contains the definitions of multiverifier unforgeability (\cref{appendix:multi-ver-unforge}), and multiverifier security against sabotage (\cref{appendix:multi-ver-fair}), the security results for our construction in the multiverifier setting (\cref{appendix:multi-ver-results}), and a discussion on how the multiverifier security notions capture all possible efficient attacks under the user manual (\cref{appendix:user_manual-multi-ver-equivalence}). The proofs of security results, other than unforgeability, are given in \cref{appendix:appendix_proofs}, which also contains the proof of completeness for our construction. Finally, nomenclature is given on the last page.
%%%%%%%%%%%%%%%%%%%%%%%%%%%%%%%%%%%%

Next, we show a pictorial representation of how we achieve the different notions of security for our scheme in \cref{fig:struct_fairness,fig:struct_unforgeability}. %These results are also listed in \cref{table:unforgeability_results} and \cref{table:fairness_results}.

%\ifnum\tdgm=1
\begin{figure}[htp]
\centering 
\begin{adjustbox}{max width=1.2\textwidth,center}
\begin{tikzpicture}%[scale=0.5, align = left]
  \matrix (m) [ row sep=2em, column sep=3em, align=center]
  {& \node[draw, line width=1mm, square](1){$\begin{matrix}\text{Inefficient or stateful}\\ \text{multiverifier nonadaptive}\\
   \text{unforgeable private coin scheme}\\ \text{with rank-$1$ verification}\\ \text{\cite{MS10,AMR20}}\end{matrix}$}; & \node[draw, line width=1mm, rectangle](2){$\begin{matrix}\text{Inefficient or stateful nonadaptive}\\ \text{unforgeable private coin scheme}\\ \text{with rank-$1$ verification}\end{matrix}$};\\
   & \node[draw, ellipse](3){$\begin{matrix}\text{Quantum secure}\\ \text{Pseudorandom States exist}\end{matrix}$}; &\\ 
   & \node[draw, rectangle, line width=1mm](4){$\begin{matrix}\text{Multiverifier nonadaptive unforgeable private}\\ \text{coin scheme with rank $1$ verification}\end{matrix}$}; &\\ \node[draw, rectangle](5){$\begin{matrix}\text{Multiverifier nonadaptive}\\ \text{rational unforgeability}\end{matrix}$}; & \node[draw, rectangle, line width=1mm](7){$\begin{matrix}\text{Nonadaptive unforgeable private}\\ \text{coin scheme with rank $1$ verification}\end{matrix}$}; & \\
     \node[draw, diamond](6){User manual}; &  & %\\
   %\node[draw, rectangle](8){$\begin{matrix}\text{Adaptive rational unforgeability}\end{matrix}$};
   %& & 
   \node[draw, rectangle](9){$\begin{matrix}\text{Nonadaptive}\\\text{rational unforgeability}\end{matrix}$};\\};
   \path[->,draw,thick]
   (1) edge [] node[labeled]  [scale=1] 
   {\text{$\begin{matrix}\text{Follows}\\[.001mm]\text{immediately}\end{matrix}$}} (2) ;
   \path[->,draw,dashed]
   (2) edge [above right] node[labeled]  [scale=1] 
   {\text{\cref{prop: rational-unforge}}} (9) ;
   \path[->,draw,dashed]
   (3) edge [right] node[labeled]  [scale=1] 
   {\text{\cite{JLS18}}} (4) ;
   \path[->,draw,dashed]
   (4) edge [left] node[labeled]  [scale=1] 
   {\text{\cref{prop: multi_ver-unforge}}} (5) ;
   \path[->,draw,dashed]
   (1) edge [left, bend right=22] node[labeled]  [scale=1] 
   {\text{\cref{prop: multi_ver-unforge}}} (5) ;
   \path[->,draw,thick]
   (4) edge [right] node[labeled]  [scale=1] 
   {\text{Follows immediately}} (7) ;
   % \path[->,draw,thick]
   % (5) edge [left, bend right=60] node[labeled]  [scale=1] 
   % {\text{\cref{remark:adapt-multiver_nonadapt}}} (8) ;
   \path[-,draw,thick]
   (5) edge [left] node[labeled]  [scale=1] 
   {\cref{appendix:user_manual-multi-ver-equivalence}} (6) ;
   \path[->,draw,dashed]
   (7) edge [right] node[labeled]  [scale=1] 
   {\text{\cref{prop: rational-unforge}}} (9) ;
\end{tikzpicture}
\end{adjustbox} 
\vspace{-10pt}
\caption{The above figure contains the structural diagram for the  unforgeability properties. The figure contains nodes which are either squares, rectangles, diamonds or ellipses. It also contains arrows which connect two nodes. A rectangle represents a security definition and a square represents the existence of a security definition, which holds unconditionally. It also contains the reference of the proof. An ellipse represents the  assumption on the existence of a security notion. Arrows are of two types: solid arrow and dashed arrow. A solid arrow connecting two nodes A and B means if A holds for some construction then B also holds for that construction. A dashed arrow connecting two nodes A and B means if A exists, then B exists via some construction. Most arrows contain a reference to the proof. However, there are two arrows which follow immediately, and hence, we do not have references for them. The rectangles or squares with thick outline contain security notions regarding a private coin scheme whereas the ones with thin outline contain security notions regarding our construction, which is a comparison-based public quantum coin scheme with private verification. All the results for our construction hold with respect to the all-or-nothing utility, and for the private scheme, we use the flexible utility. There is also the user manual depicted in a diamond box, which points to the security notion that ensures the user-manual is secure. The notions of security, including non-adaptive security which was not presented in the introduction, are discussed in \cref{subsec:unforgeability and security,appendix:multi-ver-unforge}.}
\label{fig:struct_unforgeability}
\end{figure}

\begin{figure}[htp]
\centering
%\begin{adjustbox}{max width=0.6 \textwidth,center}
\begin{tikzpicture}%[scale=0.5, align = left]
  \matrix (m) [ row sep=2em, column sep=0em, align=center]
  {& \node[draw, square,  line width=1mm](1){$\begin{matrix}\text{Private quantum}\\ \text{coin scheme with}\\ \text{with rank-$1$ verification}\\[.5 cm] \text{\cite{JLS18,MS10,AMR20}}\end{matrix}$}; &\\
  & &\\
  \node[draw, rectangle](2){$\begin{matrix}\text{Nonadaptive rational security}\\\text{against private sabotage}\end{matrix}$}; & & \node[draw, rectangle](3){$\begin{matrix}\text{Nonadaptive rational security}\\\text{against public sabotage}\end{matrix}$};\\
  & &\\
  \node[draw, rectangle](4){$\begin{matrix}\text{Multiverifier nonadaptive}\\\text{rational security}\\\text{against private sabotage}\end{matrix}$}; & & \node[draw, rectangle](5){$\begin{matrix}\text{Multiverifier nonadaptive}\\\text{rational security}\\\text{against public sabotage}\end{matrix}$};\\
  & \node[draw, diamond](6){User Manual}; &\\
  % \node[draw, rectangle](7){$\begin{matrix}\text{Adaptive rational security}\\\text{against private sabotage}\end{matrix}$}; & & \node[draw, rectangle](8){$\begin{matrix}\text{Adaptive rational security}\\\text{against public sabotage}\end{matrix}$};\\
  };
  
  \path[->,draw,dashed]
   (1) edge [above left] node[labeled]  [scale=1] 
   {\cref{thm:rational priv fair}
    } (2) ;
  \path[->,draw,dashed]
   (1) edge [above right] node[labeled]  [scale=1] 
   {\cref{thm:rational public fair}
    } (3) ;
    \path[->,draw,thick]
   (2) edge [left] node[labeled]  [scale=1] 
   {\cref{prop:sing-multi-rational-fair}
    } (4) ;
    \path[->,draw,thick]
   (3) edge [right] node[labeled]  [scale=1] 
   {\cref{remark:single-multi-rational-fair-public}
    } (5) ;
    \path[-,draw,thick]
   (4) edge [left] node[labeled]  [scale=1] 
   {\cref{appendix:user_manual-multi-ver-equivalence}} (6) ;
   \path[-,draw,thick]
   (5) edge [right] node[labeled]  [scale=1] 
   {\cref{appendix:user_manual-multi-ver-equivalence}} (6) ;
   %  \path[->,draw,thick]
   % (4) edge [left] node[labeled]  [scale=1] 
   % {\cref{remark:all-or-nothing-adpt-to-multi-sec_against_sabotage}
   %  } (7) ;
   %  \path[->,draw,thick]
   % (5) edge [right] node[labeled]  [scale=1] 
   % {\cref{remark:all-or-nothing-adpt-to-multi-sec_against_sabotage}
   %  } (8) ;
\end{tikzpicture}
\caption{The above figure represents the structural diagram for security against sabotage. The notations are the same as that for the previous figure. All the results for our construction hold with respect to the all-or-nothing loss function. We would like to point out that the proof of \cref{prop: multi_ver-unforge} used in \cref{fig:struct_unforgeability} uses the rectangle containing the security notion \emph{multiverifier nonadaptive rational security against private sabotage} in this figure. 
The notions of security are discussed in \cref{appendix:fairness,appendix:multi-ver-fair}.
%Since the security notion holds for our construction, the arrows using \cref{prop: multi_ver-unforge} also hold true in the previous figure.
}
\label{fig:struct_fairness}
\end{figure}

\section{Preliminaries and definitions}\label{sec:prelim-definitions}
%%%%%%%%%%%%%%%%%%%%%%%%%%%%%%%%%%%%
\subsection{Notations}\label{subsec:notations}
%In this section we describe a nonadaptive forging attack on the scheme $\pkqc$ which turns out to be an optimal forgery.
This subsection contains some notations and conventions that will be required only in the proofs (\cref{sec:attack,subsec:completeness_security proofs,appendix:appendix_proofs}).
 
\begin{enumerate}
    \item We use \nom{H}{$\mathbb{H}$}{$\CC^d$} to represent $\CC^d$. We fix the local dimension of each register to $d$, i.e., the state of each register is a unit vector (or an ensemble of unit vectors) in $\mathbb{H}$. 
    \item \label{item:symmetric_sub_proj_def} We use \nom{S}{$\Sym{n}$}{the symmetric subspace of $\mathbb{H}$ over $n$ registers} to denote the symmetric subspace of $\mathbb{H}$ over $n$ registers. Let \nom{P}{$\Pi_{\Sym{n}}$}{projection operator on to $\Sym{n}$} to be the projection on to $\Sym{n}$. The symmetric subspace over $n$ registers is the set of all states on $n$ register which remain the same under any permutation of the registers. For more information on the symmetric subspace, see~\cite{Har13}. %Since we fix the local dimension to $d$ for the entire manuscript, for notational convenience we will use \nom{RLETTER}{$Sym^{n}$}{same as $\Sym{n}$, the symmetric subspace of $\mathbb{H}$ over $n$ registers} instead of $\Sym{n}$. 
    %It is known that the measurement $\{\Pi_{\Sym{n}}, (I - \{\Pi_{\Sym{n}})\}$ can be efficiently implemented. For more information on symmetric subspace, see also Ref.~\cite{Har13}. 
    \item For every vector ${\vec{j}} = (j_0,j_1,\ldots j_{d-1}) \in \NN^d$ such that $\sum_{r=0}^{d-1} j_r = n$, we denote \nom{M}{$\binom{n}{{\vec{j}}}$}{$\binom{n}{j_0,j_1,\ldots j_{d-1}}$} as $\binom{n}{j_0,j_1,\ldots j_{d-1}}$.
    \item Let \nom{I}{$\mathcal{I}_{d, n}$}{Set of all vectors in $(j_0,j_1,\ldots ,j_{d-1})$ such that $\sum_{k=0}^{d-1} j_k = n$} be defined as $\mathcal{I}_{d, n} := \{(j_0,j_1,\ldots ,j_{d-1}) \in \NN^d~|~\sum_{k=0}^{d-1} j_k = n \}$ - see also Ref.~\cite[Notations]{Har13}.\label{item:I-def}
    %\nom{RLETTER}{$\mathcal I_{d,n}$}{$\{(j_0,j_1,\ldots ,j_{d-1})~|~\sum_{k \in \ZZ_d} j_k = n \}$}
    \item For every vector ${\vec{i}}= (i_1,\ldots,i_n)$ in $\ZZ_d^{n}$ we define \nom{T}{$T(\vec i) $}{vector in $\mathcal{I}_{d,n}$ whose $k^{th}$ entry (for $k \in \ZZ_d$) is the number of times $k$ appears in the vector $\vec{i}$} to be the vector in $\mathcal{I}_{d,n}$ whose $k^{th}$ entry (for $k \in \ZZ_d$) is the number of times $k$ appears in the vector $\vec{i}$---see also Ref.~\cite[discussion before Theorem 3.]{Har13}. Note that, $|T^{-1}(\vec{j})| = \binom{n}{\vec{j}}$. We shall represent the $k^{th}$ entry (for $k \in \ZZ_d$) of $T(\vec{i})$ by $(T(\vec{i}))_k$.\label{item:T-def}
    \item \label{item:product_basis_definition} We extend $\ket{\mill}$ (a private coin) to an orthonormal basis of $\mathbb{H}$  denoted by $\{\ket{\phi_j}\}_{j \in \ZZ_d}$ such that $\ket{\phi_0} = \ket{\mill}$ \footnote{In \cref{alg:opt_attack}, we require that this basis is such that, the vector $\ket{1}$ has non-zero overlap with only $\ket{\phi_0}$ (same as $\ket{\mill}$) and $\ket{\phi_1}$. In other words, the component of $\ket{1}$ orthogonal to $\ket{\mill}$ is proportional to $\ket{\phi_1}$. Such a basis exists and we fix such a basis for our analysis.}. Hence, \begin{equation}\label{eq:private-public}
    \ket{\cent} = \ket{\mill}^{\tensor \kappa} = \ket{\phi_0}^{\tensor \kappa}.
    \end{equation} Clearly, this can be extended to a basis for $\mathbb{H}^{\tensor n}$ given by \[\{\otimes_{k=1}^{n}\ket{\phi_{i_k}}\}_{\vec{i} \in \ZZ_d^n}.\] 
    \item Fix $n, d\in \NN$. For all $\vec{j}=(j_0,\ldots,j_{d-1}) \in \mathcal{I}_{d,n}$ let, 
    the states \nom{S}{$\ket{\BasisSym{n}_{\vec{j}}}$}{$\frac{1}{\sqrt{\binom{n}{\vec{j}}}}\sum_{\vec{i}: T(\vec{i}) = \vec{j}} \ket{\phi_{i_1}\ldots\phi_{i_n}}$} and \nom{S}{$\ket{\BasisSymTilde{n}_{\vec{j}}}$}{$\ket{\cent}\tensor \ket{\BasisSym{n}_{\vec{j}}}$} be defined as
    \begin{align}\label{eq:basis_states_def}
    \ket{\BasisSym{n}_{\vec{j}}} = \ket{\BasisSym{n}_{(j_0,j_1,\ldots, j_{d-1})}} &:= \frac{1}{\sqrt{\binom{n}{\vec{j}}}}\sum_{\vec{i}: T(\vec{i}) = \vec{j}} \ket{\phi_{i_1}\ldots\phi_{i_n}},\\
    \ket{\BasisSymTilde{n}_{\vec{j}}} &:= \ket{\cent}\tensor \ket{\BasisSym{n}_{\vec{j}}}. \\ 
    \end{align}

    \item Let \nom{S}{$\BasisSym{n}$}{orthonormal basis for $\Sym{n}$ denoted by $\{\ket{\BasisSym{n}_{(j_0, \ldots,j_{d-1})}}\}_{\vec{j}\in \mathcal{I}_{d,n}}$} and \nom{S}{$\BasisSymTilde{n}$}{orthonormal basis for the subspace $\SymTilde{n}$ denoted by $\{\ket{\BasisSym{n}_{(j_0, \ldots,j_{d-1})}}\}_{\vec{j}\in \mathcal{I}_{d,n}}$} be sets defined as
    \begin{align}\label{eq:basis_def}
    \BasisSym{n} &:= \{\ket{\BasisSym{n}_{(j_0, \ldots,j_{d-1})}}\}_{\vec{j}\in \mathcal{I}_{d,n}}.\\
    \BasisSymTilde{n} &:= \{\ket{\BasisSymTilde{n}_{\vec{j}}}\}_{\vec{j}\in \mathcal{I}_{d,n}}.
    \end{align}
    It is easy to see that $\BasisSym{n}$ is an orthonormal set. Hence, the set  $\BasisSymTilde{n}$ is also orthonormal. Moreover it can be shown that $\BasisSym{n}$ is an orthonormal basis for $\Sym{n}$---see also Ref.~\cite[Theorem 3]{Har13}).
    \item We will use bold letters to denote subspaces and use the same English letter to denote a particular basis for the subspace, for example - $\Sym{n}$ and $\BasisSym{n}$.
    \item For any state $\ket{\psi}$, we will use $\ket{\widetilde{\psi}}$ to denote the state $\ket{\cent}\tensor\ket{\psi}$. Similarly, for any subspace $A$, we will use \nom{A}{$\tilde{A}$}{Subspace of $n+\kappa$ register states such that the state of first $\kappa$ registers is $\ket{\cent}$ and the state of the last $m\kappa$ registers is some state in $A$, for any subspace $A\subset\mathbb{H}^n$} to represent the subspace $\{\ket{\cent}\tensor\ket{\psi}~|~\ket{\psi} \in A\}$. In a similar way for any basis $B$ we will use $\widetilde{B}$ to denote \[\{\ket{\cent}\tensor\ket{\psi}~|~\ket{\psi} \in B\}.\]
    \item For any hermitian operator $H$, we use \nom{L}{$\lambda_{\text{max}}(H)$}{largest eigenvalue of $H$ for any Hermitian operator $H$} to denote the largest eigenvalue of $H$. \label{item:lambda}

    %\item We denote $A \leq_{B} C$ to denote $A \leq B + C$. 
    %\item For any integer $n$, we denote $B_n$ to be an orthonormal basis for $\mathbb{H}^{\tensor  n} $ of the form $B_n =  \{ \ket{\phi_{i_1}}\tensor\ket{\phi_{i_2}}\tensor \cdots \ket{\phi_{i_n}}  \}_{1\leq i_1,\ldots i_n \leq d}.$ It can be thought of as a natural extension of $B$ in the $n$-qudit space. 
    %\item For $0\leq k \leq n$, we define $B_{n,k}$ to be the set of all those basis states in $B_n$ for which at least $k$ out of the $n$ registers have the quantum state $\ket{\mill}$. Hence,  $B_{n,k} = \{ \ket{\phi_{i_1}}\tensor\ket{\phi_{i_2}}\tensor \cdots \ket{\phi_{i_n}} |%1\leq i_1,\ldots i_n \leq d}, \pi_1(T({\vec{i}})) \geq k  \}$. 
    %Note that, $B_n = B_{n,0}$ forms an orthonormal basis for $\mathbb{H}^{\tensor  n}$.
     
%    \item We now provide a basis for the symmetric subspace of $n$ registers. For $\bm k  = (k_1,k_2,\ldots k_d)$ such that $\sum_{r} k_r = n$, let  $\ket{\mathbb{P}_{n,k}} \equiv \frac{1}{\sqrt{\binom{n}{k}}} \sum_{1 \leq i_1,i_2\ldots i_n \leq d; t({\bf i}) = k} \ket{\phi_{i_1}}\tensor\ket{\phi_{i_2}}\cdots \ket{\phi_{i_n}}$ where  and $\bm i= (i_1,i_2,\ldots i_n).$
\end{enumerate}
\subsection{Preliminaries}\label{subsec:prelims} In this section we will recall some definitions regarding quantum money as well as some tools from linear algebra.

\begin{definition}[Private quantum money (adapted from \cite{Aar09})]\label{definition:private_quantum_money}
A private quantum money scheme consists of the three Quantum Polynomial Time (\nom{Q}{QPT}{Quantum Polynomial Time}) algorithms: $\keygen$, $\bank$ and $\verify$.\footnote{Note that we are implicitly assuming that the quantum money scheme is stateless. Indeed, for a stateful scheme this definition does not hold, for example -~\cite{AMR20}. We will be only concerned with stateless quantum money schemes in this work.}

\begin{enumerate}
    \item $\keygen$ takes a security parameter $\secparam$ and outputs a classical secret key, $sk$. 
    \item  $\bank$  takes the secret key and prepares a quantum money state $\ket{\$}$.\footnote{Even though in most generality the quantum money state may be a mixed state, in all schemes we are aware of the money state is pure, and we use the pure state formalism for brevity.}
    \item $\verify$ receives the secret key and an (alleged) quantum money state $\rho$, which it either accepts or rejects. We emphasize that $\verify$ does not output the post measurement state. 
\end{enumerate}
    
Completeness: the quantum money scheme has perfect completeness if for all $\secpar$
\[ \Pr[sk\gets \keygen(\secparam); \ket{\$}\gets \bank(\sk): \verify(\sk,\ket{\$})=\accept]=1.\]
Notice  that repeated calls to $\bank$ could produce different money states, just like dollar bills, which have serial numbers, and therefore these are not exact copies of each other. Hence, we use $\ket{\$}$ to denote the potentially unique banknotes produced by $\bank$, in order to show the resemblance to dollar bills.
%A money scheme in which the algorithms are not QPT is called an inefficient money scheme
\end{definition}

For any subspace $A$, we will use \nom{A}{$A^\perp$}{the orthogonal subspace of $A$, for any subspace $A$} to denote the orthogonal subspace of $A$ and \nom{P}{$\Pi_A$}{projection onto $A$, for any subspace $A$} to denote the projection onto $A$. For any linear operator $T$ we use \nom{K}{$\ker(T)$}{kernel of $T$, for any linear operator $T$} to denote the kernel of $T$ and \nom{I}{$\Ima(T)$}{the image of $T$, for any linear operator $T$} to denote the image of $T$. We will use \nom{T}{$\tr(\rho)$}{trace of the matrix $\rho$, for any matrix $\rho$} to denote the trace of the matrix $\rho$, for any matrix $\rho$. For any set $S$, we use \nom{S}{$Span(S)$}{the subspace spanned by $S$} to denote the subspace spanned by $S$.
% \begin{definition}[Public quantum money]
% A public quantum money scheme consists of 3 QPT algorithms: $\keygen$, $\bank$ and $\verify$.

% \begin{enumerate}
%     \item $\keygen$ takes a security parameter $\secparam$ and outputs a secret and public key-pair $(\sk,pk)$. 
%     \item  $\bank$  takes the secret key and prepares a quantum money state $\ket{\cent}$.%
%     \item $\verify$ receives the public key and an (alleged) quantum money state $\rho$, and either accepts or reject. 
% \end{enumerate}
% Completeness: the quantum money scheme has perfect completeness if for all $\secpar$
% \[ \Pr[(\sk,pk)\gets \keygen(\secparam); \ket{\cent}\gets \bank(\sk); \verify(pk,\ket{\cent}]=1.\]
% \end{definition}
\subsection{Definitions}\label{subsec:definitions}
In this section we will see some new definitions that would be relevant for our work.
\begin{definition}[Inefficient Quantum Money]\label{definition:inefficient quantum money}
In a money scheme, if at least one of the three algorithms - $\keygen$, $\bank$ and $\verify$ is not QPT, then it is called an inefficient money scheme.
\end{definition}

We generalize the definition of public quantum money given in~\cite{Aar09} by allowing the verification key or the public key to be a quantum state and not necessarily a classical string.
\begin{definition}[Public quantum money (generalized from~\cite{Aar09})]
A public quantum money scheme consists of four QPT algorithms: $\prkeygen$, $\pkkeygen$, $\bank$ and $\verify$. Usually public quantum money has one algorithm $\keygen$ that produces a private-public key pair but we break this into two algorithms $\prkeygen$ and $\pkkeygen$. In our definition, the public key can be quantum, and hence cannot be published in a classical bulletin board; instead, users get access to it via $\pkkeygen$ oracle. Therefore, it is essential to break $\keygen$ in to $\prkeygen$ and $\pkkeygen$.

\begin{enumerate}
    \item $\keygen$ takes a security parameter $\secparam$ and outputs a secret key $sk$. 
    \item $\pkkeygen$ takes the secret key, and prepares a quantum verification key, denoted $\ket{v}$.
    \item  $\bank$  takes the secret key and prepares a quantum money state $\ket{\$}$.%
    \item $\verify$ receives a quantum verification key $\ket{v}$ and an (alleged) quantum money state $\rho$, and either accepts or rejects but never returns the money. If money is returned after verification then it might lead to adaptive attacks which we do not discuss in our work. Therefore, we deviate from the definitions used in other constructions in order to prevent adaptive attacks.
\end{enumerate}
Completeness: the quantum money scheme has perfect completeness if for all $\secpar$
\[ \Pr[sk\gets \keygen(\secparam); \ket{v} \gets \pkkeygen(\sk);  \ket{\$}\gets \bank(\sk): \verify(\ket{v},\ket{\$})=\accept]=1.\]
We require that even after repeated successful verifications of correct money states of the form $\ket{\$}\gets \bank(\sk)$ using the same verification key $\ket{v}$, for a new note $\ket{\widetilde{\$}} \gets \bank(\sk)$, \[ \verify(\ket{v},\ket{\$})=\accept]=1,\], i.e., the completeness property holds even after repeated calls of $\verify$. Note that the quantum public key might change due to both valid and invalid verifications. We use $\ket{\$}$ to denote money states for the same reason as discussed in \cref{definition:private_quantum_money}.
%Notice  that repeated calls to $\bank$ could produce different money states, just like dollar bills, which have serial numbers, and therefore these are not exact copies of each other. Hence, we use $\ket{\$}$ to denote the potentially unique banknotes produced by $\bank$, in order to show the resemblance to dollar bills.

\end{definition}

\ifab \vspace{-10pt} \fi \paragraph*{}In all the existing constructions of money schemes, the public key is a classical string and not a quantum state. Although the scheme we construct is similar to a quantum money scheme with a quantum public key, technically it is what we call a comparison-based definition---see \cref{definition:quantum_coin}. It differs from the quantum money scheme mainly in that the verification key that it uses is the quantum money itself, and therefore, the security notion is slightly different. 

When comparing a quantum public money scheme with a classical key  and a quantum money scheme with a public quantum key, the main difference is that the $\verify$ algorithm of the latter could be thought of as a stateful rather than a stateless protocol. This is because the quantum state that is used as the key can change between different calls to $\verify$. As is often the case in cryptographic protocols, the security definitions and analysis of stateful protocols require  more care than for their stateless counterparts. 
A (private or public) quantum money scheme may output different quantum money states in response to consecutive calls of $\bank(\sk)$. We call the money produced by such schemes \emph{quantum bills}.

\begin{definition}[Quantum Coins (adapted from~\cite{MS10})]\label{definition:quantum coins}
A (private or public) quantum coin scheme is a scheme in which repeated calls of $\bank()$ produce the same (pure)\footnote{May not be true for stateful constructions such as~\cite{AMR20}.} state. We will use \nom{Ce}{$\ket{\cent}$}{A public coin equivalent to $\kappa$ private coins.} to denote a public coin and \nom{M}{$\ket{\mill}$}{a private coin} to denote a private coin.

In a public coin scheme with comparison-based verification, $\verify$ uses one coin as its initial public key.
\label{definition:quantum_coin}
\end{definition}
\begin{definition}[Public Quantum Coins with Private Verification]
In a public coin scheme withe private verification, we have in addition to the public verification algorithm, $\verify_\pk( \cdot)$, a private verification algorithm $\verify_\sk(\cdot)$. These two algorithms may function differently. Note that this public key $pk$ can potentially be a quantum state $\ket{v}$.

%We will use $\cent$ to denote a public coin and a $\mill$ to denote a private coin.
In our construction, a public coin is a collection of private coins. Hence, we will use $\ket{\cent}$ to denote a public coin and $\ket{\mill}$ to denote a private coin.
\label{definition:public_and_private_verification}
\end{definition}

%In a private quantum coin scheme, the bank runs a key-generation to produce the secret-key ($sk$) and runs the minting algorithm $\prqc. \bank$ to produce a private quantum coin $\ket{\mill}$. The bank uses the secret key ($sk$) to run the private verification procedure $\prqc. \verify$ to authenticate a given private coin (alleged) from the user. We say the  private quantum money scheme has \emph{perfect completeness} or is simply \emph{complete} if the verification procedure accepts a valid coin with probability 1, i.e., 
%\[\Pr[\prqc. \verify(\sk, \ket{\mill}) = 1] = 1.\]
%Furthermore, if the coin is valid then they should remain valid even after the private verification and returned to the user.

\begin{definition}[\Count]\label{definition:count}
In any quantum money scheme $\MS$, $\Count$ is a procedure that takes a \emph{key} and a number of alleged quantum money states and runs the verification algorithm on them one by one  and outputs the number of valid quantum states that passed verification.
 
 In a private quantum money scheme, $\Count$ implicitly takes $\sk$ generated using $\keygen(\secparam)$ as \emph{key} whereas in a comparison-based public quantum money scheme instead, $\Count$ takes the wallet as \emph{key} (where the wallet is initialized to a valid coin using $\bank(\sk)$) for (public) verification. In the case of all other public quantum money schemes, \emph{key} is the public key $\ket{v}$ generated by $\pkkeygen(\sk)$.
 
 In the case of public quantum coin with private verification there are two $\Count$ operations - one for the public verification %$\Count(pk,\cdots)$ 
 and the other for the private verification. %$\Count(\sk,\cdots)$. 
 
 %Note as we discussed earlier, that the $pk$ can be a wallet (initialized to a valid coin) or a classical state or a quantum state as well.
\end{definition}

\section{Different notions of security}\label{subsec:unforgeability and security}
%%%%%%%%%%%%%%%%%%%%%%%%%%%%%%%%%%%%

%The security definition in quantum cryptographic protocols and more specifically quantum money vary from model to model as many notions of security are simply not relevant in some models. For both the private and public quantum money we seem to have a consensus for the completeness definition whereas we have different kinds notion for the definition of unforgeability. In  this section, we define the different kinds of unforgeability games and the respective notion of security that are relevant both in the public and private quantum coins setting.

In this section, we will define the different notions of unforgeability but only in the single-verifier setting. The multiverifier versions of these unforgeability notions are defined in \Cref{appendix:multi-ver-attacks}. From now on, we remove the term \emph{single-verifiers} from the names of the security notions, and hence, any security notion that does not mention explicitly in its name that it is a multiverifier notion, should be understood as a single-verifier notion.
As per usual convention, we will use a game-based definition of unforgeability notions. For every unforgeability security game $\mathcal{G}^{\MS}_{\secpar}$ (such as Game~\ref{game:nonadapt_unforge_strongest}), we define a utility of the adversary, to represent the gain it has in Game $\mathcal{G}^{\MS}_{\secpar}$.
\begin{definition}\label{definition:utility}
A utility function with respect to any unforgeability security game $\mathcal{G}^{\MS}_{\secpar}$ is a function,
\[U_{\mathcal{G}^{\MS}_{\secpar}} : \mathcal{G}^{\{\mathcal{A}\},\MS}_{\secpar} \rightarrow \RR, \] where $\mathcal{G}^{\{\mathcal{\adv}\},\MS}_{\secpar}$ is the set of all possible outcomes due to the different adversarial strategies from a set $\{\mathcal{A}\}$, in the security game $\mathcal{G}$ with respect to the money scheme $\MS$ and security parameter $\secpar$. $\{\mathcal{A}\}$ is a fixed set of adversarial strategies/algorithms in the unforgeability game $\mathcal{G}^{\MS}_{\secpar}$. $\mathcal{G}^{\adv,\MS}_\secpar$ denotes a particular outcome of the unforgeability game $\mathcal{G}^{\MS}_{\secpar}$ due to the adversarial strategy $\adv$. 
\end{definition}
Note that in an unforgeability game $\mathcal{G}^{\MS}_{\secpar}$, there can be multiple ways to define the utility of an adversary $U_{\mathcal{G}^{\MS}_{\secpar}}$. 
\begin{definition}[Standard Unforgeability]\label{definition:standard_unforgeability}
A money scheme $\MS$ is \suf in an unforgeability game $\mathcal{G}^{\MS}_{\secpar}$, with respect to a utility $U_{\mathcal{G}^{\MS}_{\secpar}}$, if for every \nom{Q}{QPA}{Quantum Poly-time Algorithm} (Quantum Poly-time Algorithm) $\adv$ in Game $\mathcal{G}^{\adv,\MS}_{\secpar}$,
there exists a negligible function $\negl$ such that,
\begin{equation}
   \Pr[U_{\mathcal{G}^{\MS}_{\secpar}}(\mathcal{G}^{\adv,\MS}_\secpar)>0]\leq \negl. %\mathbb{E}(U(\mathscr{H}))
\end{equation}
\end{definition}
\begin{definition}[Rational Unforgeability]\label{definition:rational_unforgeability}
A money scheme $\MS$ is \ruf in an unforgeability game $\mathcal{G}^{\MS}_{\secpar}$, with respect to a utility $U_{\mathcal{G}^{\MS}_{\secpar}}$, if for every QPA $\adv$ in Game $\mathcal{G}^{\adv,\MS}_{\secpar}$,
there exists a negligible function $\negl$ such that,
\begin{equation}
    \mathbb{E}(U_{\mathcal{G}^{\MS}_{\secpar}}(\mathcal{G}^{\adv,\MS}_\secpar)) \leq \negl, %\mathbb{E}(U(\mathscr{H}))  
    \label{eq:rational-unforge-strongest}
\end{equation}

\end{definition}
%The strength of any unforgeability notion (standard or rational) depends on the strength of the unforgeability game and the utility definition used with respect to the game.

Clearly, in any unforgeability game, standard unforgeability implies rational unforgeability with respect to a common utility.

From now onward, we will use the convention that vanilla unforgeability means standard unforgeability, unless mentioned otherwise.

We start with a general unforgeability game that captures both adaptive and nonadaptive unforgeability notions for all money schemes. In the adaptive setting, the adversary can query the verification procedures on the fly and adapt its strategy based on the outcomes of such queries.
The nonadaptive setting is the special case in which the adversary is allowed to query the verification procedures at most once.

For any (private or public money scheme, including comparison-based public coins and public quantum coins with private verification) quantum money scheme $\MS$, we define the following security game, see Game~\ref{game:nonadapt_unforge_strongest}.

% (We need an interactive protocol here since the adversary decides on the fly whether to submit a coin or to go for the refund. Use cryptocode package.)
\begin{boxx}
%\begin{pcvstack}[boxed, center, space=0.5cm] 
\procedure[linenumbering]{$\uf^{\adv,\MS}_{\secpar}$:}
{ 
sk\gets\keygen(\secparam)\\
fail \gets 0\\
adaptive \gets 0\\
% (\rho_1,\ldots,\rho_m, \tilde{\rho}_1,\ldots,\tilde{\rho}_{\tilde{m}}) \xleftarrow{\text{$\rho_i$ can be potentially entangled}}
\label{line:priv_and_public_ver} \adv(\secpar) \text{ gets access to } {\bank(\sk),\pkkeygen(\sk), \Count_{\pk}, \Count_{\sk}}\\
\pcif \exists \text{ failed verification in a $\Count_\pk$ query}\\
\quad fail \gets 1\\
\pcendif\\
\pcif \text{total number of $\Count_\pk$ and $\Count_\sk$ queries is strictly more than 1}\\
\quad adaptive \gets 1\\
\pcendif\\
\text{Denote by $n$ the number of times that the $\bank(\sk)$ oracle was called by $\adv$}\\
\text{Denote by $m$ the sum of the outcomes of $\Count_\pk$ queries called by $\adv$}\\
\text{Denote by $m'$ the sum of the outcomes of $\Count_\sk$ queries called by $\adv$}\\
%\text{Denote by $fail$ the boolean variable which is $1$ iff all verifications passed in every $\Count_\pk$ queries.}\\
%\text{If all the coins pass ($m=m'$), set } Counter \gets m \\
%\text{Else, set } Counter \gets 0\\
%\text{Utility of the adversary:} U(\adv) = Counter - n.\\
%\text{Utility in the byzantine sense:} \widetilde{U}(\adv) = m' - n.
\pcreturn m,m',n, fail, adapt.
}
% \pcvspace
% \procedure[linenumbering]{$\Count_\sk$}{
% \text{Runs the private verification}\\
% \pcreturn \text{Result of the verification, but not the post-verified state.}
% }
% \procedure[linenumbering]{$\Count_\pk$}{
% \text{Runs the public verification}\\
% \pcreturn \text{Result of the verification, but not the post-verified state.}
% }
% \end{pcvstack}
\caption{Unforgeability Game}
\label{game:nonadapt_unforge_strongest}
\end{boxx} 
\anote{ If the above description of the security game makes less sense we can instead add the following protocol based description.

\procedure[linenumbering, head={$\auf^{\adv,\MS}_{\secpar}$:}]{Adaptive Unforgeability Game}{
\textbf{Verifier ($m=0$)} \< \< \textbf{Adversary} \< \< \textbf{Bank ($n=m'=0$)}\\[][\hline]
\< \< \< \< \sk \gets \keygen(\secparam)\\
\< \<  \vdots \< \< \\
\<\<\< \sendmessage{->}{length=2cm, top = queries mint, width = 0.05cm} \<\\
\<\<\<\< \ket{\$} = \bank(\sk)\\
\<\<\<\< n=n+1\\
\< \< \< \sendmessage{<-}{length=1cm, top = $\ket{\$}$, width = 0.05cm} \< \\
%\<\sendmessage{<-}{length=1cm, width=0.1cm}\<\<\<\\
% \<\sendmessage{->}{length=1cm, width=0.1cm}\<\<\<\\
% \<\<\<\sendmessage{->}{length=3cm, width=0.05cm, top= queries mint, bottom= or submits for verification}\<\\
% \<\<\<\sendmessage{<-}{length=1cm, width=0.05cm}\<\\
\< \<  \vdots \< \< \\ 
\< \sendmessage{<-}{length=1cm, top={$\rho_i$}, bottom={$r_i$ alleged money states}, width = 0.7cm} \< \< \< \\
%\Count_{\pk}(\rho_i) \< \< \< \< \\ 
\label{line:pub_ver}m=m+\Count_{\pk}(\rho_i) \< \< \< \< \\ 
\pcif \Count_\pk \neq r_i\<\<\<\< \\ 
fail=fail+1\<\<\<\<\\
\pcendif \<\<\<\<\\
\< \sendmessage{->}{length=1cm, top={Result}, width = 0.7cm} \< \< \< \\
\< \<  \vdots \< \< \\
\< \< \< \sendmessage{->}{length=1cm, top={$\rho_j$}, bottom={$r_j$ alleged money states}, width = 0.05cm} \< \\
\label{line:priv_ver}\< \< \< \< m' = m' + \Count_{\sk}(\rho_j) \\
\< \< \< \< \pcif \Count_\sk(\rho_j)\neq r_j\\ \<\<\<\<fail=fail+1 \\ 
\< \< \< \< \pcendif\\
\< \< \< \sendmessage{<-}{length=1cm, top={Result}, width = 0.05cm} \< \\
\< \<  \vdots \< \< \\
\<\sendmessage{<-}{length=1cm, width=0.1cm}\<\<\<\\
\<\sendmessage{->}{length=1cm, width=0.1cm}\<\<\<\\
\<\<\<\sendmessage{->}{length=3cm, width=0.05cm, top= queries mint, bottom= or submits for verification}\<\\
\<\<\<\sendmessage{<-}{length=1cm, width=0.05cm}\<\\
\< \<  \vdots  \< \< \\[][\hline]
%\< \< \text{Let $m'$ and $m''$ denote the number of times $\verify_{\pk}$ or $\verify_{\sk}$ was successful, respectively.} \<\<\\ 
\<\<\pcreturn m, m', n, fail \<\<
}

$fail$ is a common variable between Bank and the Verifier, which is initially set to $0$.
}
\paragraph*{}
For a private scheme $\MS$ in Game~\ref{game:nonadapt_unforge_strongest}, we use the convention that $\pkkeygen$ and $\Count_\pk$ in line~\ref{line:priv_and_public_ver} outputs $\bot$. For a public coin scheme with comparison-based verification, $\pkkeygen$ outputs $\bot$, and $\Count_\pk$ represents the comparison-based $\Count$ algorithm that takes a wallet as the verification key (see \cref{definition:count}). For all other public money schemes, including (not comparison-based) public quantum coins with private verification, $\pkkeygen$ outputs a fresh public key, and $\Count_\pk$ represents the public count, i.e., the count with respect to public verification. 
Finally, for private schemes and public coin schemes with private verification, $\Count_\sk$ represents the private count; for all other public money schemes, $\Count_\sk$ outputs $\bot$.
% The procedures $\verify_\sk$ and $\verify_\pk$ in line~\ref{line:priv_and_public_ver} in Game~\ref{game:nonadapt_unforge_strongest} represent the private and public verifications of the money scheme $\MS$ respectively. In the case of comparison-based public coins with private verification, $\Count_\pk$ represents the comparison-based public verification.
We follow the convention that the Count procedures only returns the number of money states that passed verification and not the post-verified state.

The $\bank(sk)$ oracle\footnote{In older works~\cite{Aar09,Aar16,MS10,JLS18}, the adversary is allowed to ask for money states only at the beginning while in Game~\ref{game:nonadapt_unforge_strongest}, the adversary is given oracle access to minting. Giving oracle access to minting does not give the adversary $\adv$ more power in Game~\ref{game:nonadapt_unforge_strongest}, because any adversary $\adv$ with oracle access to $\bank$ is nonadaptive and hence, can be simulated by an adversary, that takes the money states form the mint, all at the beginning. This can be done by taking the maximum number of money states that $\adv$ ever asks for and simulating \adv using those money states. If some money states are unused by the end of the simulation, they can be submitted to the verifier at the end, and all such money state would pass verification due to the completeness of the scheme.} outputs a money state (no matter what the input is), thus providing a way for the forger to receive as much money as it wants to perform the forging.  %This can be done by taking from the mint, at the beginning, the maximum number of coins that $\adv$ ever asks for in its all possible simulations. Since $\adv$ is polynomially bounded the maximum number of coins it can ask for is also polynomial.

Similarly, for public money schemes (except comparison-based schemes), the $\pkkeygen(\sk)$ oracle allows the adversary to get public keys multiple times. In the classical case, that would not make any difference (there is no need for multiple keys), but in the quantum case, the adversary's actions could give it an advantage---e.g., perhaps the secret key could be extracted from multiple copies of the \emph{quantum} verification key, but not from a single copy of the verification key.  

We define separate utility functions for adaptive and nonadaptive unforgeability notions.
With respect to Game~\ref{game:nonadapt_unforge_strongest}, we define the following nonadaptive flexible utility for the adversary, \nom{uf}{$\uflna(\uf^{\adv,\MS}_\secpar)$}{flexible nonadaptive utility for the adversary in Game~\ref{game:nonadapt_unforge_strongest}}.  Since it is a nonadaptive utility, it is defined only when the total number of  $\Count$ queries (i.e., including both $\Count_\pk$ and $\Count_\sk$) is at most one in Game~\ref{game:nonadapt_unforge_strongest}, i.e, $adaptive=0$.
\begin{equation}\label{eq:flex_nonadapt_utility_def}
    \uflna(\uf^{\adv,\MS}_\secpar) = \begin{cases}
                   m + m' - n, &\text{if $adaptive=0$,}\\
                undefined, &\text{otherwise.}
                  \end{cases}
\end{equation}

Similarly, we define the adaptive flexible utility, \nom{uf}{$\ufla(\uf^{\adv,\MS}_\secpar)$}{flexible adaptive utility for the adversary in Game~\ref{game:nonadapt_unforge_strongest}}.

\begin{equation}\label{eq:flex_adapt_utility_def}
    \ufla(\uf^{\adv,\MS}_\secpar) = m + m' - n.
\end{equation} 

\begin{definition}[Standard~\cite{Aar09} and Rational Flexible Adaptive unforgeability]\label{definition:adapt_flex_unforge}
A quantum money scheme $\MS$ is \frauf (respectively, \fauf) if it is \ruf (respectively, \suf) in  Game~\ref{game:nonadapt_unforge_strongest}, with respect to  the utility $\ufla$.
% there exists a negligible function $\negl$ such that,
% \begin{equation}
%     \mathbb{E}(U_{Strong}({\adv})) \leq \negl, %\mathbb{E}(U(\mathscr{H}))  
%     \label{eq:rational-unforge-strongest}
% \end{equation}where $U_{Strong}(\adv)$ is as defined in \cref{eq:rational_utility_def_strongest}.
\end{definition} 

\begin{definition}[Standard~\cite{Aar09} and Rational Flexible Nonadaptive unforgeability]\label{definition:nonadapt_flex_unforge}
A quantum money scheme $\MS$ is \frnauf (respectively, \fnauf) if it is \ruf (respectively, \suf) in  Game~\ref{game:nonadapt_unforge_strongest}, with respect to  the utility $\uflna$.
% there exists a negligible function $\negl$ such that,
% \begin{equation}
%     \mathbb{E}(U_{Strong}({\adv})) \leq \negl, %\mathbb{E}(U(\mathscr{H}))  
%     \label{eq:rational-unforge-strongest}
% \end{equation}where $U_{Strong}(\adv)$ is as defined in \cref{eq:rational_utility_def_strongest}.
\end{definition}

We call this utility function flexible because even when some of the money states it submitted, do not pass verification, the definition counts the rest of the money states which passed verification. One can think of a more restrictive utility definition, in which unless all the money states submitted by the adversary pass verification, no coins are counted in the adversary's utility. We call this the all-or-nothing utility, which leads to a weaker unforgeability notion.% However, we would not delve in to the details here and we leave that to future work.

We will only consider the nonadaptive variant of the all-or-nothing utility \nom{ua}{$\uana(\uf^{\adv,\MS}_\secpar)$}{all-or-nothing adaptive utility for the adversary in Game~\ref{game:nonadapt_unforge_strongest}} and 
 \nom{ua}{$\uanna(\uf^{\adv,\MS}_\secpar)$}{all-or-nothing nonadaptive utility for the adversary in Game~\ref{game:nonadapt_unforge_strongest}} because the adaptive version is not required in the context of our work, respectively.
\begin{equation}\label{eq:all-or-nothing_nonadaptive_utility_def}
    \uanna(\uf^{\adv,\MS}_\secpar) = \begin{cases}
                   undefined, &\text{if $adaptive=1$,}\\
                   m + m' - n, &\text{if $fail=0$ and $adaptive=0$,}\\
                -n, &\text{otherwise.}
                  \end{cases}
\end{equation}

\begin{equation}\label{eq:all-or-nothing_adaptive_utility_def}
    \uana(\uf^{\adv,\MS}_\secpar) = \begin{cases}
                   m + m' - n, &\text{if $fail=0$,}\\
                m'-n, &\text{otherwise.}
                  \end{cases}
\end{equation}
The motivation for all-or-nothing utility comes from scenarios where only one kind of coin is used for a particular item. Suppose a person goes to buy a TV from an honest seller but is allowed to buy only one TV. He puts all the money on the table according to the worth of the TV she plans to buy. The seller either approves the transaction and gives a TV or rejects it and simply says no to the user but does not return the money back to the buyer. Even if one of the money states fails verification, the seller does not approve the transaction.

\paragraph*{}

% Hence, we would discuss nonadaptive all-or-nothing utility only instead of separate discussion for both the variants.}

In the context of our work, we will only talk about the nonadaptive all-or-nothing utility function and not its adaptive version. It is not clear if our construction (given in \cref{alg:ts}) satisfies all-or-nothing adaptive rational unforgeability.  %Later in \cref{appendix:multi-ver-unforge} (see the discussion in the last paragraph), we show that our construction is \mnaruf, a stronger notion of unforgeability that implies rational unforgeability with respect to the adaptive all-or-nothing unforgeability.

% After we describe the \nauf game (Game~\ref{game:nonadapt_unforge}), We would later show that under this utility definition both adaptive and non-adaptive unforgeability notions are the same for any public coin scheme with private verification. 

\begin{definition}[Standard and Rational All-Or-Nothing Nonadaptive unforgeability]\label{definition:nonadapt_all-or-nothing_unforge}
A quantum money scheme $\MS$ is \arnauf (respectively, \anauf) if it is \ruf (respectively, \suf) in  Game~\ref{game:nonadapt_unforge_strongest}, with respect to  the utility $\uanna$.
% there exists a negligible function $\negl$ such that,
% \begin{equation}
%     \mathbb{E}(U_{Strong}({\adv})) \leq \negl, %\mathbb{E}(U(\mathscr{H}))  
%     \label{eq:rational-unforge-strongest}
% \end{equation}where $U_{Strong}(\adv)$ is as defined in \cref{eq:rational_utility_def_strongest}.
\end{definition} 
Clearly by definitions (see \cref{eq:flex_nonadapt_utility_def,eq:all-or-nothing_nonadaptive_utility_def}), $\uflna(\uf^{\adv,\MS}_\secpar)\geq \uanna(\uf^{\adv,\MS}_\secpar)$ and hence, if a scheme is then it is also \anauf.

% As mentioned earlier, we are interested in rational unforgeability since our scheme is not standard unforgeable. 
%In order to address both the standard and rational unforgeability, we defined the above quantity  ${U}_{Strong}(\adv)$, which we call the utility for the adversary $\adv$. It represents the gain of the adversary in Game~\ref{game:nonadapt_unforge_strongest}. In standard unforgeability, we require that the probability that the adversary has a positive utility/gain is negligible. However for rational unforgeability, we only demand that the expected loss of the adversary is negligible.
 
The standard definitions of nonadaptive or adaptive unforgeability used in the literature, such as~\cite{Aar09}, coincide with \fnauf and \fauf schemes, respectively, which are the strongest among all the other unforgeability notions, discussed here in the adaptive and nonadaptive setting, respectively. Hence, from now onwards, we would simply use \nauf and \auf to mean \fnauf and \fauf schemes, respectively.

It turns out that the \fauf definition (both standard and rational) is too strong for the comparison-based public coin scheme with private verification that we construct due to the following reasons. 
\begin{remark}\label{remark:adaptive_Attack}
The construction $\pkqc$ (given in \cref{alg:ts}) is not \frauf and hence not \fauf.
\end{remark}  
The proof sketch is as follows: We can construct the following adaptive attack. There exists an attack (described in \ref{alg:opt_attack}), where the adversary starts from $n$ coins and tries to pass $n+1$ coins. The probability that all coins pass the public verification in this attack is close to $1$. Now, in Game~\ref{game:nonadapt_unforge_strongest}, the adversary would do the following: It performs the (nonadaptive) algorithm as mentioned in \cref{alg:opt_attack} to prepare the $n+1$ alleged coins. It would send these  $n+1$ alleged coins one by one to the verifier for public verification. The adversary would stop at the first failure and would submit the rest of the unsubmitted coins to the bank for private verification (refund).

Intuitively, the $n+1$ coin state constructed according to \cref{alg:opt_attack} has $n$ true coins and one false coin. Therefore, the coins submitted for refund would have, at most, one false coin. Hence, all but at most one coin would pass the private verification. On the other hand, since the adversary stops at the first failure, it loses at most one coin in the public verification phase. Therefore, the adversary loses at most two coins (utility is minus two), which happens if a failure occurs in the public verification phase. Clearly, the adversary gains a coin if no such failure happens. Therefore, if the attack is successful, i.e., all coins pass verification, then the adversary gains 1; otherwise, it loses at most two. Since the success probability of~\cref{alg:opt_attack} is close to $1$, the expected utility of the adversary would be positive and non-negligible. For example, for $n =\kappa^2$, and $\kappa$ large enough the success probability is greater than $\frac{3}{4}$ (see \cref{subsec:analysis_of_attack} for more details), in which case the expected utility of the adversary would be at least $\frac{3}{4}\cdot1 - \frac{1}{4}\cdot2 = \frac{1}{4}$. As a result, our construction is not \frauf, and therefore not \fauf either. 

Hence, we conclude that we need to soften our security definition and settle for a weaker notion which is the nonadaptive unforgeability. 
Moreover, we will see the public coins scheme with private verification that we construct is neither \nauf nor \anauf. We show that it is possible to prepare $n+1$ alleged coins starting from $n$ coins, such that all $n+1$ coins pass verification with overwhelming probability, see \cref{alg:opt_attack,sec:attack}.

However, we show that our construction $\pkqc$ (given in \cref{alg:ts}) is \arnauf, but we do not know if it is \frnauf, which we leave for future works. From now onwards, we will use the convention that nonadaptive rational unforgeability implicitly means nonadaptive all-or-nothing rational unforgeability.

\begin{definition}[Unconditional security]
We call an adversary that can apply $\poly$, and if queries to the oracles, but that is otherwise computationally unbounded an \emph{unbounded} adversary. 

For all the security notions above, we define an unconditional security flavor, in which the
definition is with respect to unbounded adversaries.
\label{definition:unconditional_unforgeability}
\end{definition}
Note that, for a \nauf or \auf (both stateless) private money scheme, the number of coins the adversary submits and the number of correct coins it takes from the mint (denoted by $n$ in Game~\ref{game:nonadapt_unforge_strongest}), cannot be exponential\footnote{For a stateful private money scheme, it is indeed possible to have both the number of coins the adversary submits and the number of correct coins it takes from the mint arbitrary, for example -~\cite{AMR20}.}. If the adversary is allowed to get exponentially many copies of the coin, then it can use standard tomography to learn the unique quantum state of the coin. On the other hand, if it is allowed to submit exponentially many coins, then she can submit the maximally mixed state, exponentially many times, which would result in a non-negligible probability for positive utility $\uflna(\uf^{\adv,\MS}_\secpar)$ in Game~\ref{game:nonadapt_unforge_strongest}.

\section{Our construction and results in the single-verifier setting}\label{sec:constructions}
Our main result is the following.
\begin{theorem}\label{thm:multi-unconditional_secure}
If \prs exist, then a \mnars (see \cref{definition:multi_rational_secure}) public quantum coin with comparison-based verification also exists, and hence such quantum coin schemes are in microcrypt\footnote{The set of all quantum cryptographic primitives weaker than one-way functions.}.
Furthermore, there exists an inefficient quantum money scheme and also a stateful quantum money scheme which are comparison-based public quantum coins with private verifications that is \mnars (see \cref{definition:multi_rational_secure}) unconditionally (see \cref{definition:unconditional_unforgeability}). 
\end{theorem}
The proof of \cref{thm:multi-unconditional_secure} is given on \cpageref{pf:thm:multi-unconditional_secure}.  
For the sake of simplicity, the formal definitions of the multiverifier security notions used in the main result are not discussed in the main text and are instead given in \cref{appendix:multi-ver-attacks}. We emphasize that the multiverifier security notions capture the real threat model under our user manual (given in \cref{subsec:user_manual}); see \cref{appendix:user_manual-multi-ver-equivalence} for more details. 
Next, note that we have not formally defined a stateful quantum money scheme but have used it in \cref{thm:multi-unconditional_secure}. As defined in~\cite{AMR20}, stateful private quantum money schemes are money schemes in which the minting and private verification procedures are stateful instead of being stateless. There is no key-generation procedure and the internal state maintained serves as the secret information analogous to a secret key in stateless schemes. 
We can extend this definition to stateful comparison-based public quantum coins with private verification by allowing the mint and private verifications to be stateful. The internal state acts as the alternative to the secret key. We do not need public keys for public verification since we consider a comparison-based money scheme. From now onward, we will mostly mean quantum money schemes to be stateless unless mentioned otherwise.

%The construction that achieves \cref{thm:multi-unconditional_secure} is given in \cref{alg:ts}.
In the main text, we will only prove the single-verifier version of \cref{thm:multi-unconditional_secure}, i.e., proving that our construction is \nars, which is a weaker yet interesting simplification of the corresponding multiverifier security notions. We upgrade the analysis for the same construction in \cref{appendix:multi-ver-results} to achieve \cref{thm:multi-unconditional_secure}. %on the single-verifier notions of security, which is a weaker yet interesting simplification of the multiverifier notions.
Our main result in the single-verifier setting is the following.
\begin{theorem}\label{thm:unconditional_secure}
There exists a comparison-based public quantum coin scheme (see \cref{definition:quantum coins}) with private verification, which is \nars, i.e., both (all-or-nothing) \rnauf and \narf (see \cref{definition:nonadapt_all-or-nothing_unforge,definition:rational_fair} respectively), based on \prs.

Furthermore, there exists an inefficient (see \cref{definition:inefficient quantum money}) quantum money scheme and also a stateful quantum money scheme which are comparison-based public quantum coin schemes with private verification, that is unconditionally \nars. 
\end{theorem}
The proof is given in \cref{appendix:appendix_proofs} on \cpageref{pf:thm:unconditional_secure}. 

%We note that our construction remains the same  also prove the multiverifier version of \Cref{thm:unconditional_secure}, see \Cref{thm:multi-unconditional_secure}. We refer the reader to \cref{appendix:multi-ver-attacks} for the security definitions in the multiverifier setting.
Notice that we have not yet discussed the definition of \narf money schemes. The definition (\cref{definition:rational_secure} is given in \cref{appendix:fairness}. We delay the discussion to the appendix for two reasons - unforgeability is the most important security notion, and the other security notions, namely, security against sabotage, are not that interesting to discuss.

We now discuss the construction that lets us achieve our results (\cref{thm:multi-unconditional_secure,thm:unconditional_secure}). The formal description of our construction is given in the \cref{alg:ts}. We denote $\Pi_{\Sym{n}}$ to denote the orthogonal projection onto the symmetric subspace of $n$ registers. Suppose $\prqc$ is a private coin scheme (with algorithms  $\prqc.\keygen$, $\prqc. \bank$ and $\prqc. \verify$). We define a public coin scheme as follows. $\pkqc. \keygen$ is the same as $\prqc. \keygen$, and $\pkqc. \bank$ produces $\kappa$ coins of the private quantum coin scheme instead of one using $\prqc. \bank$ (needs to be written in an algorithm). Hence, each public quantum coin is a collection of $\kappa$ private quantum coins where $\kappa \in \log^c(\secpar), c>1$. We define a \emph{wallet} where we keep the public coins. When the user receives a new coin for verification, it uses the public coins already in the wallet  for verification. On successful verification, it adds the new coin to the wallet. Initially the wallet is instantiated with one valid coin $\pkqc.\bank(\sk)$ from the bank. If at any point the wallet has $m$ public coins, then the running of $\pkqc. \verify$ on the one new coin that was received executes a projective measurement into the symmetric subspace on the combined $(m+1)\kappa$ registers of the wallet and the new coin. If the projective measurement succeeds, the verification algorithm accepts the new coin as authentic. It is known that the measurement $\{\Pi_{\Sym{n}}, (I - \Pi_{\Sym{n}})\}$ can be efficiently implemented~\cite{BBD+97}. From now onward, we will use the convention that for every algorithm $A$, we sometime use the pure state formalism and write $A(\ket{\psi})$ instead of $A(\ketbra{\psi})$.

%$\pkqc.\Count_{\ket{\cent}}(\rho)$ is the same as $\pkqc.\Count(\rho)$ for any quantum state $\rho$ just to emphasize that we before the count operation we initiate the wallet with one fresh coin. Moreover, for a pure state $\ket{\alpha}$, we will use the convention that % and $\ket{\beta}$
%$\pkqc.\Count_{\ket{\cent}}(\ketbra{\alpha})$ is the same as $\pkqc.\Count_{\ket{\cent}}(\ket{\alpha})$. Similarly $\pkqc.\verify(\ketbra{\alpha})$ is the same as $\pkqc.\verify(\ket{\alpha})$.

\begin{algorithm}[htbp!]
    \caption{Construction of $\pkqc$: A public quantum coin scheme }
\label{alg:ts}
%We assume $\prqc$ is a private quantum coin scheme.  
    \begin{algorithmic}[1] % The number tells where the line numbering should start
        \Procedure{$\keygen$}{$1^{\secpar}$}
            \State  $(\emptyset,sk) \gets \prqc.\keygen(\secpar)$ \Comment{Note that there is no public key.\footnote{We define the $\keygen$ procedure only if the underlying private scheme has a $\keygen$ procedure.}}
            \State \textbf{return} $(\emptyset,sk)$
        \EndProcedure
        
        \Procedure{$\bank$}{$\sk$}
            \State $\kappa \equiv \log(\secpar)^c$ for some constant $c>1$.
            \State $\ket{\mill}^{\tensor \kappa} \gets ((\prqc.\bank(\sk))^{\tensor \kappa}$
            \State \textbf{return} $\ket{\cent} = \ket{\mill}^{\tensor \kappa}$
        \EndProcedure
        \State {\bf Init:} $\omega \gets \bank(\sk)$ \label{line:init}\Comment{Before running the first verification, we assume the user receives one valid public coin from the bank.}
        \Procedure{$\verify$}{$\rho$}
            \State Denote by \nom{O}{$\tilde{\omega}$}{combined state of the wallet and the new alleged coins received for verification using the wallet} the combined wallet state \nom{O}{$\omega$}{state of wallet in general, possible after a verification involving symmetric subspace measurement} and the new coin $\rho$.\label{line:pre-measure}  \Comment{Note that $\widetilde{\omega}$ is not necessarily $\omega \tensor \rho$ since they might be entangled.}
            \State Measure the state $\widetilde{\omega}$ with respect to the two-outcome measurement $\{\Pi_{\Sym{\kappa\cdot(1+m)}}, I-\Pi_{\Sym{\kappa\cdot(1+m)}}\}$.
            \State Denote the post measurement state the new wallet state $\omega$. \label{line:post-measure} 
            \State $m \gets m + 1$
            \If{Outcome is $\Pi_{\Sym{\kappa\cdot(1+m)}}$}
                \State {\bf accept}.
            \Else
                \State {\bf reject}. \Comment{Note that we do not return any register to the person submitting the coins for verification; We only notify them that the coins were rejected.}
                %\State Go to the bank to verify the post-measurement state, $(I-\Pi_{\Sym{\kappa\cdot(1+m)}})\widetilde{\omega}(I-\Pi_{\Sym{\kappa\cdot(1+m)}})$.
            \EndIf
        \EndProcedure
        \Procedure{$\Count_{\ket{\cent}}$}{$(\rho_1,\ldots, \rho_m)$}\Comment{Here, each $\rho_i$ represents a state over $\kappa$ registers}\label{line:Count}
        \State Set $Counter \gets 0$.
        \State Run \textbf{Init} to initialize the wallet $\omega \gets \ket{\cent} = \bank(\sk) $.
        \For{$i = 1$ to $m$}
        \State Run $\verify( \rho_i)$
        \If{$\verify(\rho_i) = \text{accept}$}
            \State $Counter = Counter +1$.
        \EndIf
        %such that the initial state of the wallet is $\ket{\cent} = \bank(\sk)$.%\Comment{Note that the state of the wallet changes after each verification} %and $\omega_{i+1}$ is the state of the wallet after the $i^{th}$ verification for $i = 1,\ldots, m-1$.
        \EndFor
        \State{\textbf{Output}} $Counter$. 
        \EndProcedure
        \algstore{pageChange}
    \end{algorithmic}
\end{algorithm}
\begin{savenotes}
\begin{algorithm}
    \begin{algorithmic}[1]
        \algrestore{pageChange}
        \Procedure{$\verify_{bank}$}{$\sk$, $\rho$}\Comment{Here, $\rho$ represents a state over $\kappa$ registers.}
            \State  $k \gets \prqc.\Count(\sk, \rho)$  
            \State Destroy the coin state.
            \State With probability $\frac{k}{\kappa}$, \textbf{Accept} and mint a fresh new coin as the replacement\footnote{The fresh coin acts as the post verified state. Usually we do not return the post verified state. However in some cases such as Game~\ref{game:nonadapt_rational-fair} and Game~\ref{game:multi_nonadapt_rational-fair}, we do need the post verified state and there we use this freshly minted coin as the post verified state after acceptance.}, and with probability $1 - \frac{k}{\kappa}$, \textbf{reject}.
        \EndProcedure
        \Procedure{$\Count_{bank}$}{$\sk$, ($\rho_1,\ldots, \rho_m$)}\label{line:count_bank}\Comment{Here, $\rho_i$ represents a state over $\kappa$ 
        registers}
        \State Set $Counter \gets 0$.
        \For{$i = 1$ to $m$}
        \State Run $\verify_{bank}( \rho_i)$
        \If{$\verify_{bank}(\rho_i) = \text{accept}$}
            \State $Counter = Counter +1$.
        \EndIf
        %such that the initial state of the wallet is $\ket{\cent} = \bank(\sk)$.%\Comment{Note that the state of the wallet changes after each verification} %and $\omega_{i+1}$ is the state of the wallet after the $i^{th}$ verification for $i = 1,\ldots, m-1$.
        \EndFor
        \State{\textbf{Output}} $Counter$.
        \EndProcedure
        %\algstore{CountTotal}
    \end{algorithmic}
\end{algorithm}
\end{savenotes}
The construction \cref{alg:ts} is a comparison-based public quantum coin scheme with private verification (see \cref{definition:quantum_coin} and \cref{definition:public_and_private_verification}) where $\verify$ and $\verify_{bank}$ are interpreted as $\verify_{pk} $ and $\verify_{sk}$ respectively and similarly for $\Count$ and $\Count_{bank}$. %We describe another algorithm

It is easy to see that our construction is complete.
\begin{proposition}\label{prop:completeness}
The quantum public coin scheme $\pkqc$ is complete. 
\end{proposition}
The proof is given in \cref{appendix:appendix_proofs} on \cpageref{pf:prop:completeness}.

Our construction satisfies the all-or-nothing nonadaptive rational unforgeability, defined in the previous section (\cref{definition:nonadapt_all-or-nothing_unforge}).
\begin{proposition}\label{prop: rational-unforge}
The scheme $\pkqc$ in \cref{alg:ts} is \arnauf (see \cref{definition:nonadapt_all-or-nothing_unforge}) if the underlying private scheme $\prqc$ is \nauf (see \cref{definition:adapt_flex_unforge}) and $\prqc.\verify$ is a rank-$1$ projective measurement. Moreover if the $\prqc$ is \nauuf (see  \cref{definition:unconditional_unforgeability}) then the $\pkqc$ will be unconditionally \arnauf (see \cref{definition:unconditional_unforgeability}). If the underlying $\prqc$ scheme is inefficient then the $\pkqc$ will also be inefficient but still all the results will hold. 
\end{proposition} 
The proof is given in \cref{subsec:completeness_security proofs} on \cpageref{pf:prop: rational-unforge}. The multiverifier result of the above lifting result also holds, and it is stated in \cref{prop: multi_ver-unforge}.

%In \cref{prop: rational-unforge}, we show that our construction $\pkqc$ in \cref{alg:ts} is \arnauf. 
Later in \cref{sec:attack} (see \cref{alg:opt_attack}), we show that the relaxation of the unforgeability notion to rational unforgeability is necessary, and that standard nonadaptive unforgeability (both with respect to all-or-nothing and flexible utilities, see \cref{eq:all-or-nothing_nonadaptive_utility_def} and \cref{eq:flex_nonadapt_utility_def}) does not hold for our construction, $\pkqc$. In fact, the attack succeeds with probability, inverse polynomially close to $1$ (see \cref{subsec:analysis_of_attack}).
%Our construction $\pkqc$ (described in \cref{alg:ts}) is not \nauf in the usual notion of unforgeability. Hence $\pkqc$ serves as an example of quantum money schemes which is not \nauf in the strict sense but still is \arnauf, thus having strong security guarantees. Later we have discussed an attack (see \cref{alg:opt_attack} and \cref{prop:opt_attack}) the success probability of which is asymptotically close to $1$. 

Our construction also satisfies security against sabotage (see \cref{appendix:fairness}). We elaborately discuss these properties in \cref{appendix:fairness}.

We instantiate our construction (see \cref{alg:ts}) $\pkqc$ with the private quantum coin scheme in~\cite{JLS18} \anote{(or the simplified version in~\cite{BS19}} and~\cite{MS10} (as the underlying $\prqc$ scheme). The private coin schemes provide the following results

\begin{theorem}[{Adapted from~\cite[Theorem 6]{JLS18}}]
If \prs exist, then there exists a private quantum coin scheme that is \nauf (see \cref{definition:nonadapt_flex_unforge}) such that the verification algorithm is a rank-$1$ projective measurement.
\label{thm:qaruf}
\end{theorem}

%(In fact, the private coin scheme constructed in~\cite{JLS18,BS19} is $\qarf$ (see  {\auf}))
\begin{theorem}[{Adapted from~\cite[Theorem 4.3]{MS10}}]
There exists an inefficient private quantum coin scheme that is \auuf and hence \nauuf (see \cref{definition:adapt_flex_unforge,definition:nonadapt_flex_unforge} and \cref{definition:unconditional_unforgeability})  such that the verification algorithm is a rank-$1$ projective measurement.
\label{thm:qaruuf}
\end{theorem}
The authors in~\cite[Theorem 4.3]{MS10} prove Black-box unforgeability for their construction which is the same as unconditional (flexible) adaptive unforgeability (see \cref{definition:adapt_flex_unforge}) in private quantum coin schemes.

The private coin schemes in~\cite{JLS18,MS10} are stateless quantum money schemes. In a recent work~\cite{AMR20}, the authors construct a stateful private coin scheme. They show a way to simulate Haar random states using which, they construct the stateful coins scheme. Their scheme can be viewed as a stateful implementation of the quantum coins scheme by~\cite{MS10}. Their result can be summed up as follows.

\begin{theorem}[{Restated from~\cite[Proposition 3]{AMR20}}]
There exists a stateful private coin scheme which is 
%\ut (see \cref{definition:byzantine_untrace_coins}) and
\nauuf (see \cref{definition:nonadapt_flex_unforge} and \cref{definition:unconditional_unforgeability}).
\label{thm:qaruuf_Stateful}
\end{theorem}

On instantiating our construction $\pkqc$ using the stateful private coin scheme in~\cite{AMR20}, we construct a stateful comparison-based public quantum coins scheme with private verification.

Combining the lifting result (\cref{prop: rational-unforge}) and completeness, \cref{prop:completeness} with \cref{thm:qaruf}, \cref{thm:qaruuf} and the result~\cite[Theorem 7]{AMR20} which shows the indistinguishability of the schemes in~\cite{AMR20} and~\cite{MS10}, along with \cref{thm:rational public fair,thm:rational priv fair}, that we will see later in \cref{appendix:fairness} respectively, gives us \cref{thm:unconditional_secure}.

\begin{remark}
As noted by Aaronson and Christiano~\cite{AC13}:
\begin{quotation}\label{quote: impossibility}
It is easy to see that, if public-key quantum money is possible, then it must rely on some computational assumption, in addition to the No-Cloning Theorem. 
This is because a counterfeiter with unlimited time could simply search for a state $\ket{\psi}$ that the (publicly-known) verification procedure accepted.
\end{quotation}
Although this argument holds equally well for most public quantum money schemes, it breaks down when the public scheme uses a quantum state as the public key: As the exponential number of verifications could perturb the public quantum key, a state that passes verification by the perturbed quantum key may fail with a fresh quantum key. Note that by tweaking the definition of public quantum money, i.e., by adding the notion of a public quantum key, we managed to circumvent this impossibility result in \cref{thm:unconditional_secure}.
\end{remark}

\subsection{Complexity and efficient implementations of \texorpdfstring{$\pkqc$}{Pk-QC}}\label{subsec:complexity}
Note that in the scheme \pkqc, each public coin is a quantum state over $\kappa$ many registers (private coins) where $\kappa$ is polylogarithmic in $\secpar$ (where $\secpar$ is the security parameter), and the local dimension of each register is given by $d$ (see notations in \cref{subsec:notations}). In other words, each public coin is a state over $\kappa\log(d)$ qubits, where $d$ depends on the private money scheme, $\prqc$.

The private coin scheme in~\cite{MS10} is secure, even if the number of qubits for each private coin is set to $\log^c(\secpar)$ for some $c>1$. In fact, what they essentially show is that, as long as $n$ is superlogarithmic in the security parameter, there exists an inefficient private quantum coin scheme on $n$ qubits, that is black-box unforgeable. The security guarantees hold due to  the \emph{complexity-theoretic no-cloning theorem} (\cite[Theorem 2]{Aar09}), which asserts the following fact : Given polynomially many copies of a Haar random state on $n$-qubits, and an oracle access (with polynomially many queries) to the reflection around the state, the optimal cloning fidelity, is negligible, as long as $n$ is superlogarithmic. Hence, on instantiating $\pkqc$ with the~\cite{MS10} scheme on polylogarithmic qubits, we get a public coin construction on polylogarithmic qubits.
%This can be done by choosing a Haar random state, over only $\log^c(\secpar)$ (with $c>1$) many qubits, as the private coin. Since it is an inefficient scheme, we do not worry about the complexity. The security guarantees hold since the optimal cloning fidelity, given polynomially many copies of a Haar random state, over only $\log^c(\secpar)$ (with $c>1$) many qubits, is still negligible, by the ``complexity-theoretic no-cloning theorem'' (\cite[Theorem 2]{Aar09}).  If we use the above version of private scheme in~\cite{MS10}, to instantiate \pkqc, then the number of qubits required for each public coin, will be $\log^{2c}(\secpar)$ (with $c>1$). In the private coin scheme in~\cite{JLS18} (as well as in the simplified version of \cite{BS19}), the local dimension of a register representing a coin, is taken to be exponential in the security parameter, i.e., the number of qubits required in each private coin is linear in the security parameter. However, 

The private coin construction given in~\cite{JLS18} is a modular construction using Pseudo-Random family of States (PRS, defined in~\cite{JLS18}). We would not delve into the discussion about \nom{P}{PRS}{Pseudo Random family of States}, but the authors prove the following result. 
\begin{theorem}[Private coin construction from PRS,~\cite{JLS18}]
Let $n\in\omega(\log(\secpar))$.
Suppose, there exists a PRS $\{\ket{\phi_k}\}_{k\in \mathcal{K}}$ on $n$-qubits, such that  for every $k\in \mathcal{K}$, the state $\ket{\phi_k}$, given the key $k$, can be constructed in time $t(\secpar)$\footnote{$t(\secpar)\in O(\poly)$ by the definition of PRS.}. Then, there exists a \nauf private coin scheme such that each coin is a quantum state on $n$-qubits, such that $\bank$ and $\verify$ algorithm runs in time $O(t(\secpar))$. Moreover, the $\keygen$ algorithm takes $O(\log(|\mathcal{K}|))$ time, where $\mathcal{K}$ is the key-space of the PRS.
\label{thm:priv-coin_from_PRS}
\end{theorem}

%\begin{comment}%Since we removed all discussions regarding one-way functions
In~\cite{JLS18}, the authors also propose a PRS construction based on a quantum-secure Pseudo-Random Function family (\nom{P}{PRF}{Pseudo-Random Function family}), in order to instantiate \cref{thm:priv-coin_from_PRS}. More precisely, they prove the following:
\begin{theorem}[PRS based on PRF,~\cite{JLS18}]
Suppose, there exists a quantum secure PRF $\{f_k\}_{k\in\mathcal{K}}$ on $n$-bit inputs, such that $n\in\omega(\log(\secpar))$ and for every $k\in\mathcal{K}$, $f_k$ can be implemented on a quantum computer, in time $t$\footnote{$t\in\poly$ by definition of PRF.}. Then, there exists a PRS family $\{\ket{\phi_k}\}_{k\in \mathcal{K}}$ on $n$-qubits, such that the key-space is the same as the key-space of the PRF. Moreover, for every $k\in \mathcal{K}$, given the key $k$, the state $\ket{\phi_k}$, can be constructed in time, $\poly[n] + t$.
\label{thm:PRS_from_PRF}
\end{theorem}
It is known that by~\cite{Zha12}, PRFs on inputs of bit-size polynomial, exist assuming the existence of quantum-secure one-way functions. Hence, using \cref{thm:PRS_from_PRF}, the authors construct a PRS over polynomially many qubits, and polynomial construction time, based on quantum-secure one-way functions. By instantiating \cref{thm:priv-coin_from_PRS} with such a PRS, they prove their main result, \cref{thm:qaruf}. However, in order to get close to optimal result using \cref{thm:priv-coin_from_PRS}, we require a PRS over $n$ qubits such that $n=\log^c(\secpar)$, for some $c>1$. By \cref{thm:PRS_from_PRF}, we can construct such a PRS, using a PRF on $n$ bits such that $n=\log^c(\secpar)$, for some $c>1$. Moreover, in order to achieve polylogarithmic running time of the PRS, we would require that the PRF has polylogarithmic running time\footnote{Note that the running time of the PRS determines the running time of $\bank$ and $\verify$ in the private coin scheme, and are hence crucial for the efficiency of the private coin scheme.}. For such an optimal PRS, note that, the corresponding private quantum money scheme by \cref{thm:priv-coin_from_PRS}, would have polylogarithmic time complexities and each coin would be on polylogarithmic qubits. Note that, on instantiating \pkqc, with such a private scheme, would mean that each public coin is over $n\kappa$ qubits, which is polylogarithmic for the choice of $\kappa$ and $n$. Since, $\pkqc.\bank$ is the same as running $\prqc.\bank$  $\kappa$ many times, it can be done in polylogarithmic time. Moreover, the verification of a public coin, \pkqc.\verify, using a wallet with a fresh coin is a symmetric subspace measurement over $2\kappa$ private coins, and hence can also be done in polylogarithmic time. This is because the projective measurement into the  symmetric subspace, can be implemented in time, quadratic in the number of registers and square logarithmic in the local dimension of registers~\cite{BBD+97}.
However, we require a PRF with polylogarithmic input size and running time, which is a strong form of PRF. In practice, for such purposes, block ciphers are used (such as AES~\cite{DR11}) instead of PRF, see~\cite[Chapter 3.5]{KL14} for more details. Hence, we can use a block cipher with the same properties namely, polylogarithmic input size and running time. The main downside of using block ciphers is that they use a fixed block size and hence, do not fit the asymptotic analysis, which we use through out this work. At the same time though, block ciphers have the advantage that the best known quantum attack is slightly below $2^{z/2}$, where $z$ is the key-size (which is expected due to Grover's search, see the post-quantum cryptanalysis of AES~\cite{BNS19}). %Hence, we can use a block cipher with polylogarithmic key and input size that would run in polylogarithmic time, required for the optimal PRS.

\paragraph*{}
Another way to implement the private coin scheme in~\cite{JLS18} efficiently, is to use a \emph{Scalable PRS} construction, a notion that was recently introduced in~\cite{BS20b}. In~\cite{BS20b}, the authors prove the following:
\begin{corollary}[{Restated from~\cite{BS20b}}]
 If post-quantum one-way functions exist, then for every $n\in \NN$ (even constant), independent of the security parameter $\secpar$, there exists a PRS on $n$ qubits.\label{cor:scalable_PRS}
\end{corollary}
In particular, we can get a PRS by fixing $n=\log^c(\secpar)$, for some $c>1$, as required for the optimal PRS in \cref{thm:priv-coin_from_PRS}. This would result in a private coin scheme on $n$-qubits by \cref{thm:priv-coin_from_PRS}, where $n$ is as above, and instantiate $\pkqc$ using such a private scheme. Thus, we get the following result.

\begin{theorem}[Nonadaptive unforgeability for \cite{JLS18} with PRS over roughly logarithmic qubits]
If quantum-secure one-way functions exist, then there exists a \nauf private coin scheme, with a rank-$1$ projective measurement and $\log^c(\secpar)$ (with $c>1$) qubits required for each coin.\label{thm:poly-log_private_coins}
\end{theorem}

This way, we can avoid the use of block ciphers and their limitations, as discussed in the previous paragraph. However, using the PRS construction in~\cref{cor:scalable_PRS}, we lose out on the guarantee of polylogarithmic running time, that we had with block ciphers. The definition of PRS only ensures that the running time is polynomial. As a result, if we instantiate $\pkqc$ using a private scheme with such a PRS, then the minting algorithm would run in polynomial time and not polylogarithmic time. The verification time is still polylogarithmic since it only depends on the number of qubits used for each private coin.

%\end{comment}%Since we removed all discussions regarding one-way functions

\subsection{How to use \texorpdfstring{$\pkqc$}{Pk-QC}: User manual}\label{subsec:user_manual}
\paragraph*{Motivation}Our construction $\pkqc$ (see \cref{alg:ts}), %which is a comparison-based public quantum coin with private verification 
is \arnauf (see \cref{definition:nonadapt_all-or-nothing_unforge}), but it is not secure against adaptive attacks (see the discussion after \cref{definition:nonadapt_all-or-nothing_unforge}). One way to avoid adaptive attacks, is by using a fresh wallet for receiving each transaction, thereby preventing adaptive attacks. Moreover, we could only prove that our scheme is \arnauf. In the all-or-nothing utility (see \cref{eq:all-or-nothing_nonadaptive_utility_def}), we require that in every transaction, the user either approves all the coins (if all of them pass verification) or approves none. The users are allowed to go to the bank to get a refund or valuation of the coins she possess. There can be potential sabotage attacks where an honest user, after doing transaction with adversarial merchant, gets a refund less than what she should get. In order to provide a meaningful way of refund especially in the case of failed verification, we use the $\pkqc.\Count_{bank}$ or the private count. In \cref{appendix: priv rational fair}, we prove that if the user goes for refund after a transaction, she would be secure against sabotage attacks, in the rational sense.  Moreover, in \cref{appendix:public_sabotage}, we show that if an honest user spends money that was accepted upon successful verification, then also she is rationally secure against sabotage attacks in the all-or-nothing model of transactions. Since, these notions are fairly technical, we skip the discussion on the results regarding sabotage attacks to the appendix (see \cref{appendix:fairness}). %We require a slightly non-standard way of refund from the bank in order to prove that sabotage attacks are avoided, i.e., the honest users do not got cheated (discussed in detail, in \cref{appendix:other_results}). This non-standard way of refund is needed only to ensure security against sabotage.

\paragraph*{Specification} Our construction $\pkqc$, (see \cref{alg:ts}) should be used in the following way: the user starts with a wallet called the main wallet which contains public coins from the bank. 
The receiver possesses multiple receiving wallets---one receiving wallet per received payment. In order to receive a payment, the user needs to have at least one coin in her main wallet.
The receiver brings out one coin from her main wallet and creates a separate receiving wallet with that coin. He uses this new wallet to receive and apply $\pkqc.\verify()$ on the received sum from the payer. The transaction is approved if and only if the $\pkqc.\verify()$ accepts all the coins of the submitted sum using the newly formed receiving wallet. If the transaction fails, the receiver does not return the post-verified money state to the (cheating) payer. The user can pay money from the main wallet as well as the receiving wallet, if the money received had passed verification. In case of a failed verification, the user needs to go and get a refund of the received money in the corresponding receiving wallet. 
To refund any given receiving wallet, the bank applies $\pkqc.\Count_{bank}$ on the wallet coins.

%To refund any given receiving wallet, the bank applies $\prqc.\verify$ on the wallet and gives a refund of $\floor{\prqc.\Count(\rho)/\kappa}$ many fresh public coins. The remaining fraction of a coin (namely, $\prqc.\Count(\rho)/\kappa- \floor{\prqc.\Count(\rho)/\kappa}$) would be added to the user's bank account (that can be taken into account in later refunds), since a refund cannot contain a fraction of a (public) coin.\footnote{We use this definition of bank's refund in order to prove sabotage attacks are not possible. In case, the designer of the money scheme does not care about the security against sabotage, the bank's refund could simply be defined as $\Count_{bank}(\rho)$, where $\rho$ is the quantum state of the given receiving wallet.}
 
%This is what motivates us to tweak the nonadaptive rational security against sabotage game (see -  Game~\ref{game:nonadapt_rational-fair}) in the following way. From now onwards, we will 

\paragraph*{Analysis of the user manual} 
The user manual is well-suited for our construction, $\pkqc$, described in \cref{alg:ts} as well as for any comparison-based public quantum coin scheme with private verification. In the user manual, we implicitly assume that the verification of the money scheme is done using wallets initialized by a fresh coin, which is very similar to comparison-based verification.
%In the \ro mode described above, the receiver needs a fresh coin from the bank to verify the received coins. Hence, the \ro mode of operation already subsumes that the scheme in use is a public quantum coin scheme with comparison-based verification.
We also require a separate private procedure for the bank's refund. Hence, the user manual implicitly assumes that the scheme in use is a public quantum coin scheme with private verification.

Note that every received wallet is used only once to receive a transaction, which is either successful and all the coins are approved or none of the coins are accepted. Hence, the user manual indeed ensures the all-or-nothing utility definition that we use in \cref{definition:nonadapt_all-or-nothing_unforge}. 

It can be shown that if the user manual is followed, any cheating forger can be viewed as an adversary in a multiple verifier version of Game~\ref{game:nonadapt_unforge_strongest} with a nonadaptive all-or-nothing utility. The notions of \mnaruf money schemes, and how it captures all attacks on the scheme $\pkqc$, are discussed more precisely in \cref{appendix:multi-ver-unforge,appendix:user_manual-multi-ver-equivalence}\footnote{In \cref{appendix:user_manual-multi-ver-equivalence}, it is also discussed how all kinds of sabotage attacks under the user manual are captured by the multiverifier version of nonadaptive rational security against sabotage, defined in \cref{appendix:fairness}}. The scheme \pkqc is indeed \mnaruf (see \cref{prop: multi_ver-unforge}). The proof is fairly easy but has some technicalities; because of this, we skip the corresponding results and their proofs to the appendix (see \cref{appendix:multi-ver-unforge,appendix:appendix_proofs}). The proof goes via a couple of reductions. We first prove that if the underlying private scheme, $\prqc$, is \mnauf, then in the nonadaptive setting, rational security against private sabotage attacks against multiple verifiers implies multiverifier rational unforgeability for the scheme, \pkqc. We then use the observation that in the rational sense, any nonadaptive multiverifier private sabotage attack for any general public money scheme can be reduced to a nonadaptive private sabotage adversary against a single-verifier using a single payment. Then, we prove that the scheme, \pkqc, is rationally secure against private sabotage targeted on a single-verifier using a single payment in the nonadaptive setting. Hence, as a side product, we also prove multiverifier security against private sabotage for the scheme, \pkqc.%This can be shown as follows. If the underlying $\prqc$ scheme is \nauf, a nonadaptive forging attack against multiple verifiers can be reduced to a nonadaptive sabotage attack against multiple verifiers, i.e., the expected utility in the forging attack is less than the expected loss in the sabotage attack, up to negligible correction. Moreover, for any nonadaptive sabotage attack against multiple verifiers, we can construct a nonadaptive sabotage attack against a single-verifier (see Game~\ref{game:nonadapt_rational-fair}), without much decrease in the expected loss. the scheme, $\pkqc$ is indeed \narpf (see \cref{thm:rational priv fair}).  Therefore, the user manual ensures that, if the scheme, $\pkqc$ (see \cref{alg:ts}) is \arnauf, then any forging attack on the scheme can only result in negligible utility in expectation. The last reduction (from multiple to single-verifier sabotage) can be done by randomly choosing one of the many verifiers to target, and simulating the rest. In order to simulate other verifiers, the sabotage adversary against a single-verifier, is required to have oracle access to $\prqc.\verify$. Note that, this is not allowed in Game~\ref{game:nonadapt_rational-fair}. Hence, we require a stronger definition of nonadaptive rational security against sabotage than \cref{definition:rational_fair}. Nevertheless, the proof of nonadaptive rational security  against sabotage for our scheme (\cref{thm:rational priv fair}), given on \cpageref{pf:thm:rational priv fair}, can be adapted to show that our scheme is nonadaptive rationally secure against sabotage even in the strong definition\footnote{In fact, the proof does not require any assumption on the state submitted by the adversary.}.

The user manual, just as in the case of unforgeability, ensures that any kind of sabotage attack against our scheme $\pkqc$ (with respect to the user manual) can be viewed as a multiverifier nonadaptive sabotage attack, i.e., an adversary that tries to sabotage by submitting to multiple verifiers, one by one. The discussion about sabotage attacks and the multiverifier version are neither that interesting nor important, and hence we skip it to the appendix (\cref{appendix:fairness,appendix:multi-ver-fair}). In \cref{appendix:multi-ver-fair}, we deduce that the scheme, $\pkqc$ is rationally secure, even against multiverifier sabotage attacks\footnote{We do not discuss multiverifier rational security against public sabotage, but the same reduction in the private sabotage case (see \cref{prop:sing-multi-rational-fair}) shows that public sabotage attacks on multiple verifiers are equivalent to that of the single-verifier case, see \cref{remark:single-multi-rational-fair-public}. Hence, we instead prove rational security against public sabotage attacks on a single-verifier in \cref{appendix:public_sabotage}.}, in the nonadaptive setting (see \cref{cor:multi_ver-rational-fair}). The proof goes through some intermediate results, the proofs of which are given in \cref{appendix:appendix_proofs}.  All these results regarding unforgeability and security against sabotage in the multiple verifiers setting are summarized in  \cref{appendix:multi-ver-results}, which together imply the main result, \cref{thm:multi-unconditional_secure}. %We prove by observing that in the rational sense, any nonadaptive multiverifier sabotage attack for any general public money scheme can be reduced to a nonadaptive sabotage adversary against single-verifier. Then, we prove that the scheme, \pkqc, is indeed rationally secure against sabotage against single-verifiers in the nonadaptive setting.

\paragraph*{Potential use case} The user manual has the following main restrictions, namely the all-or-nothing mode of transaction and the requirement of a fresh coin to receive a transaction. However, it is relevant and applicable in various cases. For example consider a shop selling electronic goods such as TV or computer, the vendor usually receives money from buyers and gives the item to the buyer only if it the transaction is successful. In general, the vendor never has to pay. In particular the credit card terminal machine, that are used in practice operate in a manner similar to the user manual, since they only receive money\footnote{A typical credit card terminal also allows refunds. In our setting, it is not so simple; the vendor can pay from her wallet, including the receiving wallets.}. Similarly, the vendor can operate through quantum coins using the user manual - receiving the sum of money from buyers into separate receiving wallets (one for each transaction), and approving the transaction only if all the coins pass. She can spend the money which was received upon successful verification or go to the bank to get a refund of her receiving wallets. 

Since, the user manual allows either approve all coins or no coins in a transaction, in order to verify $n$ coins in a transaction, it suffices to do just one symmetric subspace projective measurement on all the $(n+1)\kappa$ registers of the new coins as well as the one fresh coin in the wallet. We accept all coins if the measurement outcome is into the symmetric subspace. This is the same as doing $n$ verifications one by one, because the symmetric subspace of a bigger system is a subspace of the the symmetric subspace of a smaller subsystem. Therefore, the probability that all the coins pass verification subjected to $\pkqc.\verify$, one by one, is the same as the squared overlap of the $(n+1)\kappa$ registers (wallet coin and new coins) with the symmetric subspace over $(n+1)\kappa$ registers. 
As a result, the time required for verifying $n$ coins in a single transaction, is equivalent to the time required to perform a projective measurement into the symmetric subspace over $(n+1)\kappa$ registers, which requires time quadratic in the number of registers (which is $O(n\kappa)$), see~\cite{BBD+97}. Note that, if the $n$ coins are submitted in multiple transactions, then the total verification time can only decrease. Hence, $n$ coins can be verified using $O(n^2\kappa^2)$ time.

\iffalse
Before proving the soundness
\begin{theorem}\label{weak soundness}
For every $ n > 0$ and for all QPA $\adv$ having all the public information $k$, producing $n$ public quantum coins $\adv(k)$ (alleged) from no coins, $$\Pr[\pkqc. \verify(\adv(k))] \leq negl(\kappa).$$  
\end{theorem}

\begin{proof}

%\qed 
\end{proof}
\fi

\section{Unforgeability}\label{sec:attack}
As mentioned earlier, our construction is not unforgeable according to the standard unforgeability notions, i.e., the scheme $\pkqc$ is neither \nauf (i.e., with respect to flexible utitlity) nor \anauf (see \cref{definition:nonadapt_flex_unforge,definition:nonadapt_all-or-nothing_unforge}). 
In the next two subsections, \cref{subsec:attack_description,subsec:analysis_of_attack}, we discuss a class of nonadaptive attacks (see \cref{alg:opt_attack}) on our construction parameterized by $n, m \in \poly$. In \cref{subsec:optimality}, we prove that for any nonadaptive QPT adversary which takes $n$ coins from the mint and submits $m$ alleged coins for public verification, the attack has the maximum probability (up to negligible corrections) for passing all the $m$ verifications provided the underlying private scheme, $\prqc$ (the private scheme that we lift to $\pkqc$ in \cref{alg:ts}) is \nauf (see \cref{definition:adapt_flex_unforge}). The analysis of this attack will be vital in the proof of (all-or-nothing) nonadaptive rational unforgeability for our construction, given in \cref{subsec:completeness_security proofs}.  %in the \ro mode.

%First lets define a few useful states. Recall the orthogonal basis $\mathbb{B}$ containing $\ket{\mill}$ (mentioned in the notations) where $\ket{\mill}$ is the private quantum money state.

%\ket{\mill}\tensor \cdots \underbrace{\ket{\mill^{\perp}}}_{i_1} \tensor \cdots \underbrace{\ket{\mill^{\perp}}}_{i_2} \tensor \cdots \underbrace{\ket{\mill^{\perp}}}_{i_{\kappa}} \tensor \cdots \ket{\mill}
 %As discussed before, our construction is unforgeable in the strong Byzantine sense. Next, we discuss the optimal forging attack (see \cref{alg:opt_attack}) and the optimal success probability and then using them prove the rational security for our construction in the \ro mode. %Later we also prove the stricter notion of security for the $\rpkqc$ setting.
\subsection{Candidate nonadaptive attack}\label{subsec:attack_description}
A class of nonadaptive forgery attacks parameterized by $m, n \in \NN$ such that $m>n$, is described in  \cref{alg:opt_attack}, in which the adversary gets $n$ coins from the mint, and submits $m$ alleged coins. Hence, for every $n$, the attack is successful if running $\pkqc.\Count_{\ket{\cent}}()$ (see  \cref{line:Count}) on the submitted coins reads $m$.
%\begin{savenotes}
\begin{algorithm}
\caption{A class of Nonadaptive attacks on the scheme $\pkqc$, parameterized by $n$}
\label{alg:opt_attack}
\begin{algorithmic}
    \State Obtain $n$ copies of public coins $\ket{\cent}^{\tensor n} \gets (\pkqc.\bank(\sk))^{\tensor n}$.
    \State Construct the $m\kappa$ register state $\ket{\BasisSym{m\kappa}_{(n\kappa, (m-n)\kappa, 0\ldots,0)}} $ (see Notations in \cref{subsec:notations}) which is the same as     
    $\frac{1}{\sqrt{\binom{m\kappa}{n\kappa}}}\sum_{\vec{i}, T(\vec{i}) = (n\kappa, (m-n)\kappa, 0,\ldots)} \ket{\phi_{i_1}}\tensor \cdots \ket{\phi_{i_{m\kappa}}}.$
    %$ \frac{1}{\sqrt{\binom{(n+1)\kappa}{n\kappa}}}\sum_{i_1\ldots i_{\kappa}}\ket{\phi_1}\tensor \cdots \underbrace{\ket{\phi_2}}_{i_1} \tensor \cdots \underbrace{\ket{\phi_2}}_{i_2} \tensor \cdots \underbrace{\ket{\phi_2}}_{i_{\kappa}} \tensor \cdots \ket{\phi_1}$ %which can be done by adding another register initialized to $\ket{0}$ and then doing control swap operation.
    \State Submit the state $\ket{\BasisSym{m\kappa}_{(n\kappa, (m-n)\kappa, 0\ldots,0)}}$ to the verifier.
\end{algorithmic}
\end{algorithm}
%\end{savenotes}
The construction of the state $\ket{\BasisSym{m\kappa}_{(n\kappa, (m-n)\kappa, 0\ldots,0)}}$ from $\ket{\cent}^{\tensor n}$ can be done as follows: 
Add $(m-n)\kappa$ registers each initialized to $\ket{1}$ to the $n\kappa$ registers and call these $m\kappa$ registers the input registers.\footnote{We use the fact that the basis for $\mathbb{H}$, that we fixed in \cref{item:product_basis_definition} in \cref{subsec:notations}, is such that the vector $\ket{1}$ has non-zero overlap with only $\ket{\phi_0}$ (same as $\ket{\mill}$ and $\ket{\phi_1}$. Hence, the component of $\ket{1}$, orthogonal to $\ket{\mill}$ (which is overwhelmingly large in our case), is proportional to $\ket{\phi_1}$.}
Note that the underlying private scheme $\prqc$ is \nauf. In particular, the state $\ket{1}$, which can be prepared efficiently, must have very little fidelity with the correct coin state $\ket{\cent}$, otherwise, the QPT algorithm which produces the state $\ket{1}$ can  nonadaptively forge the scheme $\prqc$. Therefore the state has overwhelmingly high fidelity with a state of the form $\ket{\mill}^{\tensor n\kappa}\tensor\ket{\mill^\perp}^{\tensor (m-n)\kappa}$ where $\ket{\mill^\perp}$ is some state orthogonal to $\ket{\mill}$. The fidelity of $\ket{\mill^\perp}$ with $\ket{1}$ is overwhelmingly large.

Add another $m\kappa$  work registers initialized to \[\frac{1}{\sqrt{\binom{m\kappa}{n\kappa}}}\sum_{\vec{i}, T(\vec{i}) = (n\kappa, (m-n)\kappa, 0,\ldots)}\ket{i_1}\tensor \cdots \ket{i_{m\kappa}}.\]%\[\frac{1}{\sqrt{\binom{(n+1)\kappa}{n\kappa}}}\sum_{i_1\ldots i_{\kappa}}\ket{0}\tensor \cdots \underbrace{\ket{1}}_{i_1} \tensor \cdots \underbrace{\ket{1}}_{i_{\kappa}} \tensor \cdots \ket{1}.\] 
Apply controlled swap operation controlled at the work registers to get the following intermediate state with high fidelity\[ \frac{1}{\sqrt{\binom{m\kappa}{n\kappa}}}\sum_{\vec{i}, T(\vec{i}) = (n\kappa, (m-n)\kappa, 0,\ldots)} \left(\ket{\phi_{i_1}}\tensor \cdots \ket{\phi_{i_{m\kappa}}}\right)\tensor \left(\ket{i_1}\tensor \cdots \ket{i_{m\kappa}}\right).\]
%\[ \frac{1}{\sqrt{\binom{(n+1)\kappa}{n\kappa}}}\sum_{i_1\ldots i_{\kappa}}(\ket{\phi_1}\tensor \cdots \underbrace{\ket{\phi_2}}_{i_1} \tensor  \cdots \underbrace{\ket{\phi_2}}_{i_{\kappa}} \tensor \cdots \ket{\phi_1})\tensor \ket{0}\tensor \cdots \underbrace{\ket{1}}_{i_1}  \tensor \cdots \underbrace{\ket{1}}_{i_{\kappa}} \tensor \cdots \ket{1}.\]
Apply C-Swap operations again but this time controlled on the input registers. Since the state $\ket{\mill^\perp}$ is close to $\ket{1}$ (fidelity wise) applying the C-Swap operation is almost the same as disentangling the work and the input registers such that we are left with a pure state in the input registers which has an overwhelmingly high fidelity with the state $\ket{\BasisSym{m\kappa}_{(n\kappa, (m-n)\kappa, 0\ldots,0)}}$.\label{subsec:analysis_of_attack} 
%\begin{equation}
%\sum_{i_1,\ldots,i_{(m-n)\kappa}} 
%\end{equation}Therefore the state written above is very $\ket{\BasisSym{m\kappa}_{(n\kappa, (m-n)\kappa, 0\ldots,0)}}$ on the input registers. 

\subsection{Analysis of the attack}

Clearly, $\ket{\BasisSym{m\kappa}_{(n\kappa, (m-n)\kappa, 0\ldots,0)}}$ is a symmetric state such that \[\Pr[\prqc.\Count(\ket{\BasisSym{m\kappa}_{(n\kappa, (m-n)\kappa, 0\ldots,0)}}) = n\kappa] = 1,\] for every $n$ and $m > n$. Hence, the attack does not violate the adaptive unforgeability (see \cref{definition:adapt_flex_unforge}) of the underlying private scheme $\prqc$. Next, for every $n$ and $m>n$, the success probability of the attack:
\begin{equation}
Pr[\pkqc.\Count_{\ket{\cent}}(\ket{\BasisSym{m\kappa}_{(n\kappa, (m-n)\kappa, 0\ldots,0)}}) = m] = \frac{\binom{m\kappa}{n\kappa}}{\binom{(m+1)\kappa}{(n+1)\kappa}}.
\label{eq:probability_opt_attack}
\end{equation} 
This can be seen in the following way.
Observe that the combined state of the new coins and the wallet (initialized to $\ket{\cent}$) just before the $\pkqc.\Count$ operation (see  \cref{line:Count} in  \cref{alg:ts}) is $\ket{\BasisSymTilde{m\kappa}_{(n\kappa, (m-n)\kappa, 0\ldots,0)}}$ (similar to $\widetilde{\omega}$ in  \cref{line:pre-measure} in  \cref{alg:ts}). Recall, 
\begin{align}
\ket{\BasisSymTilde{m\kappa}_{(n\kappa, (m-n)\kappa, 0\ldots,0)}} &= \ket{\cent}\tensor \ket{\BasisSym{m\kappa}_{(n\kappa, (m-n)\kappa, 0\ldots,0)}}\\ &= \ket{\phi_0}^{\tensor \kappa} \tensor \ket{\BasisSym{m\kappa}_{(n\kappa, (m-n)\kappa, 0\ldots,0)}}.\\
\end{align}
For notations, see \cref{eq:basis_states_def} and \cref{eq:private-public} in \cref{subsec:notations}. Notice that, $\ket{\BasisSymTilde{m\kappa}_{(n\kappa, (m-n)\kappa, 0\ldots,0)}}$ has a non-trivial overlap with only one vector in the basis $\BasisSym{(m+1)\kappa}$, which is $\ket{\BasisSym{(m+1)\kappa}_{((n+1)\kappa, (m-n)\kappa, 0\ldots,0)}}$.
%test $\braket{a}{b}$
It is not hard to see that \[ \left|\braket{\BasisSym{(m+1)\kappa}_{((n+1)\kappa, (m-n)\kappa\ldots)}}{\BasisSymTilde{m\kappa}_{(n\kappa, (m-n)\kappa, 0\ldots,0)}}\right|^2  = \frac{\binom{m\kappa}{n\kappa}}{\binom{(m+1)\kappa}{(n+1)\kappa}}.\]

Hence, the squared overlap of $\ket{\BasisSymTilde{m\kappa}_{(n\kappa, (m-n)\kappa, 0\ldots,0)}}$ with $\Sym{(m+1)\kappa}$ is $\frac{\binom{m\kappa}{n\kappa}}{\binom{(m+1)\kappa}{(n+1)\kappa}}.$
This completes the derivation of \cref{eq:probability_opt_attack}.
\paragraph*{}Next we show that the attack described in \cref{alg:opt_attack} also shows that our scheme $\pkqc$ is not \nauf in the traditional sense. Note that, the probability of passing at least $(n+1)$ verifications out of $m$ is
\begin{align}
    \Pr&[\pkqc.\Count_{\ket{\cent}}(\ket{\BasisSym{m\kappa}_{(n\kappa, (m-n)\kappa, 0\ldots,0)}})>n]\\ 
    &\geq \Pr[\pkqc.\Count_{\ket{\cent}}(\ket{\BasisSym{m\kappa}_{(n\kappa, (m-n)\kappa, 0\ldots,0)}}) = m]\\
    &= \frac{\binom{m\kappa}{n\kappa}}{\binom{(m+1)\kappa}{(n+1)\kappa}}.
\end{align}
%Therefore, our scheme is not \nauf in the traditional sense.

\paragraph*{}For $m= n+1$, the probability that more than $n$ coins pass the verification is exactly equal to  $\frac{\binom{((n+1)\kappa}{n\kappa}}{\binom{(n+2)\kappa}{(n+1)\kappa}}$. It can be shown that the term $\frac{\binom{((n+1)\kappa}{n\kappa}}{\binom{(n+2)\kappa}{(n+1)\kappa}}$ asymptotically converges to $1$ when $n \rightarrow \infty$. This is because,
\begin{align}
    \frac{\binom{((n+1)\kappa}{n\kappa}}{\binom{(n+2)\kappa}{(n+1)\kappa}} &= \frac{((n+1)\kappa)!((n+1)\kappa)!}{((n+2)\kappa)! (n\kappa) !}\\
    &=\prod_{r=1}^\kappa 
    \frac{n\kappa +r}{(n+1)\kappa + r}\\ 
    &> \left(\frac{n\kappa}{(n+1)\kappa}\right)^\kappa = \left(\frac{n}{n+1}\right)^\kappa > \left(1 - \frac{1}{n}\right)^\kappa.
\end{align}
Clearly, the term $\left(\frac{n}{n+1}\right)^\kappa > \left(1 - \frac{1}{n}\right)^\kappa \rightarrow 1,$ as $n \rightarrow \infty$.
Hence, the scheme $\pkqc$ is not \nauf. Moreover, a little analysis also shows that for $n = c \cdot\kappa$ and taking the limit of large $\kappa$, the term goes to $e^{-1/c}=1-\frac{1}{c}+O(\frac{1}{c^2})$, although we do not use it in any of our results. 

On the other hand, this attack fails when the adversary starts with one public coin, i.e., $n = 1$. In this case the expression becomes $\frac{\binom{2\kappa}{\kappa}}{\binom{3\kappa}{\kappa}}$, which is upper bounded by $\left(\frac{2}{3}\right)^\kappa=\negl$, for our choice of $\kappa$. %large enough ($\log^c(\secpar)$, $c>1$ ).

\subsection{Optimal success probability for nonadaptive forgery}\label{subsec:optimality}

%Now we will prove its optimality - $ \ket{\BasisSym{m\kappa}_{(n\kappa, (m-n)\kappa, 0\ldots,0),(n+1)\kappa}}$ is the state that the adversary could have submitted to achieve the optimal success probability, having received $n$ copies of the public coins from the bank.
In this section we will prove the optimality (up to negligible corrections) of the attack given in \cref{alg:opt_attack} in \cref{subsec:attack_description}. 

\begin{proposition}[Optimality of the attack]\label{prop:opt_attack}
Suppose $\prqc$ is $\nauf$ (see \cref{definition:nonadapt_flex_unforge}), and $\prqc.\verify$ is a rank-$1$ projective measurement.
Consider a nonadaptive QPT adversary in Game~\ref{game:nonadapt_unforge_strongest}, which takes $n$ coins from the mint does not query the private verification oracle, and submits $m$ registers for public verification such that $m,n \in \poly$, and $m>n$. For such an adversary, the attack described in \cref{alg:opt_attack} is optimal (i.e., has the highest possible probability that all $m$ are accepted), up to additive negligible corrections, against $\pkqc$ (see \cref{alg:ts}). Moreover if the underlying $\prqc$ scheme is \auuf (see \cref{definition:nonadapt_flex_unforge} and \cref{definition:unconditional_unforgeability}), then the attack is optimal even for computationally unbounded adversaries. Note that even for such an adversary, $m,n \in \poly$, i.e., it can submit and receive polynomially many coins. %\nauf (see \cref{eq:unforgeability})
\end{proposition}
%We will start by stating and proving a simple algebraic lemma. We will later use this lemma to prove the optimality of the attack given in \cref{alg:opt_attack} in the previous section.

The full proof is given on \cpageref{pf:prop:opt_attack}.

The proof follows by combining the security guarantees of the underlying $\prqc$ scheme along with some algebraic results that we are going to see in the next lemmas.
\paragraph*{} For every $m,n \in \NN$ such that $m > n$, let \nom{G}{$\Good{m\kappa}{n\kappa}$}{Subspace of $m\kappa$ registers states on which $\prqc.\Count$ is $\leq n\kappa$ with probability $1$}, \nom{G}{$\GoodTilde{m\kappa}{n\kappa}$}{Subspace of $(\kappa+m\kappa)$ register states such that the quantum state of the first $\kappa$ registers is $\ket{\cent}$ and the state of the rest $m\kappa$ register is a vector in $\Good{m\kappa}{n\kappa}$}, \nom{B}{$\Bad{m\kappa}{n\kappa}$}{$(\Good{m\kappa}{n\kappa})^{\perp}$} and  \nom{B}{$\BadTilde{m\kappa}{n\kappa}$}{Subspace of $m\kappa$ registers such that the state of the first $\kappa$ registers is $\ket{\cent}$ and the state of the last $m\kappa$ registers is some state in $\Bad{m\kappa}{n\kappa}$} and  be subspaces defined as 
\begin{align}
    \Good{m\kappa}{n\kappa} &:= \{ \ket{\psi} \in (\mathbb{H})^{m\kappa}| \Pr[\prqc.\Count(\sk, \ketbra{\psi}) \leq n\kappa] = 1 \},\\
    \GoodTilde{m\kappa}{n\kappa}  &:= \{\ket{\cent}\tensor \ket{\psi}| \ket{\psi} \in \Good{m\kappa}{n\kappa}\}.\\
    \Bad{m\kappa}{n\kappa} &:= (\Good{m\kappa}{n\kappa})^{\perp},\\
    \BadTilde{m\kappa}{n\kappa} &:= \{\ket{\cent}\tensor\ket{\psi}~|~\ket{\psi} \in \Bad{m\kappa}{n\kappa}\}.
    \label{eq:subspace definition}
\end{align}
Since we assume that $\prqc.\verify$ is a rank-$1$ projective measurement $(\ketbra{\mill}, I - \ketbra{\mill})$, $\Good{m\kappa}{n\kappa}$ is essentially the span of all the states with at least $(m-n)\kappa$ out of the $m\kappa$ registers having quantum state orthogonal to $\ket{\mill}$ and $\mathbb{{\GoodTilde{m\kappa}{n\kappa}}}$ is the subspace of all $(\kappa + m\kappa)$ registers such that the quantum state of the first $\kappa$ registers is $\ket{\cent}$ and the state of the rest $m\kappa$ register is a vector in $\Good{m\kappa}{n\kappa}$. Similarly, the subspace $\BadTilde{m\kappa}{n\kappa}$ consists of all $(\kappa + m\kappa)$ registers such that the quantum state of the first $\kappa$ registers is $\ket{\cent}$ and the state of the rest $m\kappa$ register is a vector in $\Bad{m\kappa}{n\kappa}$. Since we assume that the underlying $\prqc$ scheme is \nauf (see \cref{definition:nonadapt_flex_unforge}), if any QPT adversary that in the unforgeability game (Game~\ref{game:nonadapt_unforge_strongest}) against $\pkqc$  takes $n$ public coins, and submits $m$ (such that $n\leq m$) alleged coins for public verification, then the quantum state of the submitted coins for public verification, must have an overwhelming overlap (squared) with $\Good{m\kappa}{n\kappa}$ and negligible overlap (squared) with $\Bad{m\kappa}{n\kappa}$. Hence, the quantum state of the submitted coins has overwhelming overlap (squared) with $\Good{m\kappa}{n\kappa}$ and negligible overlap with $\Bad{m\kappa}{n\kappa}$. 
% Note that, the same is true in the following scenario, the underlying \prqc is \auf. The adversary starts with $n+m'$ coins, submits $m$ and $\tilde{m}$ coins for public and private verification respectively out of which $m'$ successfully passes private verification.
Every vector in $\GoodTilde{m\kappa}{n\kappa}$ (resp. $\BadTilde{m\kappa}{n\kappa}$) represents the combined state of the verifier's wallet (initialized to $\ket{\cent}$) and a $\kappa m$ register state in ${\Good{m\kappa}{n\kappa}}$ (resp.  $\Bad{m\kappa}{n\kappa}$) submitted by the adversary, just before the
$\pkqc.\Count_{\ket{\cent}}()$ operation (see \cref{line:Count} in \cref{alg:ts}).
\paragraph*{}Clearly the subspaces $\BadTilde{m\kappa}{n\kappa}$ and $\GoodTilde{m\kappa}{n\kappa}$ are orthogonal spaces. It follows from the definition that for every $m,n \in \NN$ and $m>n$,
\begin{equation}\label{eq:dollarspacesubspacelargerspace}
    \GoodTilde{m\kappa}{n\kappa} \subset \Good{\kappa + m\kappa}{\kappa + n\kappa}.
\end{equation}
The relation between the subspaces $\Good{(m+1)\kappa}{(n+1)\kappa}$, $\GoodTilde{m\kappa}{n\kappa}$, $\Bad{(m+1)\kappa}{(n+1)\kappa}$ and $\BadTilde{m\kappa}{n\kappa}$ is described in \cref{fig:pictures_subspaces}. The subspace $\Ima(\Pi_{\Sym{(m+1)\kappa}}\cdot \Pi_{\GoodTilde{m\kappa}{n\kappa}})$, the image of $\GoodTilde{m\kappa}{n\kappa}$ under $\Pi_{\Sym{(m+1)\kappa}}$, is also of great importance and its relation with the good and bad subspaces are also shown in the figure. The following few lemmas prove that this is indeed the case.
\begin{figure}[ht]
\includegraphics[scale=0.374]{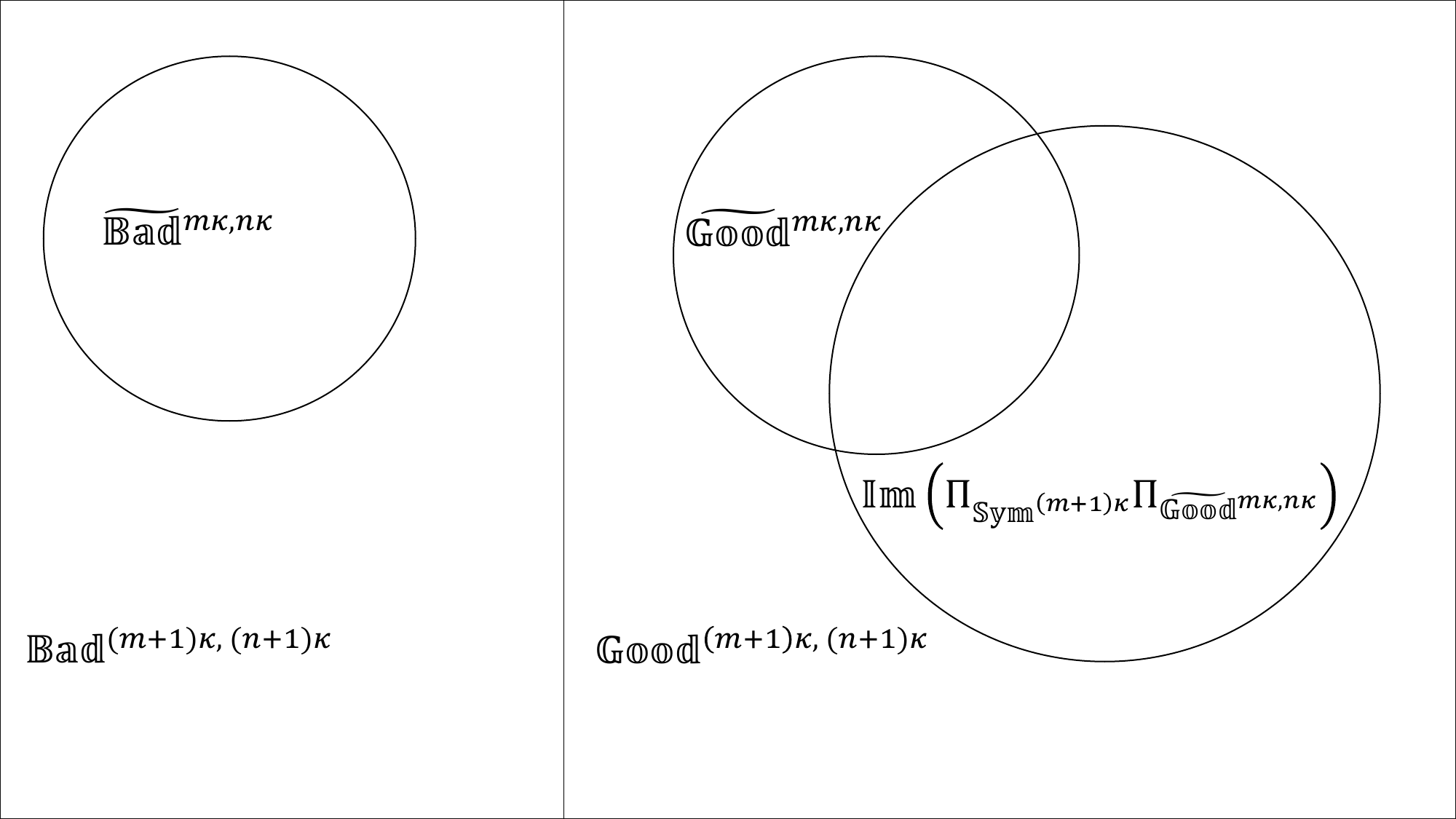}
\caption{In this figure, we see the relation between the different subspaces. The space $\mathbb{H}^{(m+1)\kappa}$ represented by the entire large rectangle is decomposed as the direct sum of the spaces $\Good{(m+1)\kappa}{(n+1)\kappa}$ and $\Bad{(m+1)\kappa}{(n+1)\kappa}$ represented by the left and right rectangles respectively. The subspace labeled $\Ima(\Pi_{\Sym{(m+1)\kappa}}\Pi_{\GoodTilde{m\kappa}{n\kappa}})$ in the figure, is the image of the operator $\Pi_{\Sym{(m+1)\kappa}}\Pi_{\GoodTilde{m\kappa}{n\kappa}}$.  %The straight line represents an orthonormal basis for $\mathbb{H}^{(m+1)\kappa}$ which can be decomposed into different subsets (shown in the picture using double arrow line segments) each spanning the corresponding subspace, it is labeled with.
}
\label{fig:pictures_subspaces}
\end{figure}

%Since applying $\prqc.\verify$ on any state $\ket{\beta}$ is the same as measuring $\ket{\beta}$

%Next we are going to see some of the properties of $\Bad{m\kappa}{n\kappa}$, the perpendicular subspace of $\Good{m\kappa}{n\kappa}$.

\begin{lemma}\label{lemma:ASubperp} 
$\Bad{m\kappa}{n\kappa}$ is the same as the subspace \[\{\ket{\psi} \in (\mathbb{H})^{m\kappa}| \Pr[\prqc.\Count(\sk, \ketbra{\psi}) > n\kappa] = 1\},\] where $\Good{m\kappa}{n\kappa}$ is as defined in \cref{eq:subspace definition}. Moreover, for any $m\kappa$ register state $\ket{\alpha} := a_1\ket{\alpha_1} + a_2\ket{\alpha_2}$ such that $\ket{\alpha_1} \in \Good{m\kappa}{n\kappa}$ and $\ket{\alpha_2} \in \Bad{m\kappa}{n\kappa}$, \begin{equation}\Pr[\prqc.\Count(\sk, \ket{\alpha}) > n\kappa] = |a_2|^2.\end{equation}
\end{lemma}

The proof is given on \cpageref{pf:lemma:ASubperp}.
\paragraph*{}

Let \nom{P}{$\Pi_{\BadTilde{m\kappa}{n\kappa}}$}{Projection on to $\BadTilde{m\kappa}{n\kappa}$} and  \nom{P}{$\Pi_{\GoodTilde{m\kappa}{n\kappa}}$}{projection on to $\GoodTilde{m\kappa}{n\kappa}$} denote the projection operators on to the subspaces $\BadTilde{m\kappa}{n\kappa}$ and $\GoodTilde{m\kappa}{n\kappa}$ respectively (see \cref{eq:subspace definition} for the definition of $\GoodTilde{m\kappa}{n\kappa}$). %Recall the projection on to the symmetric subspace $\Sym{m\kappa}$, $\Pi_{\Sym{m\kappa}}$ (see \cref{item:symmetric_sub_proj_def} in \cref{subsec:notations}). 
The following holds:
\begin{lemma}\label{lemma:Bad-Good}
For every $m,n \in \NN$ and $m> n$,
\[\Pi_{\BadTilde{m\kappa}{n\kappa}}\Pi_{\Sym{(m+1)\kappa}}\Pi_{\GoodTilde{m\kappa}{n\kappa}} = 0,\]
where $\Pi_{\Sym{(m+1)\kappa}}$ is the projection on to the symmetric subspace $\Sym{m\kappa}$ (see \cref{item:symmetric_sub_proj_def} in \cref{subsec:notations}).
\end{lemma}
The proof is given on \cpageref{pf:lemma:Bad-Good}.
It might seem that the operators $\Pi_{\Sym{(m+1)\kappa}}$ and $\Pi_{\GoodTilde{m\kappa}{n\kappa}}$ commute and since $\BadTilde{m\kappa}{n\kappa}$ and $\GoodTilde{m\kappa}{n\kappa}$ are orthogonal spaces, \cref{lemma:Bad-Good} follows. It is not hard to show that this is not the case and as shown in \cref{fig:pictures_subspaces}, \[\Ima(\Pi_{\Sym{(m+1)\kappa}}\Pi_{\GoodTilde{m\kappa}{n\kappa}}) \not\subset \GoodTilde{m\kappa}{n\kappa}.\] Hence, \[[\Pi_{\Sym{(m+1)\kappa}}, \Pi_{\GoodTilde{m\kappa}{n\kappa}}] \neq 0.\] 
 Therefore, in order to prove \cref{lemma:Bad-Good} we need a commutation property described in the next lemma. 

\paragraph*{}Let \nom{P}{$\Pi_{\Good{m\kappa}{n\kappa}}$}{projection on to $\Good{m\kappa}{n\kappa}$} be the projection operator on to $\Good{m\kappa}{n\kappa}$. The following holds.
\begin{lemma}\label{lemma:ASub-symmetric_invariant_properties}
For every $m,n \in \NN$ and $m>n$,
\[[\Pi_{\Good{m\kappa}{n\kappa}}, \Pi_{\Sym{m\kappa}}] = 0,\] where $\Pi_{\Sym{m\kappa}}$ is the projection onto the symmetric subspace over $m\kappa$ registers (see notations in \cref{subsec:notations}).
%The projection operators $\Pi_{\Good{m\kappa}{n\kappa}}$ and $\Pi_{\Sym{m\kappa}}$ commute where $\Pi_{\Good{m\kappa}{n\kappa}}$ and $\Pi_{\Sym{m\kappa}}$ are the projection operators of the subspaces $\Good{m\kappa}{n\kappa}$ and $\Sym{m\kappa}$ respectively, i.e., $\Pi_{\Sym{m\kappa}}$ keeps the subspace $\Good{m\kappa}{n\kappa}$ invariant.
\end{lemma}
The proof is given on \cpageref{pf:lemma:ASub-symmetric_invariant_properties}.
%$A:= \{\ket{\cent}\tensor \ket{\beta}| \ket{\beta'} \in {\Good{}{}}\}$ and $B:= \vee^{(n+2)\kappa}(\mathbb{H})$ and 
\paragraph*{}From now onwards, we fix an arbitrary $m,n \in \NN$ and $m>n$. Recall $\Pi_{\GoodTilde{m\kappa}{n\kappa}}$, the projection operator on to $\GoodTilde{m\kappa}{n\kappa}$ (see \cref{eq:subspace definition} for the definition of $\GoodTilde{m\kappa}{n\kappa}$). Define, the operator \nom{P}{$\Poperator{m}{n}$}{$\Pi_{\GoodTilde{m\kappa}{n\kappa}} \Pi_{\Sym{(m+1)\kappa}} \Pi_{\GoodTilde{m\kappa}{n\kappa}}$} as follows:
\begin{equation}
\Poperator{m}{n}:=\Pi_{\GoodTilde{m\kappa}{n\kappa}} \Pi_{\Sym{(m+1)\kappa}} \Pi_{\GoodTilde{m\kappa}{n\kappa}},
\label{eq:P_defined_as_A_Sym_A}
\end{equation} where $\Pi_{\Sym{(m+1)\kappa}}$ is the projection onto the symmetric subspace over $(m+1)\kappa$ registers (see notations in \cref{subsec:notations}).

\begin{lemma}\label{lemma:restricted}
For every $m,n \in \NN$ and $m>n$, and for every $\ket{\beta} \in \Good{m\kappa}{n\kappa}$, 
\begin{equation}
\Pr[\pkqc.\Count_{\ket{\cent}}(\ket{\beta}) = m] \leq  \lambda_{\text{max}}(\Poperator{m}{n}),
\end{equation} where $\Poperator{m}{n}$ is as defined in \cref{eq:P_defined_as_A_Sym_A}. 
\end{lemma}
The proof is given on \cpageref{pf:lemma:restricted}.
The next lemma estimates the largest eigenvalue of $\Poperator{m}{n}$.
\begin{lemma}\label{lemma:eigenbasis}
For every $m,n \in \NN$ and $m>n$, \[\lambda_{\text{max}}(\Poperator{m}{n})= \frac{\binom{m\kappa}{n\kappa}}{\binom{(m+1)\kappa}{(n+1)\kappa}},\] where $\Poperator{m}{n}$ is as defined in \cref{eq:P_defined_as_A_Sym_A}.
\end{lemma}
Note the r.h.s in the lemma above is equal to the success probability of the attack described in \cref{alg:opt_attack}, see \cref{eq:probability_opt_attack} and the discussion below it regarding how this term should be interpreted (essentially, it can be close to $1$ even for $m=n+1$ and polynomial $n$, and converges to $1$ in the large $n$ limit). The proof is given on \cpageref{pf:lemma:eigenbasis}.

Next, we will see a proof of \cref{prop:opt_attack} using \cref{lemma:ASubperp,lemma:Bad-Good,lemma:ASub-symmetric_invariant_properties,lemma:restricted,lemma:eigenbasis}.
\begin{proof}[Proof of \cref{prop:opt_attack}]\label{pf:prop:opt_attack}

%We construct $\bdv$ as follows. On getting $n\kappa$ private coins as input, $\bdv$ arranges them into $n$ many $\kappa$ blocks to form $n$ correct public coins. It simulates $\adv$ on these $n$ public coins and collects the first $(n+1)$ public coins (in the same arrangement order as $\adv$ submits to the bank), which are equivalent to $(n+1)\kappa$ private coins (allegedly correct) and submits these $(n+1)\kappa$ private coins to the bank for verification. %Clearly, $\adv( \ket{\cent}^{\tensor m}) = \bdv( \ket{\mill}^{\tensor m\kappa}) .$
%%Fix an arbitrary adversary say $\adv$, which receives $n$ coins and submits $m$ public coins in the unforgeability game (Game~\ref{game:nonadapt_unforge}). Let us consider the case where $m = n+1$. We will show that the optimal success probability for $\adv$ (to pass all the $(n+1)$ coins) is achieved by the attack described in \cref{alg:opt_attack} on \cpageref{alg:opt_attack}.
We assume that a QPT adversary $\adv$ receives $n$ public coins from the bank. Of course, the bank generates these coins using $\pkqc.\bank$ in \cref{alg:ts} (recall that by construction, the state is the same as $n\kappa$ private coins). It submits $m$ alleged public coins, which is a $\kappa m$-register state which we denote be $\ket{\alpha}$ (it would become clear that the assumption that the submitted state is a pure state is WLOG later). %We assume $m,n \in \poly$.
Since the verifier's wallet is initialized with one fresh public coin, $\ket{\cent}$ (see $\pkqc.\Count$ in  \cref{line:Count} in  \cref{alg:ts}),  the total state of the wallet and the $m$ alleged new coins submitted by the adversary should be\[\ket{\widetilde{\alpha}} := \ket{\cent}\tensor \ket{\alpha}.\]
%%WLOG, assume that the total state of the  $(n+1)$ coins is pure. Call it $\ket{\alpha}$. %Hence, $ \bdv( \ket{\mill}^{\tensor n\kappa}) = \ket{\widetilde{\alpha}}$
%Therefore the success probability of the adversary is \\cent\Pr[\pkqc.\Count_{\ket{\cent}}(\ket{\alpha})) = n+1]$.

%Let the state of the wallet be written as $\widetilde{\rho} = \sum_j p_j \ket{P_j}\bra{P_j}$ and say, $\ket{P_j} = \sum_{1\leq i_1,i_2,\ldots i_{n\kappa} \leq d; \pi_1(t(i)) \geq \kappa}{c_i}^{j}\ket{\mill_1}\ldots \ket{\mill_{n\kappa}}$ (by lemma \eqref{W_State}).

%Hence,

Express $\ket{\alpha}$ as $(a_1 \ket{\alpha_1} + a_2\ket{\alpha_2})$ such that \[\ket{\alpha_1} \in \Good{m\kappa}{n\kappa}, \ket{\alpha_2} \in \Bad{m\kappa}{n\kappa}\text{ and }\sum_{i=1}^2 |a_i|^2 = 1.\]see \cref{eq:subspace definition} for the definition of $\Good{m\kappa}{n\kappa}$.
By \cref{lemma:ASubperp}, \[\Pr[\prqc.\Count(\sk, \ket{\alpha}) > n\kappa] = |a_2|^2.\]
Therefore by the nonadaptive unforgeability (see \cref{definition:adapt_flex_unforge}) of the underlying $\prqc$ scheme (the private coin scheme that we lift to $\pkqc$ in \cref{alg:ts}), there exists a negligible function $\negl$ such that
\begin{equation}\label{eq:forbidden_prob}
\Pr[\prqc.\Count(\sk, \ket{\alpha}) > n\kappa] = |a_2|^2 = \negl.  
\end{equation}
Note that if the underlying $\prqc$ scheme is \nauuf (see \cref{definition:adapt_flex_unforge} and \cref{definition:unconditional_unforgeability}), then \cref{eq:forbidden_prob} holds even if $\adv$ is computationally unbounded.
Let $\ket{\widetilde{\alpha}_1} := \ket{\cent}\tensor \ket{\alpha_1}$ and similarly define $\ket{\widetilde{\alpha}_2}$. 
Hence,
\begin{align}
\ket{\widetilde{\alpha}} = \ket{\cent}\tensor \ket{\alpha} = \ket{\cent}\tensor (a_1 \ket{\alpha_1} + a_2\ket{\alpha_2}) = a_1 \ket{\widetilde{\alpha}_1} + a_2\ket{\widetilde{\alpha}_2}.
\end{align}
%We have written the state as a linear combination of two orthogonal states such that $\sum_{i} |a_i|^2 = 1$, $\ket{\alpha_2} \in ({\Good{}{}})$ and $\ket{\alpha_1} \in ({\Good{}{}})^{\perp}$.

%Hence by \cref{eq:ASub_invariance,eq:ASubperp_invariance},
By definition, 
\begin{equation}\label{eq:dollarelements}
\ket{\widetilde{\alpha}_1} \in \GoodTilde{m\kappa}{n\kappa}, \ket{\widetilde{\alpha}_2} \in \BadTilde{m\kappa}{n\kappa}.\end{equation}
Therefore,
\begin{align}
    \Pi_{\BadTilde{m\kappa}{n\kappa}}&(\Pi_{\Sym{m\kappa}}\ket{\widetilde{\alpha}_1})\\
    &= \Pi_{\BadTilde{m\kappa}{n\kappa}}\Pi_{\Sym{m\kappa}}\Pi_{\GoodTilde{m\kappa}{n\kappa}}\ket{\widetilde{\alpha}_1}\\
    &= 0. &\text{By \cref{lemma:Bad-Good}}
\end{align}
Hence, \[\Pi_{\Sym{m\kappa}}\ket{\widetilde{\alpha}_1} \in (\BadTilde{m\kappa}{n\kappa})^\perp.\]
Since $\ket{\widetilde{\alpha}_2} \in \BadTilde{m\kappa}{n\kappa}$ (see \cref{eq:dollarelements}), the states $\Pi_{\Sym{m\kappa}}\ket{\widetilde{\alpha}_1}$ and $\ket{\widetilde{\alpha}_2}$ are mutually orthogonal and hence, the following holds:

\begin{equation}\label{eq:cross_terms_vanish}
     \tr(\Pi_{\Sym{(m+1)\kappa}}\ket{\widetilde{\alpha}_2}\bra{\widetilde{\alpha_1}}) = \overline{\tr(\Pi_{\Sym{(m+1)\kappa}}\ket{\widetilde{\alpha}_1}\bra{\widetilde{\alpha_2}})}=\overline{\bra{\widetilde{\alpha_2}}\Pi_{\Sym{(m+1)\kappa}}\ket{\widetilde{\alpha}_1}}=0.
\end{equation}

The symmetric subspace over $(m+1)\kappa$ registers is the subspace over all $(m+1)\kappa$-register pure states which are invariant under any permutation of the registers. Clearly, any state in the symmetric subspace over $(m+1)\kappa$ register must remain invariant under an arbitrary permutation of the last $m\kappa$ registers (keeping the first $\kappa$ registers intact) since any permutation on the last $m\kappa$ registers is also a permutation of the entire $(m+1)\kappa$ (which does nothing to the first $\kappa$ registers). Therefore, the symmetric subspace over $(m+1)\kappa$ register is a subspace of the symmetric subspace over $m\kappa$ registers, i.e.,
\begin{equation}\label{eq:symsub_inclusion} \Sym{(m+1)\kappa} \subset \mathbb{H}^{\tensor{\kappa}} \tensor \Sym{m\kappa}.\end{equation}
Hence, for any state $\ket{\psi}$,
\begin{align}\label{eq:trace_collapse}
 \Pr&[\pkqc.\Count_{\ket{\cent}}(\ket{\psi})) = m]\\
&= \tr\left(\Pi_{\Sym{(m+1)\kappa}} (\Pi_{\Sym{m\kappa}}\tensor I_{\kappa})\cdots (\Pi_{\Sym{2\kappa}} \tensor I_{(m-1)\kappa}) \ketbra{\cent \tensor \psi}\right)\\
&=  \tr(\Pi_{\Sym{(m+1)\kappa}} \ketbra{\cent \tensor \psi}) . 
\end{align}

Therefore, the success probability of $\adv$, i.e., the probability that all $m$ coins pass verification is:

\begin{align}
%&  \Pr[\pkqc.\Count_{\ket{\cent}}(\adv(\ket{\cent}^{\tensor n})) >n]]\\
 \Pr&[\pkqc.\Count_{\ket{\cent}}(\ket{\alpha})) = m]\\
&=  \tr(\Pi_{\Sym{(m+1)\kappa}} \ket{\widetilde{\alpha}}\bra{\widetilde{\alpha}}) &\text{By \cref{eq:trace_collapse}}\\
& = \tr(\Pi_{\Sym{(m+1)\kappa}} (|a_1|^2\ket{\widetilde{\alpha}_1}\bra{\widetilde{\alpha}_1}) +  |a_2|^2\ket{\widetilde{\alpha}_2}\bra{\widetilde{\alpha}_2})\\
+ &\Pi_{\Sym{(m+1)\kappa}} (a_1 \bar{a_2}\ket{\widetilde{\alpha}_1}\bra{\widetilde{\alpha}_2}) +  a_2 \bar{a_1}\ket{\widetilde{\alpha}_2}\bra{\widetilde{\alpha}_2}))\\
&= |a_1|^2 \tr(\Pi_{\Sym{(m+1)\kappa}} (\ket{\widetilde{\alpha}_1}\bra{\widetilde{\alpha}_1})\\
 + &|a_2|^2\tr(\Pi_{\Sym{(m+1)\kappa}} (\ket{\widetilde{\alpha}_2}\bra{\widetilde{\alpha}_2}) + 0 &\text{By \cref{eq:cross_terms_vanish}}\\
&\leq \tr(\Pi_{\Sym{(m+1)\kappa}} (\ket{\widetilde{\alpha}_1}\bra{\widetilde{\alpha}_1}) + |a_2|^2\\
&= \Pr[\pkqc.\Count_{\ket{\cent}}(\ket{\alpha_1})) = m] + |a_2|^2  &\text{By \cref{eq:trace_collapse}}\\ 
&= \Pr[\pkqc.\Count_{\ket{\cent}}(\ket{\alpha_1})) = m] + \negl &\text{By \cref{eq:forbidden_prob}}\\
&\leq \lambda_{\text{max}}(\Poperator{m}{n}) + \negl &\text{(By \cref{lemma:restricted}}\\
&&\text{since $\ket{\alpha_2} \in {\Good{m\kappa}{n\kappa}}$)}\\
&= \frac{\binom{m\kappa}{n\kappa}}{\binom{(m+1)\kappa}{(n+1)\kappa}} + \negl &\text{By \cref{lemma:eigenbasis}}\\
&= \Pr[\pkqc.\Count_{\ket{\cent}}(\ket{\BasisSym{m\kappa}_{(n\kappa, (m-n)\kappa, 0\ldots,0)}}) = m] \\
&+ \negl, &\text{By \cref{eq:probability_opt_attack}}
\end{align}
where $\negl$ is the negligible function, used in \cref{eq:forbidden_prob}.

%Since the symmetric subspace over $k$ registers is a proper subspace of the symmetric subspace over any $(k-1)$ subset of the $k$ registers. where $\negl$ is a negligible function.\\
 Note that $\ket{\BasisSym{m\kappa}_{(n\kappa, (m-n)\kappa, 0\ldots,0)}}$ is the same as the state submitted in \cref{alg:opt_attack} and hence, we are done with the proof for the case in which the adversary submits a pure state. The proof can be easily extended to the general case when $\adv$ submits a mixed state using a standard convexity argument. By the nonadaptive unforgeability (see \cref{definition:adapt_flex_unforge}) of the underlying $\prqc$ scheme (the private scheme that we lift to $\pkqc$ in \cref{alg:ts}), the ensemble submitted by the adversary must have overwhelming overlap with ${\Good{m\kappa}{n\kappa}}$. Note that, every ensemble or a mixed state is a convex combination of pure states. Therefore, up to some negligible correction, the optimal success probability for the submitted mixed state to pass verification for all $m$ coins, is bounded by bounded by $\lambda_{\text{max}}(\Poperator{m}{n})$, by \cref{lemma:restricted}). This along with  \cref{lemma:eigenbasis} concludes the proof for the first statement of the proposition.
  %We can again break the expression for the success probability against $\pkqc$ into two parts, one of which is exactly same as the success probability against $\prqc$ and hence negligible. The other part is less than equal to the

%Since $\adv$ is $QPA$, $n \in \poly[\kappa]$ and $\kappa \in \poly[\log(\lambda)]$ hence, $\frac{\kappa}{\binom{2\kappa}{ \kappa }} \in \negl$.
 %Verified using Mathematica that this is negligible (code in the source, commented out.)
%In[56]:= Limit[Binomial[Log[n]^d, (Log[n]^d)/2]/n^c, n -> Infinity, 
 %Assumptions -> {c > 1, d > 1}]
%Out[56]= \[Infinity]
%Amit:Is that all we need?}
%Therefore,
%\[\Pr[\prqc.\Count(\sk, \bdv(\ket{\mill}^{\tensor n\kappa}))) > n\kappa)] \geq  \Pr[\pkqc.\Count_{\ket{\cent}}(\adv(\ket{\cent}^{\tensor n})) >n]] - \negl,\]to conclude the proof.
%In fact, more generally if the adversary produces $m$ coins with combined $\ket{\beta}$ such that $\prqc.\Count(\sk, \ketbra{\beta}) = r'$ then $ \Pr[\pkqc.\Count_{\ket{\cent}}(\ket{\beta}) = m] \leq \frac{\binom{m\kappa}{{(r', m\kappa - r',0\ldots 0)}}}{\binom{(m+1)\kappa}{(r'+\kappa,m\kappa - r',0\ldots, 0)}}$. We call this probability ${\mu_{opt}}^{r'}$. (Note that ${\mu_{opt}}^{n\kappa} = \mu_{opt}$.). The 

We now prove the ``Moreover'' part of the proposition.
The only place in the proof where we might need computational assumptions on $\adv$ is \cref{eq:forbidden_prob} depending on whether the underlying $\prqc$ scheme is \nauf (see \cref{definition:adapt_flex_unforge}) only against QPT adversaries or computationally unbounded adversaries. Hence, if the underlying $\prqc$ scheme is \nauuf (see \cref{definition:adapt_flex_unforge} and \cref{definition:unconditional_unforgeability}) then the attack described in \cref{alg:opt_attack} will be optimal even for computationally unbounded nonadaptive adversaries as well, who get $n$ public coins from the mint and submit $m$ alleged public coins ($m,n \in \poly$ and $m>n$).

%\qed 
\end{proof}

% \begin{proof}\label{pf:prop:adapt_opt_attack}
% Suppose a QPT adversary $\adv$ receives $n+m'$ public coins from the bank, submits $\tilde{m}$ alleged public coins to private verification oracle in Game~\ref{game:nonadapt_unforge_strongest}, out of which $m'$ are accepted. Of course, the bank generates these coins using $\pkqc.\bank$ in \cref{alg:ts} (recall that by construction, the state is the same as $n\kappa$ private coins). It submits $m$ alleged public coins, which is a $\kappa m$-register state which we denote be $\ket{\alpha}$ (it would become clear that the assumption that the submitted state is a pure state is WLOG later). %We assume $m,n \in \poly$.
% The proof differs from the proof of \cref{prop:opt_attack} only in the cryptographic security guarantees
% \qed \end{proof}

Finally, we will see proofs of  \cref{lemma:ASubperp,lemma:Bad-Good,lemma:ASub-symmetric_invariant_properties,lemma:restricted,lemma:eigenbasis} which completes the proof of  \cref{prop:opt_attack} and our discussion regarding optimal $n$ to $m$ nonadaptive attacks ($m,n \in \poly$ and $m > n$) on the scheme described in  \cref{alg:ts}.

%We will now see proofs of \cref{lemma:restricted} and  \cref{lemma:eigenbasis}.
\begin{proof}[Proof of \cref{lemma:ASubperp}]\label{pf:lemma:ASubperp}
Since $\prqc.\verify$ is a projective measurement \[\{\ketbra{\mill}, I - \ketbra{\mill}\},\] and $\Good{m\kappa}{n\kappa}$ is as defined in \cref{eq:subspace definition}, 
\begin{equation}\label{eq:count-restricted-projections}
    \Pr[\prqc.\Count(\sk, \ketbra{\psi}) \leq n\kappa] = \bra{\psi}\Pi_{\Good{m\kappa}{n\kappa}}\ket{\psi}.
\end{equation}
where \nom{P}{$\Pi_{\Good{m\kappa}{n\kappa}}$}{projection on to the subspace $\Good{m\kappa}{n\kappa}$} is the projection on to the subspace $\Good{m\kappa}{n\kappa}$.
Hence,
\begin{align}\label{eq:count-forbidden-projections}
    \Pr&[\prqc.\Count(\sk, \ketbra{\psi}) > n\kappa]\\ &= 1 - \Pr[\prqc.\Count(\sk, \ketbra{\psi}) \leq n\kappa]\\ &= \bra{\psi}(I - \Pi_{\Good{m\kappa}{n\kappa}})\ket{\psi}\\ &= \bra{\psi}\Pi_{(\Good{m\kappa}{n\kappa})^{\perp}}\ket{\psi}\\ &= \bra{\psi}\Pi_{\Bad{m\kappa}{n\kappa}}\ket{\psi}.
\end{align}
where \nom{P}{$\Pi_{\Bad{m\kappa}{n\kappa}}$}{projection on to the subspace $\Bad{m\kappa}{n\kappa}$} is the projection on to the subspace $\Bad{m\kappa}{n\kappa}$. Hence, for any $\ket{\psi} \in (\mathbb{H})^{m\kappa}$, 
\begin{equation}
    \ket{\psi} \in \Bad{m\kappa}{n\kappa} \iff \Pr[\prqc.\Count(\sk, \ketbra{\psi}) > n\kappa] = 1.
\end{equation}
Therefore, \[\Bad{m\kappa}{n\kappa} = \{ \ket{\psi} \in (\mathbb{H})^{m\kappa}| \Pr[\prqc.\Count(\sk, \ketbra{\psi}) > n\kappa] = 1 \}.\]
Moreover, for any $m\kappa$ register state $\ket{\alpha} := a_1\ket{\alpha_1} + a_2\ket{\alpha_2}$ such that $\ket{\alpha_1} \in \Good{m\kappa}{n\kappa}$ and $\ket{\alpha_2} \in \Bad{m\kappa}{n\kappa}$,
\begin{align}
    \Pr[\prqc.\Count(\sk, \ket{\alpha}) > n\kappa] &= \bra{\alpha}\Pi_{\Bad{m\kappa}{n\kappa}}\ket{\alpha}& \text{(By \cref{eq:count-forbidden-projections})}\\
    &= |a_2|^2. 
\end{align}
%\qed 
\end{proof}

\begin{proof}[Proof of \cref{lemma:Bad-Good}]\label{pf:lemma:Bad-Good}
By \cref{eq:dollarspacesubspacelargerspace} for every $m,n \in \NN$ and $m>n$,
\begin{equation}\label{eq:good-projec-relations}
\Pi_{\GoodTilde{m\kappa}{n\kappa}}=\Pi_{\Good{(m+1)\kappa}{(n+1)\kappa}}\Pi_{\GoodTilde{m\kappa}{n\kappa}}.
\end{equation}
By \cref{lemma:ASubperp}, we know that 
\[\Bad{m\kappa}{n\kappa} = \{\ket{\psi} \in (\mathbb{H})^{m\kappa}| \Pr[\prqc.\Count(\sk, \ketbra{\psi}) > n\kappa] = 1\}.\] Therefore by the definition of $\BadTilde{m\kappa}{n\kappa}$,
\[\BadTilde{m\kappa}{n\kappa} \subset \{\ket{\psi} \in (\mathbb{H})^{(m+1)\kappa}| \Pr[\prqc.\Count(\sk, \ketbra{\psi}) > (n+1)\kappa] = 1\}.\]
Hence, by \cref{lemma:ASubperp},
\begin{equation}\label{eq:dollar_perp_subspaces}
    \BadTilde{m\kappa}{n\kappa} \subset \Bad{(m+1)\kappa}{(n+1)\kappa}.
\end{equation}
Therefore,
\begin{equation}\label{eq:bad-projec-relations}
    \Pi_{\BadTilde{m\kappa}{n\kappa}}\Pi_{\Good{(m+1)\kappa}{(n+1)\kappa}}= \Pi_{\BadTilde{m\kappa}{n\kappa}}\Pi_{(\Bad{(m+1)\kappa}{(n+1)\kappa})^{\perp}} = 0.
\end{equation}
The rest of the proof follows by combining \cref{eq:good-projec-relations,eq:bad-projec-relations} with \cref{lemma:ASub-symmetric_invariant_properties}.
\begin{align}
    &\Pi_{\BadTilde{m\kappa}{n\kappa}}\Pi_{\Sym{(m+1)\kappa}}\Pi_{\GoodTilde{m\kappa}{n\kappa}}\\
    &=\Pi_{\BadTilde{m\kappa}{n\kappa}}\Pi_{\Sym{(m+1)\kappa}}\Pi_{\Good{(m+1)\kappa}{(n+1)\kappa}}\Pi_{\GoodTilde{m\kappa}{n\kappa}}&\text{By \cref{eq:good-projec-relations}}\\
    &=\Pi_{\BadTilde{m\kappa}{n\kappa}}\Pi_{\Good{(m+1)\kappa}{(n+1)\kappa}}\Pi_{\Sym{(m+1)\kappa}}\Pi_{\GoodTilde{m\kappa}{n\kappa}}&\text{By \cref{lemma:ASub-symmetric_invariant_properties}}\\
    &=0.&\text{By \cref{eq:bad-projec-relations}}
\end{align}
%\qed 
\end{proof}
\begin{proof}[Proof of \cref{lemma:ASub-symmetric_invariant_properties}]\label{pf:lemma:ASub-symmetric_invariant_properties}

For every $m$, recall the basis $\BasisSym{m\kappa}$ (see \cref{eq:basis_def} in \cref{subsec:notations}) for the symmetric subspace $\Sym{m\kappa}$. Therefore,
\begin{equation}\label{eq:projection-sym-basis-definition}
    \Pi_{\Sym{m\kappa}} = \sum_{\vec{j} \in \mathcal{I}_{d, m\kappa}} \ketbra{\BasisSym{m\kappa}_{\vec{j}}}.
\end{equation}
Recall \cref{eq:basis_states_def}, 
\[\ket{\BasisSym{m\kappa}}_{\vec{j}} = \frac{1}{\sqrt{\binom{m\kappa}{\vec{j}}}}\sum_{\vec{i}: T(\vec{i}) = \vec{j}} \ket{\phi_{i_1}\ldots\phi_{i_{m\kappa}}}.\]
Moreover the set, \[\{\bigotimes_{k = 1}^{m\kappa} \ket{\phi_{i_k}}\}_{(i_1,\ldots i_{m\kappa}) \in (\ZZ_d)^{m\kappa}, (T(\vec{i}))_{0} \leq n\kappa },\] forms as an orthonormal basis for $\Good{m\kappa}{n\kappa}$ for every $m,n \NN$ and $m>n$ (see \cref{item:product_basis_definition} in \cref{subsec:notations}).  Hence,
\begin{align}\label{eq:projection-good-basis-definition}
    \Pi_{\Good{m\kappa}{n\kappa}} = \sum_{\substack{\vec{i} \in (\ZZ_d)^{m\kappa}\\ (T(\vec{i}))_{0} \leq n\kappa}}\bigotimes_{k = 1}^{m\kappa} \ketbra{\phi_{i_k}}.\\
\end{align}
Therefore,
\begin{align}
    &\Pi_{\Sym{m\kappa}}\Pi_{\Good{m\kappa}{n\kappa}}\\
    &=\Pi_{\Sym{m\kappa}} \left(\sum_{\substack{\vec{i} \in (\ZZ_d)^{m\kappa}\\ (T(\vec{i}))_{0} \leq n\kappa}}\bigotimes_{k = 1}^{m\kappa} \ketbra{\phi_{i_k}}\right)\\
    &=\frac{1}{\sqrt{\binom{m\kappa}{T(\vec{i})}}}\sum_{\substack{\vec{i} \in (\ZZ_d)^{m\kappa}\\ (T(\vec{i}))_{0} \leq n\kappa}}\ket{\BasisSym{m\kappa}_{T(\vec{i})}}\bra{\phi_{i_1}\ldots\phi_{i_{m\kappa}}}\\
    &=\sum_{\substack{\vec{i},\vec{r} \in (\ZZ_d)^{m\kappa}\\T(\vec{i})=T(\vec{r}) \\(T(\vec{i}))_{0}=(T(\vec{r}))_{0} \leq n\kappa}}\frac{1}{\binom{m\kappa}{T(\vec{i})}}\ket{\phi_{r_1}\ldots\phi_{r_{m\kappa}}}\bra{\phi_{i_1}\ldots\phi_{i_{m\kappa}}}\\
    &=\frac{1}{\sqrt{\binom{m\kappa}{T(\vec{r})}}}\sum_{\substack{\vec{r} \in (\ZZ_d)^{m\kappa}\\ (T(\vec{r}))_{0} \leq n\kappa}} \ket{\phi_{r_1}\ldots\phi_{r_{m\kappa}}}\bra{\BasisSym{m\kappa}_{T(\vec{r})}}\\
    &=\Pi_{\Good{m\kappa}{n\kappa}}\left(\sum_{\vec{j} \in \mathcal{I}_{d, m\kappa}} \ketbra{\BasisSym{m\kappa}_{\vec{j}}}\right)\\
    &=\Pi_{\Good{m\kappa}{n\kappa}}\Pi_{\Sym{m\kappa}}.
\end{align}
This concludes the proof.

\end{proof}

\begin{proof}[Proof of  \cref{lemma:restricted}]\label{pf:lemma:restricted}
Fix $m,n \in \NN$ such that $m>n$. Let the state submitted for $\pkqc.\Count_{\ket{\cent}}()$ operation be $\ket{\beta} \in \Good{m\kappa}{n\kappa}$. The state of the verifier's wallet along with the new registers just before the $\pkqc.\Count_{\ket{\cent}}()$ operation involving symmetric subspace measurement (see \cref{line:Count} in  \cref{alg:ts}) is $\ket{\widetilde{\beta}} := \ket{\cent}\tensor \ket{\beta} \in \GoodTilde{m\kappa}{n\kappa}$. 
Therefore,
\begin{align}
\Pr[\pkqc.\Count_{\ket{\cent}}(\ket{\beta}) = m] &= \bra{\widetilde{\beta}}\Pi_{\Sym{(m+1)\kappa}}\ket{\widetilde{\beta}}  \\
&= \bra{\widetilde{\beta}}\Pi_{\GoodTilde{m\kappa}{n\kappa}}^{\dagger} \Pi_{\Sym{(m+1)\kappa}} \Pi_{\GoodTilde{m\kappa}{n\kappa}} \ket{\widetilde{\beta}} & \text{ (since $\ket{\beta} \in \GoodTilde{m\kappa}{n\kappa}$)} \\
&= \bra{\widetilde{\beta}}\Pi_{\GoodTilde{m\kappa}{n\kappa}} \Pi_{\Sym{(m+1)\kappa}} \Pi_{\GoodTilde{m\kappa}{n\kappa}} \ket{\widetilde{\beta}} \\
&\leq \lambda_{\text{max}}(\Poperator{m}{n}). & \text{(see  \cref{eq:P_defined_as_A_Sym_A})}
\end{align}
%\qed 
\end{proof}

\begin{proof}[Proof of  \cref{lemma:eigenbasis}]\label{pf:lemma:eigenbasis}
Fix $m,n \in \NN$ such that $m>n$. 
In order to estimate $\lambda_{\text{max}}(\Poperator{m}{n})$ we will find a set of orthonormal eigenvectors with non-zero eigenvalues, such that the vectors span $\ker(\Poperator{m}{n})^{\perp}$, the subspace where $\Poperator{m}{n}$ acts non-trivially (where $\Poperator{m}{n}$ is as defined in \cref{eq:P_defined_as_A_Sym_A}).

Let the sets \nom{G}{$\BasisGoodSym{m\kappa}{n\kappa}$}{vectors in $\BasisSym{m\kappa}$, which are in the subspace $\Good{m\kappa}{n\kappa}$} and \nom{G}{$\BasisGoodSymTilde{m\kappa}{n\kappa}$}{set of $(m+1)\kappa$ register states obtained by tensoring $\ket{\cent}$ with every vector in $\BasisGoodSym{m\kappa}{n\kappa}$} be defined as: 
\begin{align}\label{eq:eigenbasis_def}
\BasisGoodSym{m\kappa}{n\kappa} &:= \{\ket{\BasisSym{m\kappa}_{\vec{j}}}\}_{\vec{j} \in \mathcal{I}_{d, m\kappa}: j_0 \leq n\kappa}, \\ %(see \cref{subsec:notations}).
\BasisGoodSymTilde{m\kappa}{n\kappa} &:= \{\ket{\BasisSymTilde{m\kappa}_{\vec{j}}}\}_{\vec{j} \in \mathcal{I}_{d, m\kappa}: j_0 \leq n\kappa}.\\
\end{align}
 For the definitions of $\mathcal{I}_{d, m\kappa}$, $\ket{\BasisSym{m\kappa}_{\vec{j}}}$ and $\ket{\BasisSymTilde{m\kappa}_{\vec{j}}}$, see Notations   \cref{eq:basis_states_def} and  \cref{eq:basis_def} in \cref{subsec:notations}.
 Clearly, $\BasisGoodSym{m\kappa}{n\kappa}$ is a subset of the basis, $\BasisSym{m\kappa}$ (see \cref{subsec:notations}) and hence, is an orthogonal set. Therefore, $\BasisGoodSymTilde{m\kappa}{n\kappa}$ is also an orthogonal set.

We will prove the lemma by proving these parts:
\begin{enumerate} 
    \item \label{item:W_n_spans_ker_P_perp}$\BasisGoodSymTilde{m\kappa}{n\kappa}$ spans $\ker(\Poperator{m}{n})^\perp$. Hence, $Span(\BasisGoodSymTilde{m\kappa}{n\kappa})^{\perp} \subset \ker(\Poperator{m}{n}).$
    \item \label{item:W_n_eigenvectors_P}$\BasisGoodSymTilde{m\kappa}{n\kappa}$ is a set of orthogonal eigenvectors of $\Poperator{m}{n}$ with positive eigenvalues.
    \item \label{item:lambda_max_P}$\ket{\BasisSymTilde{m\kappa}_{(n\kappa, (m-n)\kappa, 0\ldots,0)}}\in \BasisGoodSymTilde{m\kappa}{n\kappa}$ has the largest eigenvalue, $\lambda_{\text{max}}(\Poperator{m}{n})$, which is equal to $\frac{\binom{m\kappa}{n\kappa}}{\binom{(m+1)\kappa}{(n+1)\kappa}}$. 
\end{enumerate}
Note that \cref{item:lambda_max_P} proves the lemma. 

\paragraph*{\cref{item:W_n_spans_ker_P_perp}:}  Observe that, for every $\ket{\BasisSym{m\kappa}_{\vec{j}}} \in \BasisSym{m\kappa}$, \[\Pr[\prqc.\Count(\sk, \ket{\BasisSym{m\kappa}_{\vec{j}}}) = j_0] = 1.\] Hence, by definition, $\BasisGoodSym{m\kappa}{n\kappa}$ (see definition in \cref{eq:eigenbasis_def}) is the subset of those vectors from the basis $\BasisSym{m\kappa}$, which are in $\Good{m\kappa}{n\kappa}$ (see definition of $\Good{m\kappa}{n\kappa}$ in \cref{eq:subspace definition} and $\BasisGoodSym{m\kappa}{n\kappa}$ in \cref{eq:eigenbasis_def}). Moreover, for every $\ket{\BasisSym{m\kappa}_{\vec{j}}} \in \BasisSym{m\kappa} \setminus \BasisGoodSym{m\kappa}{n\kappa}$, 
\[\ket{\BasisSym{m\kappa}_{\vec{j}}} \in \{\ket{\psi} \in \mathbb{H}^{\tensor m\kappa} ~|~ \Pr[\prqc.\Count(\sk, \ket{\psi}) >n\kappa] = 1\}.\]
Therefore by \cref{lemma:ASubperp},
\[\ket{\BasisSym{m\kappa}_{\vec{j}}} \in \Bad{m\kappa}{n\kappa} = (\Good{m\kappa}{n\kappa})^{\perp}.\]
Therefore $\BasisGoodSym{m\kappa}{n\kappa}$ forms an orthonormal basis for $\Good{m\kappa}{n\kappa} \bigcap \Sym{m\kappa} =:$ \nom{G}{$\GoodSym{m\kappa}{n\kappa}$}{intersection of the symmetric subspace $\Sym{m\kappa}$ with ${\Good{m\kappa}{n\kappa}}$}. Hence, $\BasisGoodSymTilde{m\kappa}{n\kappa}$ forms an orthonormal basis for the subspace \nom{G}{$\GoodSymTilde{m\kappa}{n\kappa}$}{subspace of $(m+1)\kappa$ register states such that the quantum state of the first $\kappa$ registers is $\ket{\cent}$ and the state of the rest $m\kappa$ register is a vector in $\GoodSym{m\kappa}{n\kappa}$} defined as 
\begin{align}\label{eq:SymASub_def}
\GoodSymTilde{m\kappa}{n\kappa} &:= \{\ket{\cent}\tensor\ket{\beta}~|~\ket{\beta} \in \GoodSym{m\kappa}{n\kappa}\}\\ 
&= \GoodTilde{m\kappa}{n\kappa} \bigcap \left(\mathbb{H}^{\tensor \kappa} \tensor \Sym{m\kappa}\right),
\end{align} 
where $\GoodTilde{m\kappa}{n\kappa}$ is as defined in \cref{eq:subspace definition}. Essentially, $\GoodSym{m\kappa}{n\kappa}$ is the subspace of all symmetric states in $\Good{m\kappa}{n\kappa}$ (see the discussion below \cref{eq:subspace definition} to interpret $\Good{m\kappa}{n\kappa}$) and $\GoodSymTilde{m\kappa}{n\kappa}$ is the subspace consisting of all $(m+1)\kappa$ states such that the quantum state of the first $\kappa$ register is $\ket{\cent}$ and the state of the rest $m\kappa$ registers is a symmetric state in $\Good{m\kappa}{n\kappa}$. Therefore, it is enough (for \cref{item:W_n_spans_ker_P_perp}) to show that, \[\ker{\Poperator{m}{n}}^{\perp} \subset \GoodSymTilde{m\kappa}{n\kappa}.\] %the symmetric subspace of $A$, $A\bigcap (\mathbb{H}^{\tensor \kappa}\vee^{(n+1)\kappa}(\mathbb{H}))$. 

%The symmetric subspace over $(m+1)\kappa$ registers is essentially the subspace over all $(m+1)\kappa$-register pure states which are invariant under any permutation of the registers. Clearly, any state in the symmetric subspace over $(m+1)\kappa$ register must remain invariant under an arbitrary permutation of the last $m\kappa$ registers (keeping the first $\kappa$ registers intact) since any permutation on the last $m\kappa$ registers is also a permutation of the entire $(m+1)\kappa$ (which does nothing to the first $\kappa$ registers). 
As discussed earlier in the proof of \cref{prop:opt_attack} (see \cref{eq:symsub_inclusion}), the symmetric subspace over $(m+1)\kappa$ register is a subspace of the symmetric subspace over $m\kappa$ registers, i.e.,
\[ \Sym{(m+1)\kappa} \subset \mathbb{H}^{\tensor{\kappa}} \tensor (\Sym{m\kappa}.\] Hence, 
\begin{equation}\label{eq:Poperator_rewrite}
\Pi_{\Sym{(m+1)\kappa}}\Pi_{\mathbb{H}^{\tensor\kappa}\tensor\Sym{m\kappa}} = \Pi_{\Sym{(m+1)\kappa}},
\end{equation} where \nom{P}{$\Pi_{\mathbb{H}^{\tensor\kappa}\tensor\Sym{m\kappa}}$}{projection on to $\mathbb{H}^{\tensor\kappa}\tensor\Sym{m\kappa}$} is the projection on to the subspace, $\mathbb{H}^{\tensor\kappa}\tensor\Sym{m\kappa}$. Note that the following commutation property holds, 
\begin{align}
&\Pi_{\mathbb{H}^{\tensor\kappa}\tensor\Sym{m\kappa}}\cdot \Pi_{\GoodTilde{m\kappa}{n\kappa}}\\
&= \mathbb{I}_\kappa\tensor\Pi_{\Sym{m\kappa}}\cdot (\ketbra{\cent}\tensor\Pi_{\Good{m\kappa}{n\kappa}})&\text{by definition of $\GoodTilde{m\kappa}{n\kappa}$, see \cref{eq:subspace definition}}\\
&= (\mathbb{I}_\kappa \cdot \ketbra{\cent}) \tensor (\Pi_{\Sym{m\kappa}}\cdot
\Pi_{\Good{m\kappa}{n\kappa}})\\
&= (\ketbra{\cent}\cdot \mathbb{I}_\kappa) \tensor (\Pi_{\Good{m\kappa}{n\kappa}}\cdot \Pi_{\Sym{m\kappa}})&\text{by \cref{lemma:ASub-symmetric_invariant_properties}}\\
&= \Pi_{\GoodTilde{m\kappa}{n\kappa}}\cdot \Pi_{\mathbb{H}^{\tensor\kappa}\tensor\Sym{m\kappa}}
    \label{eq:second_commutation}
\end{align}
Therefore,
\begin{align}
    \Poperator{m}{n}&=\Pi_{\GoodTilde{m\kappa}{n\kappa}}\Pi_{\Sym{(m+1)\kappa}}\Pi_{\GoodTilde{m\kappa}{n\kappa}}\\
    &=\Pi_{\GoodTilde{m\kappa}{n\kappa}}\Pi_{\Sym{(m+1)\kappa}}\Pi_{\mathbb{H}^{\tensor\kappa}\tensor\Sym{m\kappa}}\Pi_{\GoodTilde{m\kappa}{n\kappa}}&\text{by \cref{eq:Poperator_rewrite}}\\
    &=\Pi_{\GoodTilde{m\kappa}{n\kappa}}\Pi_{\Sym{(m+1)\kappa}}\Pi_{\GoodTilde{m\kappa}{n\kappa}}\Pi_{\mathbb{H}^{\tensor\kappa}\tensor\Sym{m\kappa}}.&\text{by \cref{eq:second_commutation}}.\\
\end{align}Hence, $(\mathbb{H}^{\tensor\kappa}\tensor\Sym{m\kappa})^\perp \subset \ker(\Poperator{m}{n})$ which is equivalent to \[\ker(\Poperator{m}{n})^\perp \subset \mathbb{H}^{\tensor\kappa}\tensor\Sym{m\kappa}.\] Similarly, since, $\Poperator{m}{n}=\Pi_{\GoodTilde{m\kappa}{n\kappa}}\Pi_{\Sym{(m+1)\kappa}}\Pi_{\GoodTilde{m\kappa}{n\kappa}}$, $\ker(\Poperator{m}{n})^\perp \subset \GoodTilde{m\kappa}{n\kappa}.$
Therefore, by the above two arguments,
\begin{align}
    \ker(\Poperator{m}{n})^\perp &\subset {\GoodTilde{m\kappa}{n\kappa}} \bigcap \left(\mathbb{H}^{\tensor \kappa}\tensor\Sym{m\kappa}\right)\\
    &= \GoodSymTilde{m\kappa}{n\kappa}. & \text{by \cref{eq:SymASub_def}}
\end{align}
\paragraph*{\cref{item:W_n_eigenvectors_P}:}Now, we will show that $\BasisGoodSymTilde{m\kappa}{n\kappa}$ is a set of orthogonal eigenvectors for $\Poperator{m}{n}$ with positive eigenvalues. \[\forall \vec{j} \in \mathcal{I}_{d, m\kappa}:  j_0 \leq n\kappa,\]  

\begin{align}\label{eq:eigen_vec}
    \Poperator{m}{n}\ket{\BasisSymTilde{m\kappa}_{\vec{j}}} &= \Pi_{\GoodTilde{m\kappa}{n\kappa}} \Pi_{\Sym{(m+1)\kappa}} \Pi_{\GoodTilde{m\kappa}{n\kappa}} \ket{\BasisSymTilde{m\kappa}_{\vec{j}}} \\
    &= \Pi_{\GoodTilde{m\kappa}{n\kappa}}\Pi_{\Sym{(m+1)\kappa}}\ket{\BasisSymTilde{m\kappa}_{\vec{j}}} \\
    &= \sqrt{\frac{\binom{m\kappa}{\vec{j}}}{\binom{(m+1)\kappa}{(j_0 + \kappa, j_1,\ldots, j_{d-1})}}}\Pi_{\GoodTilde{m\kappa}{n\kappa}}\ket{\BasisSym{(m+1)\kappa}_{(j_0 +\kappa, j_1\ldots j_{d-1})}} \\
    &= \frac{\binom{m\kappa}{\vec{j}}}{\binom{(m+1)\kappa}{(j_0 + \kappa,j_1,\ldots, j_{d-1})}}\ket{\BasisSymTilde{m\kappa}_{\vec{j}}}.\\
\end{align}
Therefore, $\ket{\BasisSymTilde{m\kappa}_{(j_0\ldots j_{d-1})}}$ is an eigenvector with eigenvalue
\[  \frac{\binom{m\kappa}{\vec{j}}}{\binom{(m+1)\kappa}{(j_0 + \kappa,j_1,\ldots, j_{d-1})}} \quad \forall \vec{j} \in \mathcal{I}_{d, m\kappa}:  j_0 \leq n\kappa.\]

\paragraph*{\cref{item:lambda_max_P}:}  \cref{item:W_n_spans_ker_P_perp} shows that all non-zero eigenvalues of $\Poperator{m}{n}$ must be in $Span(\BasisGoodSymTilde{m\kappa}{n\kappa})$. Clearly, $\Poperator{m}{n}$ is positive semidefinite by definition. Therefore, $\lambda_{\text{max}}(\Poperator{m}{n})$ is attained in $\BasisGoodSymTilde{m\kappa}{n\kappa}$, i.e., $\lambda_{\text{max}}(\Poperator{m}{n})$ is the eigenvalue for some eigenvector in  $\BasisGoodSymTilde{m\kappa}{n\kappa}$. %A little analysis shows that for every $\vec{j} \in \mathcal{I}_{d, m\kappa}$ the term for the corresponding eigenvalue is independent of $j_1, j_2\ldots,j_{d-1}$ and is maximized for $j_0 = n\kappa$. 
By \cref{item:W_n_eigenvectors_P}, for every $\vec{j} \in \mathcal{I}_{d, m\kappa}$ the term for the corresponding eigenvalue is 
\begin{align}
&= \frac{\binom{m\kappa}{\vec{j}}}{\binom{(m+1)\kappa}{(j_0 + \kappa,j_1,\ldots, j_{d-1})}} \\
&= \frac{\binom{m\kappa}{j_0}\binom{m\kappa - j_0}{(j_1, j_2, \ldots j_{d-1})}}{\binom{(m+1)\kappa}{\kappa + j_0}\binom{m\kappa - j_0}{(j_1, j_2, \ldots j_{d-1})}}\\
&= \frac{\binom{m\kappa}{m\kappa - j_0}}{\binom{(m+1)\kappa}{m\kappa - j_0}}\\
%&= \frac{(m\kappa)!(j_0 + \kappa)!}{((m+1)\kappa)!j_0!}\\
%&= \frac{(j_0 + 1)\ldots (j_0 + \kappa)}{(m\kappa + 1)\ldots (m\kappa + \kappa)}\\
&= \prod_{r=1}^\kappa \frac{j_0 + r}{m\kappa + r}.\\
\end{align}
%A little analysis shows that
Hence, the term for the eigenvalue increases with increasing $j_0$ and is independent of $j_1, j_2, \ldots j_{d-1}$. Since, $j_0 \leq n\kappa$, the maximum value is attained by all the vectors $\vec{j}$ in $\mathcal{I}_{d, m\kappa}$ such that $j_0 = n\kappa$.
Therefore, $\ket{\BasisSymTilde{m\kappa}_{(n\kappa, (m-n)\kappa,0\ldots,0)}}$ is one of the largest eigenvectors. Hence, 
\begin{equation}\label{eq:max_eigenvalue}
 \lambda_{\text{max}}(\Poperator{m}{n}) = \frac{\binom{m\kappa}{{n\kappa}}}{\binom{(m+1)\kappa}{(n+1)\kappa}}, \quad \text{the eigenvalue for $ \ket{\BasisSymTilde{m\kappa}_{(n\kappa, (m-n)\kappa,0\ldots,0)}}$}.%\qedhere
\end{equation}
%\qed 
\end{proof}

%This also shows that $\ket{\BasisSymTilde{m\kappa}_{(n\kappa, (m-n)\kappa, 0\ldots,0)}}$ is one of the eigenvectors with largest eigenvalue. 
%Hence, $\ket{\BasisSym{m\kappa}_{(n\kappa, (m-n)\kappa, 0\ldots,0)}}$ is an optimal choice for an adversary $\adv$ when it is restricted to submit states from ${\Good{m\kappa}{n\kappa}}$.
%Hence, $\widetilde{Sym_{{\BasisGood{}{}}}}$ is an orthogonal set of non-trivial eigenvectors of $\Poperator{m}{n}$. Next, we show that $P\ket{v} = 0$  for every $\ket{v} \in Span(\widetilde{Sym_{{\BasisGood{}{}}}})^{\perp}$. 

%Note that $Sym_{{\BasisGood{}{}}}$ forms an orthogonal basis for $A\bigcap \vee^{(n+1)\kappa}(\mathbb{H})$. Let $V_{n, n+1}$ be an orthogonal basis for $A$ obtained by extending $Sym_{{\BasisGood{}{}}}$. 

\section{Security proofs}\label{subsec:completeness_security proofs} In this subsection we use \cref{prop:opt_attack} to prove our lifting result, \cref{prop: rational-unforge}.
%Next, we prove rational unforgeability for our scheme and also untraceability under \ro restrictions. %It should be noted that we assume that the private quantum coin scheme that we are trying to lift to a public scheme is complete.  Note that in most of our results, we assume that $\prqc. \bank$ is a projective, two-outcome measurement $\{\ket{\mill}\bra{\mill}, I - \ket{\mill}\bra{\mill} \}$, where $\ket{\mill}$ is one valid coin minted by the bank.

\begin{proof}[Proof of \cref{prop: rational-unforge}]\label{pf:prop: rational-unforge}
Let $\adv$ be an arbitrary nonadaptive adversary (QPT or computationally unbounded depending on the unforgeability guarantees of the underlying $\prqc$ scheme) in Game~\ref{game:nonadapt_unforge_strongest}. We denote $n$ to be the number of coins, the adversary $\adv$ receives from mint. Since $\adv$ is nonadaptive in Game~\ref{game:nonadapt_unforge_strongest}, either she submits alleged public coins for private verification or instead submits for public verification. 
Suppose, it chooses to submit $\tilde{m}$ coins for private verification, where $\tilde{m} \in \poly$. Let $\omega$ be the quantum state of the coins she submits. Therefore the all-or-nothing utility (see \cref{eq:all-or-nothing_nonadaptive_utility_def}) of the adversary $\adv$, \footnote{Since, we assume that the adversary only submits for private verification, the flexible nonadaptive utility, $\uflna(\uf^{\adv,\MS}_\secpar)$ (see \cref{eq:flex_nonadapt_utility_def}) is also the same as $\uanna(\uf^{\adv,\MS}_\secpar)$.}
\begin{equation}
 \uanna(\uf^{\adv,\MS}_\secpar)  = \pkqc.\Count_{bank}(\omega) - n.
\end{equation}
Therefore,
\begin{equation}\label{eq:pf_uanna}
\mathbb{E}(\uanna(\uf^{\adv,\MS}_\secpar)) = \mathbb{E}(\pkqc.\Count(\omega))-n.
\end{equation}
By definition of $\pkqc.\Count_{bank}$ given in \cref{line:count_bank} in \cref{alg:ts},
\begin{equation}
    \mathbb{E}(\pkqc.\Count_{bank}(\omega))=\frac{\mathbb{E}(\prqc.\Count(\sk,\omega)}{\kappa}.
\end{equation}
Substituting this in \cref{eq:pf_uanna}, we get,
\[\mathbb{E}(\uanna(\uf^{\adv,\MS}_\secpar)) = \frac{\mathbb{E}(\prqc.\Count(\sk, \omega))-n\kappa}{\kappa}.\]
Hence, it suffices to show that the numerator, $\mathbb{E}(\prqc.\Count(\sk, \omega))-n\kappa$ is bounded above by a negligible function. This follows from the unforgeability of the private scheme.
Since we assume the underlying scheme $\prqc$ is \nauf, and since $n$ public coins taken by $\adv$, is the same as $n\kappa$ private coins, there exists a negligible function $\negl$ such that,
\[\Pr[\prqc.\Count(\sk, \omega)-n\kappa>0] = \Pr[\prqc.\Count(\sk, \omega)>n\kappa] \leq \negl.\]
Hence, $\mathbb{E}(\prqc.\Count(\sk, \omega))-n\kappa$ is also bounded above by a negligible function.

\paragraph*{} Next, suppose the nonadaptive adversary $\adv$ submits $m$ alleged coins for public verification, after receiving $n$ public coins from the bank, such that $m,n\in \poly$.

% Let $\adv$ be an arbitrary adversary (QPT or computationally unbounded depending on the unforgeability guarantees of the underlying $\prqc$ scheme). As discussed in the unforgeability game (see Game~\ref{game:nonadapt_unforge}), we denote $n$ to be the number of coins, the adversary $\adv$ receives and $m$ be the number of coins it submits such that $m,n \in \poly$ and $m>n$.

According to the all-or-nothing nonadaptive utility definition (see \cref{eq:all-or-nothing_nonadaptive_utility_def}) with respect to Game~\ref{game:nonadapt_unforge_strongest}, either the public verification is successful and the verifier accepts all the $m$ coins, in which case the utility $U(\adv)$ is $(m-n)$. Otherwise, the verifier rejects and $\adv$'s utility is $-n$.

%For the sake of simplicity, let us first consider the case where $m = n+1$. We already saw in the proof of , in Eq.~\eqref{eq:forging_prob} that 
Since the underlying $\prqc$ scheme is \nauf (see \cref{definition:adapt_flex_unforge}), by \cref{prop:opt_attack}, the success probability of $\adv$, i.e., all the $m$ coins are accepted, is less than or equal to
\[\negl + \frac{\binom{m\kappa}{n\kappa}}{\binom{(m+1)\kappa}{(n+1)\kappa }},\] for some negligible function $\negl$. As discussed before, if the underlying $\prqc$ scheme is \nauuf (see \cref{definition:adapt_flex_unforge} and \cref{definition:unconditional_unforgeability}), then this holds for any computationally unbounded adversary. Therefore, the expected utility of the adversary $\adv$
\begin{align}\label{eq:rational_unforge_utility}
& \mathbb{E}(U(\adv))\\
&\leq \negl +  \frac{\binom{m\kappa}{n\kappa}}{\binom{(m+1)\kappa}{(n+1)\kappa }}.(m-n) + (1 - \frac{\binom{m\kappa}{n\kappa}}{\binom{(m+1)\kappa}{(n+1)\kappa }})(-n)\\
& = m\frac{\binom{m\kappa}{n\kappa}}{\binom{(m+1)\kappa}{(n+1)\kappa }} - n + \negl\\
%&= m\cdot\frac{m\kappa!(n\kappa+\kappa)!}{(m\kappa +\kappa)!n\kappa!} - n + \negl\\
%&=m\cdot\prod_{r=1}^\kappa\frac{n\kappa + r}{m\kappa + r} - n + \negl\\
&\leq m\left(\frac{n+1}{m+1}\right)^{\kappa} - n + \negl\\%&\text{since $\frac{n\kappa+r}{m\kappa +r} \leq \frac{n\kappa + \kappa}{m\kappa + \kappa}$ for $r\leq \kappa$.}\\
&= \frac{m(n+1)^\kappa - n(m+1)^\kappa}{(m+1)^\kappa} + \negl\\
%&= \frac{m\sum_{i=0}^\kappa \binom{\kappa}{i}n^i - n\sum_{i=0}^\kappa \binom{\kappa}{i} m^i}{(m+1)^\kappa} +\negl\\
%&= \frac{\sum_{i=0}^\kappa \binom{\kappa}{i}\cdot mn(n^{i-1} - m^{i-1})}{(m+1)^\kappa} + \negl\\
%&= \frac{(m-n) + \sum_{i=1}^\kappa \binom{\kappa}{i}\cdot mn(n^{i-1} - m^{i-1})}{(m+1)^\kappa} + \negl\\
%&\leq\frac{(m-n) + \sum_{i=1}^\kappa 0}{(m+1)^\kappa} + \negl&\text{since $n^j \leq m^j$ for $j\in \NN$.}\\ 
&\leq \frac{m-n}{(m+1)^\kappa} + \negl&\text{$m\geq n.$}\\
&\leq \frac{1}{(m+1)^{\kappa - 1}} + \negl\\
&\leq \frac{1}{2^{\kappa - 1}} + \negl&\text{since $m \geq 1$.}\\
&= \frac{1}{2^{\log^c(\secpar) - 1}}+ \negl&\text{since $\kappa=(\log(\secpar))^c$, $c>1$.}\\
&=\frac{2}{\secpar^{\log^{c-1}(\secpar)}}+ \negl\\
%&\leq\frac{2}{\secpar^{\log(\secpar)}}+ \negl\\ 
&= \text{negl}'(\secpar)&\text{since $c>1.$}\\
\end{align} %for some function $\epsilon(\secpar)$, which is a negligible function, %if $m,n \in \poly$
%for any arbitrary $m$, $n$, and our choice of $\kappa$.
%oldproof
%&\leq m\left(1 - \frac{\kappa}{d}\frac{m-n}{m+1}\right)^\kappa - n + \negl &\text{for some constant $d$.}\\
%&= m-n - (m-n)\frac{m\kappa}{d(m+1)} + \negl\\
%&= (m-n)\left(1 - \frac{m\kappa}{d(m+1)}\right) + \negl\\
%&\leq 0 +\negl& \text{eventually for $\kappa = \log(\secpar)^c$}\\
%In[11]:= Limit[-n + m ((1 + n)/(1 + m))^k, m -> Infinity, Assumptions -> {k > 1, m >= n >= 0}]
%Out[11]:= -n

\begin{comment}
For a fixed $n$ and $\kappa$ the utility is maximized when $m = n+1$. Moreover, for $m= n+1$, the utility is upper bounded by $(n+1)(\frac{n+1}{n+2})^\kappa - n.$ For $0\geq n\leq k$, this expression is a strictly decreasing function. For $n=0$ it is positive but negligible and $n=1$ onwards it is negative (assuming $\kappa$ is large enough).
Mathematica code:
In[7]:= Manipulate[Plot[-x + ((1 + x)/(2 + x))^k, {x, 0, 1000}], {k, 0, 200}].

Mathematica showed that the graph is strictly decreasing as a function of $n$ and a fixed $\kappa$.
This can also be realized as follows:  For $n \in O(\kappa^c)$, we can use binomial approximation of $(1+x)^\epsilon \approx 1 +\epsilon x$ when $|x| \leq 1$ and $|\epsilon x| <<< 1$. Using the approximation it is easy to see that the term is negative. For $n \in o(\kappa)$, the term is approximately $(n+1)e^{-kappa/n+2} - n$ which is smaller than $0$ for $n>1$ and sufficiently large $\kappa$. For $n=0$ it is positive but negligible in $\kappa$.
\end{comment}
Hence, if the underlying $\prqc$ scheme is \nauf (see \cref{definition:adapt_flex_unforge}) then the $\pkqc$ scheme is \arnauf (see \cref{definition:nonadapt_all-or-nothing_unforge}). Moreover, if the underlying $\prqc$ scheme is \nauuf (see \cref{definition:adapt_flex_unforge} and \cref{definition:unconditional_unforgeability}), then $\pkqc$ is unconditionally \arnauf (see \cref{definition:nonadapt_all-or-nothing_unforge} and \cref{definition:unconditional_unforgeability}). Note that % even if the adversary is computationally unbounded, the verifiers aren't and hence it makes sense that $m,n \in \poly$. Also note that 
 this holds even if the underlying $\prqc$ scheme is inefficient.
%Limit[m*(Binomial[(m)*k, (n)*k]/Binomial[(m+1)*k, (n+1)*k]) - n, m -> Infinity, Assumptions -> {k > 1, m >= n > 0}]
%\qed 
\end{proof}

\section{Discussion and future Work}\label{sec:future_work}
\label{sec:discussion}
Our techniques for the construction of public coins scheme are fairly simple and general. We hope that this could be relevant to other related notions in quantum cryptography, such as quantum copy-protection~\cite{Aar09,ALL+20,ALZ20,CMP20}, quantum tokens for digital signatures~\cite{BS16a}, disposable backdoors~\cite{CGLZ19}, secure software leasing~\cite{KNY20}, uncloneable quantum encryption~\cite{BL20} (see also~\cite{BI20}),  and one shot signatures~\cite{AGK+20}. 
The lifting technique used in our work to lift a private coin scheme to an almost public coin scheme is based on the general idea of comparison-based verification.
Can such a lifting technique be used in related topics, such as quantum copy-protection~\cite{Aar09}? 
E.g., suppose a software company issues several quantum states as replicas of a quantum copy-protected program, to its users. It might be possible to use comparison-based verification to avail second-hand transactions of such programs without going to the issuer for verification. However, in order to use comparison-based verification, it is crucial that the algorithm generating the copy-protected programs must always generate the same quantum state for a fixed program. This is an additional requirement which does not follow from the existing definition of quantum copy-protection~\cite{Aar09}.

Next, even though $\PRS$ are sufficient to build private quantum coins and, hence, private quantum money, it is unclear if their existence is necessary to build private quantum money. In fact, for private quantum money, we currently do not know of any microcrypt lower bound, i.e., a candidate primitive in microcrypt primitive, the existence of which is necessary to build private quantum money. This raises the following natural question. 
\begin{openproblem}\label{op:private-quantum money-necessary-assumption}
Is there a microcrypt primitive implied by private quantum money? 
\end{openproblem}
It is known that One-Way State Generators (OWSG), a weaker primitive than $\PRS$, are necessary to build private quantum \emph{coins}~\cite{MY22b}, but for general private quantum money, no such lower bound is known.%,\anote{UPSG is still not implied by PRS. Do we still want to remove this?} which hints towards the possibility of constructing private quantum money from assumptions even weaker than $\PRS$, thus raising the following question.

% \begin{openproblem}\label{op:private-quantum money-sufficient-assumption}
% Is there a microcrypt primitive, potentially weaker, or implied by $\PRS$ that implies private quantum money? 
% \end{openproblem}
% Note that a positive answer to \Cref{op:private-quantum money-sufficient-assumption}, which also constructs private quantum coins from a primitive potentially weaker than $\PRS$ would immediately lead to a $\arnauf$ scheme under the same assumption by \Cref{prop: rational-unforge}.

% Recently, a partial negative answer to \Cref{op:private-quantum money-sufficient-assumption} was given in a future work~\cite{BMM+24}, showing that there is no fully black-box construction of private quantum money from EFI pairs and one-way puzzles.

%It is known that the private coin scheme, given in~\cite{JLS18} is \nauf, based on quantum secure one-way functions. It is not known if the \cite{JLS18} scheme is \auf in the strongest sense. Note that our definition of adaptive unforgeability is still weak, since we do not return post verified states in Game~\ref{game:nonadapt_unforge_strongest}. In our work, we prove that the private money scheme in~\cite{JLS18} is \mnauf (see \cref{thm:sqaruf}), which is a weaker form of adaptive threat model, but still stronger than nonadaptive unforgeability.

There are a number of candidate constructions of public quantum money, some of which are provably adaptively unforgeable in the strongest sense, based on strong assumptions such as indistinguishability obfuscation. However, all these constructions, to the best of our knowledge, are bill schemes, and there is no known candidate construction for public quantum coins,  which raises the following question.

\begin{openproblem}
    Is there a public quantum coins construction that is adaptively unforgeable based on generic assumptions?
\end{openproblem}

\ifab \vspace{-10pt} \fi 
%\ifnum\masterthesis=0
% \if\anonymous=0
\paragraph{Acknowledgments}
We would like to thank Zhengfeng Ji for helpful discussions. 
\BeforeBeginEnvironment{wrapfigure}{\setlength{\intextsep}{0pt}}
            \begin{wrapfigure}{r}{100px}
                %\centering
                \includegraphics[width=100px]{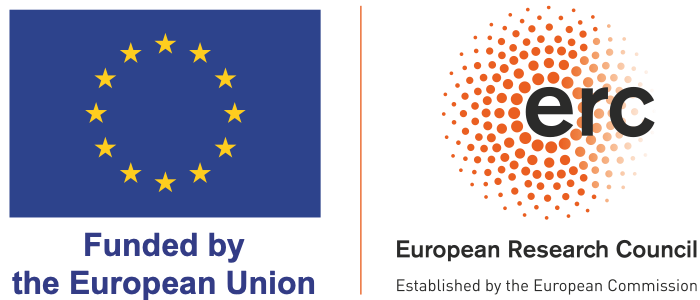}
            \end{wrapfigure}
            This work was funded by the European Union (ERC-2022-COG, ACQUA, 101087742). Views and opinions expressed are however those of the author(s) only and do not necessarily reflect those of the European Union or the European Research Council Executive Agency. Neither the European Union nor the granting authority can be held responsible for them.
%\fi 
%\fi
%the following line should be commented when compiling the "final version" (long version). It should not be commented out in the "abstract version". Same thing just before the bibliography
\fi 

\ifnum\sigconf=1
    \bibliographystyle{ACM-Reference-Format}
\else
    \bibliographystyle{alphaabbrurldoieprint}
\fi
{
%\ifab
%\footnotesize
%\fi

\ifnum\masterthesis=0
\bibliography{almost_public_quantum_coins_crypto}}
\fi

% \iffv
% \appendix
% \ifnum\llncs=1

% \fi

%\section{Temp for Nomenclature}
\iffv
    \appendix
    \ifnum\llncs=1
        \appendixpageoff
        \chapter*{Supplementary Material}
    %\section{Temp for Nomenclature}
    \fi

\section{Security against sabotage}\label{appendix:fairness} In the context of money, we should also consider attacks which are intended to hurt honest users rather than forging money. %, so that the honest users end up with less money than what they should have had. 
 For example - suppose an adversary who starts with one fresh coin, gives you a coin which you accepted after successful verification, but later when you send this money to somebody else, it does not pass verification. Note that, the adversary did not manage to forge money, but was able to sabotage the honest user. Such attacks are called sabotage attacks (as discussed in~\cite{BS16a}).  

In public key quantum money schemes, usually unforgeability ensures that the users do not get cheated. However in case of quantum money scheme with comparison-based verification this implication is not very clear. Since the verification involves the quantum coins of the user's wallet, it is essential that the honest users (money receivers) do not end up having less valid money than what they should have had, due to transactions with adversarial merchants. For example - suppose an adversary starts with one fresh coin, and gives you a coin. You accept it after successful public verification using a wallet, which initially had one true coin. Later, you send both these money states (the one received and the one initially had) to somebody else, but only one of the two coins passes verification. %Such attacks can not be 

%In case of public quantum coin scheme with private verification this can be ensured if applying the private verification $\verify_{sk}()$ on the honest user/verifier's wallet by the bank, results in at least $(n+m)$ public coins where $n$ is the number of coins successfully approved by the honest user and $m$ is the number of valid coins it had initially. 
For a comparison-based public quantum coin scheme with private verification, we split the discussion about sabotage attacks into two categories: security against private sabotage attacks and security against public sabotage attacks. We only discuss the nonadaptive model since in the context of our work, it suffices to consider only nonadaptive sabotage attacks. %Our construction does not satisfy rational adaptive security against private sabotage in the flexible setting, and in the all-or-nothing setting, adaptive security against sabotage is equivalent to the multiverifier nonadaptive case, see \cref{remark:all-or-nothing-adpt-to-multi-sec_against_sabotage}.

Next, we argue why our construction $\pkqc$ (given in \cref{alg:ts}) is not rationally secure against private sabotage in the flexible adaptive setting.
\begin{remark}\label{remark:flex-adapt-priv-sabotage}
The construction $\pkqc$ (given in \cref{alg:ts}) is not rationally flexible-adaptive secure against private sabotage.
\end{remark}
The proof sketch is as follows: The attack given in the proof of \cref{remark:adaptive_Attack} can also be seen as a private sabotage attack in the flexible adaptive setting. In that attack, if the verification of all the adversary's money succeeds, then the honest verifier, on doing a refund from the bank will get one coin less, and hence, loses one coin. If the verification fails at some stage, then due to the flexible setting, the verifier gains at most a coin, after doing refund. Since the success probability is large, the expected loss of the adversary is non-negligible.

\subsection{Security against private sabotage}\label{appendix: priv rational fair}
First, we discuss what happens when the user after receiving a transaction goes to the private verification for refund. We call sabotage attacks in such cases as \emph{private sabotage} attacks. 

In order to address private sabotage attacks, we follow the same way as we did to define unforgeability. We will define a security game and define a loss function for the honest verifier, that the adversary tries to maximize. We will call a comparison-based public quantum coin scheme (standard) with private verification, secure against private sabotage with respect to loss function defined in a security game if for every QPT adversary the probability that the loss is positive is negligible function. Similarly, it is rationally secure against private sabotage with respect to the loss function if for every QPT adversary, the expected loss is negligible.

%In a standard public quantum money with a classical public key, the refund was calculated using the public count but in a public quantum coins with private verification, private count might be used to calculate the refund %the $\Count$ used to define refund in the next security game (Game~\ref{game:nonadapt-fair}), is with respect to the private verification. 
In the next security game, Game~\ref{game:nonadapt_rational-fair}, $\verify_\sk$ represents the private verification, and $\Count_{refund}$ refers to the private count. We assume that in this particular security game, the private verification returns the post-verified state as well. The $\Count$ algorithm used in Game~\ref{game:nonadapt_rational-fair}, is the public count.
% For all other public money schemes, we use the convention that, $\verify_\sk$ outputs $\bot$, and $\Count_{refund}$ refers to the public count in Game~\ref{game:nonadapt_rational-fair}. For all public money schemes, the $\Count_\pk$ algorithm used in Game~\ref{game:nonadapt_rational-fair}, is the public count. For comparison based money schemes, $\pkkeygen$ outputs $\bot$. %For other money schemes, we use the convention that $\verify_\sk$ represents $\bot$ in Game~\ref{game:nonadapt_rational-fair}. 
\begin{boxx}{\procedure[linenumbering, space=auto]{$\napf^{\adv,\MS}_{\secpar}$:}
{sk\gets\keygen(\secparam)\\
\rho_1,\ldots,\rho_m \xleftarrow{\text{$\rho_i$ can be potentially entangled}} \adv^{\bank(\sk), \verify_\sk()}(\secparam)\\
\rho\equiv (\rho_1,\ldots,\rho_m)\\
% \text{Denote $\omega$ be the state of the wallet before receiving the coins, respectively}\\
% refund \gets \Count_{refund}(\omega) \text{ \Comment{Refund of the wallet before accepting the coins.}}\\ 
%\text{Denote $\rho$ to be a collection of $m$ alleged coins.}\\
%\text{If all the coins pass, } Counter \gets m \\
%\text{Else } Counter \gets 0\\%\text{Game output is $1$ if and only if }m'>n
m' \gets \Count(\rho)\\%\text{\Comment{Here we use the public Count.}}\\
\text{Denote $\omega'$ be the state of the wallet after receiving the coins, respectively}\\
refund' \gets \Count_{refund}(\omega') \text{ \Comment{Refund of the wallet after accepting the coins.}}\\ 
%\text{Loss of the honest verifier: } L(\adv) = Counter+1-refund.
\pcreturn m,m', refund'.
}
} 
\caption{Nonadaptive Security against private sabotage}
\label{game:nonadapt_rational-fair}
\end{boxx}

With respect to Game~\ref{game:nonadapt_rational-fair}, we define the following all-or-nothing loss function. Since, we only require this particular loss function from here onwards, we can write it as\\ $L(\napf^{\adv,\MS}_{\secpar})$. 

\begin{equation}\label{eq:loss_def}
     L(\napf^{\adv,\MS}_{\secpar}) = \begin{cases} 
                m + 1 - refund', &\text{if $m=m'$,}\\
                1-refund', &\text{otherwise.}
              \end{cases}
\end{equation}
Since, $\MS$ is a comparison-based scheme, we assume that the initially before receiving the adversary's money, the wallet has one fresh coin and hence has a worth of $1$.

From now onwards, for any adversary $\adv$ in Game~\ref{game:nonadapt_rational-fair}, we will use $L(\adv)$ to represent $L(\napf^{\adv,\MS}_{\secpar})$ in short.

Hence, \nom{L}{$L(\adv)$}{loss of the honest verifier due to $\adv$, in the context of nonadaptive security against private sabotage in Game~\ref{game:nonadapt_rational-fair}} represents how much loss the honest verifier had, due to $\adv$ in Game~\ref{game:multi_nonadapt_rational-fair}.

\begin{definition}[Rational Security against private sabotage]\label{definition:rational_fair}
A public money scheme $\MS$ is \narpf if for every QPA $\adv$ in Game~\ref{game:nonadapt_rational-fair}, there exists a negligible function $\negl$ such that
\begin{equation}
    \mathbb{E}(L(\napf^{\adv,\MS}_{\secpar})) \leq \negl,  
    \label{eq:rational-fairness}
\end{equation} where $L(\napf^{\adv,\MS}_{\secpar})$ (also represented as $L(\adv)$) is the loss defined in \cref{eq:loss_def}.
%Depending on the unforgeability game $X$ considered in the definition, we conclude that the money scheme is \narpf, $\nrrf$, $\arrf$, $\irrf$, $\vrrf$ or $\qarrf$.
\end{definition}
%Again the definition is with respect to the setting where either the verifier approves all $m$ coins or no coins.
%The analogous definition of this notion in the non-rational setting will be to have $L(\adv)> 0$ with only negligible probability for every QPT $\adv$ where $L(\adv)$ is as defined in \cref{eq:loss_def} with respect to Game~\ref{game:nonadapt_rational-fair}.

%In the user manual (see \cref{subsec:user_manual}), the bank can be seen as giving refunds using $\Count_{total}$. Hence, the user manual ensures, that any sabotage attack that results a non-negligible loss on honest users (in expectation), would imply a violation of \cref{definition:rational_fair}. Therefore, if a scheme is \narpf, then it can be used securely in the user manual.

\begin{theorem}\label{thm:rational priv fair}
The scheme $\pkqc$ in \cref{alg:ts}%under the \ro restrictions
, is unconditionally \narpf (see \cref{definition:rational_fair}), if the underlying $\prqc.\verify(\sk,\ldots)$ is a rank-$1$ projective measurement.
\end{theorem}
The proof is given in \cref{appendix:appendix_proofs} on \cpageref{pf:thm:rational priv fair}.

\subsection{Security against public sabotage}\label{appendix:public_sabotage} 
Next, we discuss the property of security against public sabotage, which our construction satisfies. It is essentially a form of security against sabotage for a public quantum coin scheme with private verification, but in the scenario where money received from others upon successful verification is directly spent without going through the bank. In case of failed verification, she goes to the bank for a refund.

We define the notion of security against public sabotage using a \nrsp security game, which models the scenario of directly spending money received from others. We define a loss function with respect to the \nrsp game and call a public coin scheme with private verification, $\MS$ \nrrsp if for any adversary the expected loss function is negligible.

In the next security game, Game~\ref{game:nonadapt_respendability}, we assume that the money scheme is a comparison-based public quantum coin scheme with private verification, $\verify_\sk$ represents the private verification, and $\Count_{refund}$ and $\Count$ refers to the private and public count respectively. We assume that in this particular security game, the private verification $\verify_\sk$ returns the post-verified state as well. 
%For comparison-based schemes, $\pkkeygen(\sk)$ outputs $\bot$ and $\Count_\pk$ represents the public verification. %For other money schemes, we use the convention that $\verify_\sk$ represents $\bot$ in Game~\ref{game:nonadapt_rational-fair}. 
In the next security game, Game~\ref{game:nonadapt_respendability}, $k,m\in \NN$ and $n_i\in \NN$ such that $\sum_{i=1}^kn_i=m+1$. The parameter $k$ represents the number of transactions in which a merchant spends a previously received money and $n_i$ represents the number of money states paid in the $i^{th}$ transaction. Since it is a comparison-based scheme, we can assume without loss of generality, that the initially before receiving the coins from the adversary, the wallet had one fresh coin and hence its worth was $1$.

\begin{boxx}{\procedure[linenumbering, space=auto]{$\nrsp^{\adv,\MS}_{\secpar,k,m,(n_1,\ldots,n_k)}$:} 
{sk\gets\keygen(\secparam)\\
\rho_1,\ldots,\rho_m \xleftarrow{\text{$\rho_i$ can be potentially entangled}} \adv^{\bank(\sk), \verify_\sk()}(\secparam)\\
\rho\equiv (\rho_1,\ldots,\rho_m)\\
% \text{Denote $\omega$ be the state of the wallet before receiving the coins
% }\\
% refund \gets \Count_{refund}(\omega) \\ 
%\text{Denote $\rho$ to be a collection of $m$ alleged coins.}\\
%\text{If all the coins pass, } Counter \gets m \\
%\text{Else } Counter \gets 0\\%\text{Game output is $1$ if and only if }m'>n
m' \gets \Count(\rho)\\%\text{\Comment{Here we use the public Count.}}\\
\label{line:ver-fail} \pcif m'<m  \text{\Comment{When some money states do not pass.}}\\ 
\text{Denote $\omega'$ be the state of the wallet after verifying the coins}\\
refund'\footnote{We are using the convention to use prime for quantities post-verification and hence use $refund'$ instead of $refund$.} \gets \Count_{refund}(\omega') \\ 
\pcelse\\
\label{line:lexi}\text{Group the post verified wallet into chunks, $\omega'_1,\ldots,\omega'_k$ such that $|\omega'_i|=n_i$}\\
\pcfor i\in[k]\\
Count_i\gets \Count(\omega'_i)\\
\pcendfor\\
\pcendif\\
%\text{Loss of the honest verifier: } L(\adv) = Counter+1-refund.
\pcreturn m,m', refund',Count_1,\ldots,Count_k. 
}
} 
\caption{Nonadaptive Security against public sabotage}
\label{game:nonadapt_respendability}
\end{boxx}

In Game~\ref{game:nonadapt_respendability},  in Line~\ref{line:lexi}, the states are grouped in the lexicographic order.
In Game~\ref{game:nonadapt_respendability}, the if-else statements starting in Line~\ref{line:ver-fail}, represents the fact that if all the money sent by the adversary do not pass verification, i.e., the transaction is not successful, then the merchant goes for a refund from the bank, otherwise spends it in $k$ transactions.
Next, we define a loss function in the all-or-nothing setting.
With respect to the outcomes of $\nrsp^{\adv,\MS}_{\secpar,k,m,(n_1,\ldots,n_k)}$, we define the following quantities.
For $i\in[k]$, define
\begin{equation}\label{eq:gain-i}
    Accepted_i(\nrsp^{\adv,\MS}_{\secpar,k,m,(n_1,\ldots,n_k)})= \begin{cases}
        n_i &\text{if $Count_i=n_i$,}\\
        0 &\text{otherwise.}
    \end{cases}
\end{equation}
For every $i$, we will use $Accepted_i(\adv)$ to represent \newline$Accepted_i(\nrsp^{\adv,\MS}_{\secpar,k,m,(n_1,\ldots,n_k)})$ in short.

The quantity  \nom{Ac}{$Accepted_i(\adv)$}{the number of money states that were accepted in the $i^{th}$ payment of the received money, according to the all-or-nothing way of transaction.} represents the gain or the number of money states that were accepted in the $i^{th}$ payment of the received money, according to the all-or-nothing way of transaction.
The loss function \nom{Lo}{$\lanna$}{nonadaptive all-or-nothing Loss in the security against public sabotage game, Game~\ref{game:nonadapt_respendability}} is defined as:
\begin{align}\label{eq:loss-respendable-def}
    \lanna(\nrsp^{\adv,\MS}_{\secpar,k,m,(n_1,\ldots,n_k)})\\=\begin{cases}
        1 + m - \sum_{i=1}^k Accepted_i(\adv) &\text{if $m=m'$,}\\
        1 + 0 - refund' &\text{otherwise.}
    \end{cases}
\end{align}   
Note that, in the definition of $\lanna$, we use the fact that initial worth of the wallet before receiving the adversary's coins was $1$.

We will use $\lanna(\adv_{k,m,(n_1,\ldots,n_k)})$ to represent \\$\lanna(\nrsp^{\adv,\MS}_{\secpar,k,m,(n_1,\ldots,n_k)})$ in short.
With respect to this loss function, we define nonadaptive rational security against public sabotage as follows.
\begin{definition}[Rational security against public sabotage]\label{definition:secure_against_public_sabotage}
A comparison-based public quantum coin scheme with private verification is \nrrsp if for every QPT adversary $\adv$,  in Game~\ref{game:nonadapt_respendability}, and for every $k,m\in \NN$ and $(n_,\ldots,n_k)\in \NN^k$ such that $\sum_{i=1}^k n_i = m+1$, there exists a negligible function $\negl$ such that
\begin{equation}
    \mathbb{E}(\lanna(\nrsp^{\adv,\MS}_{\secpar,k,m,(n_1,\ldots,n_k)})) \leq \negl,  
    \label{eq:rational-fairness}
\end{equation} where $\lanna$ is the loss function defined in \cref{eq:loss-respendable-def}.
\end{definition}  

\begin{definition}[Rational Security against sabotage]
A public money scheme $\MS$ is \narf if it is both \narpf and \nrrsp.
\label{definition:rational_secure_against_sabotage}
\end{definition}

\begin{definition}[Rational Security]\label{definition:rational_secure} 
A public money scheme $\MS$ is \nars if it is both \narf and \arnauf.

\end{definition}

The scheme $\pkqc$ that we construct is \nrrsp unconditionally irrespective of the underlying private scheme \prqc.

\begin{theorem}[Nonadaptive rational security against public sabotage of \pkqc]
The money scheme $\pkqc$ given in \cref{alg:ts}, is \nrrsp, unconditionally\footnote{This holds even if the parameter $m$ in Game~\ref{game:nonadapt_respendability} is unbounded.} if the underlying $\prqc.\verify(\sk,\ldots)$ is a rank-$1$ projective measurement.
\label{thm:rational public fair}
\end{theorem} 

The proof is given on \cpageref{pf:thm:rational public fair}.  
Combining \cref{prop: rational-unforge}, \cref{thm:rational priv fair} with \cref{thm:rational public fair}, we get an upgraded lifting theorem that also ensures security against sabotage.

\begin{theorem}\label{thm:nonadaptive secure}
The public money scheme $\pkqc$ constructed in~\ref{alg:ts} is \nars (see \cref{definition:rational_secure}) if the underlying private quantum money scheme $\prqc$ is \nauf (see \cref{definition:adapt_flex_unforge}) and $\prqc.\verify$ is a rank-$1$ projective measurement. Moreover if $\prqc$ is \nauuf (see \cref{definition:unconditional_unforgeability}), then the scheme $\pkqc$ is also unconditionally (see \cref{definition:unconditional_unforgeability}) \nars (see \cref{definition:rational_secure}).
\end{theorem}
See the proof in \cref{appendix:appendix_proofs} on \cpageref{pf:thm:nonadaptive secure}.

%This also shows the main result in the single-verifier setting, \cref{thm:unconditional_secure}.
% \begin{corollary}\label{cor:unconditional_secure}
% If \prs exist, then a \nars (see \cref{definition:rational_secure}) 
%  public quantum coin with comparison-based verification also exists.
% Furthermore, there exists an inefficient as well as a stateful public quantum coin scheme with private verification and comparison-based verification scheme that is \nars (see \cref{definition:rational_secure}) unconditionally (see \cref{definition:unconditional_unforgeability}). 
% \end{corollary}
% The proof is given in \cref{appendix:appendix_proofs} on \cpageref{pf:cor:unconditional_secure}.
\section{Untraceability}\label{appendix:untraceability}
In our discussion regarding Coins and Bills, we emphasized the fact that Coins are more secure than Bills in terms of privacy and untraceability. In Ref.~\cite{AMR20}, the authors formalized the notion of untraceability. They defined the following untraceability game. 
 
\setcounter{algorithm}{5}%hack to sync the caption numbering of "game" environment.
\begin{savenotes}
\begin{game}
\begin{algorithmic}
    \State \textbf{set up the trace:} $\adv(1^\secpar)$ receives oracle access to $\bank(\sk)$ and $\pkkeygen(\sk)$, and outputs registers $M_1,\ldots,M_n$ and a permutation $\pi \in S_n$;
    \State \textbf{permute and verify:} $b\gets \{0,1\}$ is sampled at random, and if $b=1$ the states $M_1,\ldots,M_n$ are permuted by $\pi$. $\verify$ is invoked on all registers, the approved registers are placed in a set\footnote{In~\cite{AMR20}, they use $\mathcal{M}$ to denote the set. Since, we use $\MS$ to denote the money scheme, we use $\mathsf{S}$ to denote this set instead.} $\mathsf{S}$ while the rest are discarded;
    \State \textbf{complete the trace:} $\adv$ receives $\mathsf{S}$ and the secret key $\sk$,\footnote{In~\cite{AMR20}, the $\adv$ is given the entire state of the bank. Since, we only consider stateless money schemes, this is the same as giving the secret key $\sk$.} and outputs a guess $b'\in \{0,1\}$.
    \State The output of the game, $\mathsf{Untrace}_{\secpar}[\MS, \adv]$ is $1$ if and only if $b=b'$. 
\end{algorithmic}
\caption{Untraceability game: $\mathsf{Untrace}_{\secpar}[\MS, \adv]$}\label{game:byzantine-untraceability}
\end{game}
\end{savenotes}

%\ 
%In the above security game, $S_2$ denotes the group of permutations over two elements.
\begin{definition}[Untraceability of quantum money~\cite{AMR20}]\label{definition:byzantine_untrace_coins}
A money scheme $\MS$ is called \ut, if for every (even computationally unbounded) $\adv$, there exists a negligible function $\negl$ such that
\begin{equation}
    \Pr(\mathsf{Untrace}^{\adv,\MS}_{ \secpar}=1) - \frac{1}{2} \leq \negl. 
    \label{eq:Untraceability}
\end{equation}

\end{definition}

\paragraph*{Potential drawbacks of the \ut definition, \cref{definition:byzantine_untrace_coins}} The untraceability definition in~\cite{AMR20} (\cref{definition:byzantine_untrace_coins}) is the first rigorous definition for security against traceability of quantum money. However, we think that it is a bit too stringent to the adversary. According to the definition (see Game~\ref{game:byzantine-untraceability}) for private quantum coins, the bank is assumed to be honest. The adversary receives the bank's internal state only in the final step, ``Compute the trace'' in Game~\ref{game:byzantine-untraceability}. However, the bank might collude with the adversary at an earlier step to mount a trace attack. 

According to this definition, every private quantum money scheme in which the verification destroys the money state and then mints a new state in place of it, is untraceable. This is because irrespective of whether or not the bank permuted or not in the second step, ``Permute and verify'' in Game~\ref{game:byzantine-untraceability}, the money states that is returned to the adversary after verification is the same.
As a result, the variant of Wiesner's scheme where money is destroyed after every verification, and a new money state is minted, is untraceable according to \cref{definition:byzantine_untrace_coins}. Since every bill in Wiesner's money has a serial number, one would expect that such a scheme should not be untraceable.

We believe that there is a need for a better definition of untraceability due to the reasons mentioned above, and encourage further research and discussion on this topic.

\section{Multiverifier attacks}\label{appendix:multi-ver-attacks} 
The security definition that we discussed so far, portrays the model where an adversary tries to forge or sabotage an honest verifier's money. The verifiers can be the different bank branches, in case of private money schemes\footnote{For private money schemes we only discuss forging attacks as sabotage attacks make sense only in the public setting.}, or public users, in case of public money schemes. However, the adversary need not be confined to such attacks, which essentially involve just one payment to a verifier. It might attack through multiple payments. This can be done by attacking one verifier with multiple small payments (for example, buying several items from the same merchant, perhaps using different identities to avoid being identified in case of failed attempt), or by attacking multiple victims. Note that, the user manual (see \cref{subsec:user_manual}) for the scheme, \pkqc (described in \cref{alg:ts}), does not forbid the attacks mentioned above, i.e., nothing in the user manual stops the adversary from paying money multiple times to the same verifier or paying multiple verifiers, as a part of its strategy to cheat or sabotage. However, according to the user manual (see \cref{subsec:user_manual}) of the scheme \pkqc (described in \cref{alg:ts}), a fresh public coin is used for every new transaction. Hence, in the context of the scheme \pkqc, attacks involving multiple small payment to one verifier, can be simply viewed as attacks involving payment to multiple verifiers. 
Therefore in the light of our work, we define a multiverifier version of the unforgeability and the security against sabotage. The multiverifier versions are the respective generalizations of the two security notions, and capture any general forging or sabotage attack against the scheme $\pkqc$, implemented according to the user manual discussed in \cref{subsec:user_manual}. 
For general private money schemes, this security model is relevant in a nonadaptive setting with multiple bank branches, i.e, the adversary tries to cheat against multiple bank branches, none of which returns the post verification state of money submitted. For general public money schemes, the security model is relevant in the case where there are multiple honest users whom the adversary tries to cheat or sabotage, but the honest users never spend or return the money received, similar to the user manual that we describe in \cref{subsec:user_manual}. %The adversary can also verify money with the same user, multiple times. Such an attack is captured as a multiverifier attack, by using a fresh public coin, for every new transaction -- see the user manual in \cref{subsec:user_manual}.
In the scheme, \pkqc, the user manual allows the adversaries to access the refund oracles, i.e., the adversary can at any point of time go to the bank for a refund of its wallet. We capture this in the multiverifier nonadaptive security against sabotage game and the multiverifier nonadaptive unforgeability game (Games~\ref{game:multi_nonadapt_rational-fair} and~\ref{game:multi_nonadapt_unforge}, respectively), by giving the adversary access to $\verify_\sk$ oracle, the private verification. Note that, this is relevant only in the case of public quantum coins with private verification, and hence for public quantum money schemes without private verification, we use the convention that $\verify_\sk$ produces $\bot$ (garbage).

\subsection{Multiverifier nonadaptive unforgeability}\label{appendix:multi-ver-unforge}  First, we define the multiverifier nonadaptive unforgeability for any (private or public money scheme, including comparison-based public coins and public quantum coins with private verification) quantum money scheme, via the following security game (see Game~\ref{game:multi_nonadapt_unforge}). 

For a private scheme $\MS$ in Game~\ref{game:nonadapt_unforge_strongest}, we use the convention that $\pkkeygen$ and $\Count$ in line~\ref{line:priv_and_public_ver} is the private count algorithm associated with the private verification. %\footnote{This is the sole reason why we use $\Count$ in Game~\ref{game:multi_nonadapt_unforge} instead of $\Count_\pk$ as used in Game~\ref{game:nonadapt_unforge_strongest}.}. 
For a public coin scheme with comparison-based verification, $\pkkeygen$ outputs $\bot$, and $\Count$ represents the comparison-based count algorithm that takes a wallet as the verification key (see \cref{definition:count}). For all other public money schemes, including (not comparison-based) public quantum coins with private verification, $\pkkeygen$ outputs a fresh public key, and $\Count$ represents the public count associated with respect to public verification. 
Finally, for private schemes and public coin schemes with private verification, $\verify_\sk$ represents the private count; for all other public money schemes, $\verify_\sk$ outputs $\bot$. The post-verification state after each $\verify_\sk$ query is never returned to the adversary $\adv$. For all other money schemes, we use the convention that $\verify_\sk$ returns $\bot$. 

There are a few disparities between the algorithms used in Game~\ref{game:multi_nonadapt_unforge} and its single-verifier counterpart, Game~\ref{game:nonadapt_unforge_strongest}. Firstly, we use $\verify_\sk$ in Game~\ref{game:multi_nonadapt_unforge} instead of $\Count_\sk$, representing the fact that the adversary can go to the bank multiple times for coin-by-coin private verification to avail a refund without any all-or-nothing restriction, irrespective of the utility model. Secondly, we use the $\Count$ algorithm in Game~\ref{game:multi_nonadapt_unforge} instead of C as done in Game~\ref{game:nonadapt_unforge_strongest} because, for multiverifier unforgeability in the case of private money schemes, we would like the $\Count$ algorithm to play the role of the private count $\Count_\sk$ instead of outputting $\bot$ as $\Count_\pk$ does in Game~\ref{game:nonadapt_unforge_strongest}. Hence it makes more sense to use $\Count$ as opposed to $\Count_\pk$ in Game~\ref{game:multi_nonadapt_unforge}.

\begin{savenotes}
\begin{boxx}
{\procedure[linenumbering, space = auto]{$\mnauf^{\adv,\MS}_{\secpar}$:}
{
sk\gets\keygen(\secparam)\\
\pcfor i \in [k] \t \text{\Comment{$k\in \poly$, is the number of verifiers/bank branches}.}\\
 (\rho_1^i,\ldots,\rho_{m_i}^i) \xleftarrow{\text{$\rho_i$ can be potentially entangled}} \adv^{\bank(\sk),\pkkeygen(\sk), \verify_\sk()}(\secparam)\\
\rho_i\equiv (\rho_1^i,\ldots,\rho_{m_i}^i)\\
m'_i \gets \Count(\rho_i)\text{ \Comment{The strategy of $\adv$ can depend on $m'_i$.}}\\ 
\pcendfor\\
\text{Denote by $n$ the number of times that the $\bank(\sk)$ oracle was called by $\adv$}\\
\text{Denote by $n'$ the number of times $\verify_\sk()$ output $1$ (accept).}\\
%\text{If all the coins pass ($m=m'$), set } Counter \gets m \\
%\text{Else, set } Counter \gets 0\\
%\text{Utility of the adversary:} U(\adv) = Counter - n.\\
%\text{Utility in the byzantine sense:} \widetilde{U}(\adv) = m' - n.
\pcreturn m_1,m'_1,m_2,m'_2,\ldots m_k, m'_k,n,n'.
}
}
\caption{Multiverifier Nonadaptive Unforgeability Game}
\label{game:multi_nonadapt_unforge}
\end{boxx}
\end{savenotes}
With respect to Game~\ref{game:multi_nonadapt_unforge}, we define the following quantities. For every $i \in [k]$, we define
\begin{equation}\label{eq:sing_rational_utility_def}
    U_i(\mnauf^{\adv,\MS}_{\secpar}) = \begin{cases} 
                m_i, &\text{if $m_i=m'_i$,}\\
                0, &\text{otherwise.}
              \end{cases}
\end{equation}
Finally, we define the following utility function,
\begin{equation}\label{eq:multi_rational_utility_def}
    \uanmna(\mnauf^{\adv,\MS}_{\secpar}) = \sum_{i=1}^{k}U_i(\adv) + n' - n.
\end{equation}
\begin{equation}\label{eq:byzantine_multi_rational_utility_def}
    \uflmna(\mnauf^{\adv,\MS}_{\secpar}) = \sum_{i=1}^{k}m'_i + n' - n.
\end{equation} 
We will use $U_i(\adv)$ and $\uanmna(\adv)$ to represent in short \newline$U_i(\text{\mnauf}^{\adv,\MS}_{\secpar})$ and
\newline$\uanmna(\mnauf^{\adv,\MS}_{\secpar})$ respectively, for every $\adv$ in Game~\ref{game:multi_nonadapt_unforge} and $i \in [k]$.

We shall see  \nom{Um}{$\uanmna(\adv)$}{combined utility of the adversary \adv, in the context of multiverifier nonadaptive rational unforgeability} as the multiverifier version of the single-verifier nonadaptive all-or-nothing utility function defined in \cref{eq:all-or-nothing_nonadaptive_utility_def}. \nom{Us}{$U_i(\adv)$}{utility of $\adv$ due to the $i^{th}$ verifier, in the context of multiverifier nonadaptive rational unforgeability} is the utility the $\adv$ gained by submitting money to the $i^{th}$ verifier, and $\uanmna(\adv)$ is the combined utility of \adv. Similarly, the utility \nom{Ut}{$\uflmna(\adv)$}{utility of \adv, in the context of multiverifier nonadaptive unforgeability} is the multiverifier version of the single-verifier nonadaptive flexible utility function defined in \cref{eq:flex_nonadapt_utility_def}. Following the usual convention, we shall use the standard multiverifier nonadaptive unforgeability to denote standard unforgeability in Game~\ref{game:multi_nonadapt_unforge}, with respect to the flexible utility function, $\uflmna(\adv)$ defined in \cref{eq:byzantine_multi_rational_utility_def}. However, as seen earlier, we will use the all-or-nothing utility, $\uanmna(\adv)$ defined in \cref{eq:multi_rational_utility_def} to prove rational unforgeability for our construction (see \cref{alg:ts}). Therefore from now onwards, we use multiverifier rational nonadaptive unforgeability to denote rational unforgeability in Game~\ref{game:multi_nonadapt_unforge} with respect to $\uanmna(\adv)$.

\begin{definition}[Multiverifier nonadaptive rational unforgeability]\label{definition:multi-rational_unforge}
A money scheme $\MS$ is \mnaruf if it is rationally unforgeable in Game~\ref{game:multi_nonadapt_unforge} with respect to the utility function, $\uanmna$ is as defined in \cref{eq:multi_rational_utility_def}, i.e., for every QPA (Quantum Poly-time Algorithm) $\adv$ in  Game~\ref{game:multi_nonadapt_unforge}, 
there exists a negligible function $\negl$ such that,
\begin{equation}
    \mathbb{E}(\uanmna(\mnauf^{\adv,\MS}_{\secpar})) \leq \negl. %\mathbb{E}(U(\mathscr{H}))
    \label{eq:multi-rational-unforge}
\end{equation}
%If the unforgeability game $X$ in consideration is \nauf, then we conclude that the scheme is \arnauf. %Similar definition hold for $\nrruf$, $\arruf$, $\irruf$, $\vrruf$ and $\qarruf$ cases.
\end{definition} 

\begin{definition}[Multiverifier nonadaptive unforgeability]\label{definition:byzantine-multi-unforge}
A money scheme $\MS$ is (standard) \mnauf if it is rationally unforgeable in Game~\ref{game:multi_nonadapt_unforge} with respect to the utility function, $\uflmna$ is as defined in \cref{eq:byzantine_multi_rational_utility_def}, i.e.,  for every QPA (Quantum Poly-time Algorithm) $\adv$ in  Game~\ref{game:multi_nonadapt_unforge} 
there exists a negligible function $\negl$ such that,
\begin{equation}
    \Pr[\uanmna(\mnauf^{\adv,\MS}_{\secpar}) > 0] = \negl. %\mathbb{E}(U(\mathscr{H}))
    \label{eq:byzantine-multi-rational-unforge}
\end{equation} \footnote{Clearly, multiverifier nonadaptive unforgeability implies multiverifier rational nonadaptive unforgeability.}
%If the unforgeability game $X$ in consideration is \nauf, then we conclude that the scheme is \arnauf. %Similar definition hold for $\nrruf$, $\arruf$, $\irruf$, $\vrruf$ and $\qarruf$ cases.
\end{definition}

At this point, it might seem that Game~\ref{game:multi_nonadapt_unforge} is very similar to Game~\ref{game:nonadapt_unforge_strongest} since going for adaptive verifications is similar to submitting to multiple verifiers on the fly. It might seem that the utilities $\uflmna(\adv)$ (same as $\uanmna(\mnauf^{\adv,\MS}_{\secpar})$) and $\uanmna(\adv)$ (same as $\uanmna(\mnauf^{\adv,\MS}_{\secpar})$) defined with respect to Game~\ref{game:multi_nonadapt_unforge} is equivalent to the utilities $\uana(\uf^{\adv,\MS}_\secpar)$ and $\ufla(\uf^{\adv,\MS}_\secpar)$ in Game~\ref{game:nonadapt_unforge_strongest}.
Indeed, they are the same when the verification key of the verification procedures ($\verify_\pk$ if the scheme is public and $\verify_\sk$ if it is private) is a fixed classical string. However, if the verification key is not a fixed classical string, then going for adaptive verifications in Game~\ref{game:nonadapt_unforge_strongest} is not the same as submitting to multiple verifiers in Game~\ref{game:multi_nonadapt_unforge}. This is because each verifier starts with a fresh key in Game~\ref{game:multi_nonadapt_rational-fair}, but the adaptive queries in Game~\ref{game:nonadapt_unforge_strongest} are instead answered using the same key, which might itself change along the way.

\subsection{Multiverifier security against sabotage}\label{appendix:multi-ver-fair}
Next we move to the multiverifier version of the security against sabotage, for public money schemes. We only discuss the private sabotage attacks in the context of multiple verifiers, and show that multiverifier rational security against private sabotage can be reduced to the single-verifier case, see \cref{prop:sing-multi-rational-fair}. The multiverifier rational security against public sabotage can be defined analogously, and the same multiverifier-to-single-verifier reduction that works in the private sabotage setting also holds for the public sabotage setting as well (see \cref{remark:single-multi-rational-fair-public}). Hence we omit the discussion about multiverifier security against public sabotage.

In the next security game, Game~\ref{game:multi_nonadapt_rational-fair}, we model private sabotage attacks against multiple verifiers for comparison-based public quantum coins with private verification. In Game~\ref{game:multi_nonadapt_rational-fair}, we assume $\verify_\sk$ represents the private verification, and $\Count_{refund}$ refers to the private count. In this particular security game, we assume that the private verification also returns the post verified state. The $\Count$ algorithm used in Game~\ref{game:nonadapt_rational-fair}, is the public count. %For other money schemes, we use the convention that $\verify_\sk$ represents $\bot$ in Game~\ref{game:nonadapt_rational-fair}. %In the context of our work, we only discuss the definition for a comparison based public quantum coins with private verification. As discussed earlier, for standard public money schemes, in which there is only one verification algorithm, and verification does not involve own money, the security against sabotage notion is much simpler. The same holds for the multiverifier setting as well.

%In the next security game,  Game~\ref{game:multi_nonadapt_rational-fair}, in case of public quantum coins with private verification, $\verify_\sk$ represents the private verification. For other money schemes, we use the convention that $\verify_\sk$ represents $\bot$ in Game~\ref{game:multi_nonadapt_rational-fair}.
\begin{boxx}{\procedure[linenumbering, space=auto]{$\mnapf^{\adv,\MS}_{\secpar}$:}
{sk\gets\keygen(\secparam)\\
%\rho_1,\ldots,\rho_m \xleftarrow{\text{$\rho_i$ can be potentially entangled}} \adv^{\bank(\sk),\pkkeygen(\sk), \verify_\sk()}(\secparam)\\
%\rho\equiv (\rho_1,\ldots,\rho_m)\\
%\text{Denote $\rho$ to be a collection of $m$ alleged coins.}\\
%\text{If all the coins pass, } Counter \gets m \\
%\text{Else } Counter \gets 0\\%\text{Game output is $1$ if and only if }m'>n 
%m' \gets \Count_\pk(\rho)\quad \text{\Comment{Here we use the public Count.}}\\
%\text{Denote $\omega$ be the state of the wallet after receiving the coins}\\ 
%refund \gets \Count_{refund}(\omega) \\ 
%\text{Loss of the honest verifier: } L(\adv) = Counter+1-refund.
\pcfor  i \in [k] \t\text{ \Comment{$k\in \poly$, refers to the number of verifiers.}}\\
\rho_1^i,\ldots,\rho_{m_i}^i \xleftarrow{\text{$\rho_i$ can be potentially entangled}} \adv^{\bank(\sk), \verify_\sk()}(\secparam)\\
\rho_i\equiv (\rho_1^i,\ldots,\rho_{m_i}^i)\\
% \text{Denote $\omega_i$ be the state of the wallet before receiving the coins}\\
% refund_i \gets \Count_{refund}(\omega_i)\text{ \Comment{ This is not revealed to the adversary, \adv.}}\\
m'_i \gets \Count(\rho_i)\\
\text{Denote $\omega'_i$ be the state of the wallet after receiving the coins}\\
refund'_i \gets \Count_{refund}(\omega'_i)\text{ \Comment{ This is not revealed to the adversary, \adv.}}\\
\pcendfor\\
%\text{Denote by $n$ the number of times that the $\bank(\sk)$ oracle was called by $\adv$}\\
%\text{Denote by $n'$ the number of times $\verify_\sk()$ output $1$ (accept).}\\
%\text{If all the coins pass ($m=m'$), set } Counter \gets m \\
%\text{Else, set } Counter \gets 0\\
%\text{Utility of the adversary:} U(\adv) = Counter - n.\\
%\text{Utility in the byzantine sense:} \widetilde{U}(\adv) = m' - n.
\pcreturn m_1,m'_1,\ldots m_k, m'_k,refund'_1\ldots,refund'_k.
%\pcreturn m,m',refund.
}
} 
\caption{Multiverifier nonadaptive Security against sabotage}
\label{game:multi_nonadapt_rational-fair}
\end{boxx}
With respect to Game~\ref{game:nonadapt_rational-fair}, we define the following quantities.
For every $i \in [k]$,
\begin{equation}\label{eq:sing_loss_def}
     L_i(\adv) = \begin{cases} 
                m_i + 1 - refund'_i, &\text{if $m_i=m'_i$,}\\
                1-refund'_i, &\text{otherwise.}
              \end{cases}
\end{equation}
Since, $\MS$ is a comparison-based scheme, we assume that the initially before receiving the adversary's money, the wallet has one fresh coin and hence has a worth of $1$.

Finally we define the following loss-function,
\begin{equation}\label{eq:multi_rational_loss_def}
     L_{\text{multi-ver}}(\adv) = \sum_{i=1}^k L_i(\adv).
\end{equation}

We will use $L_i(\adv)$ and $L_{\text{multi-ver}}(\adv)$ to represent in short \newline$L_i(\text{\mnapf}^{\adv,\MS}_{\secpar})$ and
\newline$\uanmna(\mnapf^{\adv,\MS}_{\secpar})$ respectively, for every $\adv$ in Game~\ref{game:multi_nonadapt_rational-fair} and $i \in [k]$.

We shall view \nom{Ls}{$L_i(\adv)$}{loss of the $i^{th}$ verifier due to $\adv$, in the context of multiverifier nonadaptive rational security against sabotage} as the loss of the $i^{th}$ verifier due to $\adv$, and \nom{Lm}{$L_{\text{multi-ver}}(\adv)$}{combined loss of the verifiers due to the adversary, \adv, in the context of multiverifier nonadaptive rational security against sabotage} as the combined loss of the verifiers due to the adversary. We shall view $L_{\text{multi-ver}}(\adv)$ as the multiverifier version of the all-or-nothing loss function $L(\napf^{\adv,\MS}_{\secpar})$ (also represented as $L(\adv)$) defined in \cref{eq:loss_def}. Just like in the single-verifier case, we use the all-or-nothing loss function to define rational security against private sabotage. We do not know if our construction is secure against private sabotage with respect to a flexible loss function in Game~\ref{game:multi_nonadapt_rational-fair}.
%Since, the wallet is initialized to $\cent$, in order to perform  $\pkqc.\Count_{\ket{\cent}}$ (see \cref{line:Count} in \cref{alg:ts}), the loss is always one more than the difference in the public and private count values. We use the convention that $k$ is at most polynomial even for a computationally unbounded adversary. 

\begin{definition}[Multiverifier rational security against private sabotage]\label{definition:multi_rational_fair}
A money scheme $\MS$ is \narpf if for every QPA $\adv$ in Game~\ref{game:multi_nonadapt_rational-fair}, there exists a negligible function $\negl$ such that
\begin{equation}
    \mathbb{E}(L_{\text{multi-ver}}(\mnapf^{\adv,\MS}_{\secpar})) \leq \negl,
    \label{eq:multi-rational-fairness}
\end{equation} where $L_{\text{multi-ver}}(\mnapf^{\adv,\MS}_{\secpar})$ (same as $L_{\text{multi-ver}}(\adv)$) is as defined in \cref{eq:multi_rational_loss_def}.
\end{definition}

%Just as in the case of unforgeability, the user manual ensures that any sabotage attack on the scheme $\pkqc$ (construction given in \cref{alg:ts}) can be viewed as a multiverifier nonadaptive sabotage attack. Therefore, in order to prove that the scheme is secure against sabotage attacks, it is enough to show that it is \mnarpf.
\begin{definition}\label{definition:multi_rational_secure}
A money scheme $\MS$ is \mnars
 if it is both \mnaruf and \mnarpf (see \cref{definition:multi-rational_unforge,definition:multi_rational_fair}).
\end{definition}
Ideally, one should have required multiverifier security against sabotage instead of multiverifier security against private sabotage in \cref{definition:multi_rational_secure}. Since we do not discuss public sabotage for multiverifier settings (see the opening discussion in \cref{appendix:multi-ver-fair} on \cpageref{appendix:multi-ver-fair}), we settle for just multiverifier security against private sabotage. However, public sabotage for multiverifier settings can be defined by extending the single-verifier definition analogous to the private setting, and our construction does achieve multiverifier security against public sabotage, see \cref{remark:single-multi-rational-fair-public}.

\subsection{Results in the multiverifier setting}\label{appendix:multi-ver-results}
Next, we prove that the scheme \pkqc (discussed in \cref{alg:ts}), is \mnars, if the underlying private scheme is \mnauf, see \cref{cor:multi-rational-secure}. The way we prove \cref{cor:multi-rational-secure}, is via the following steps. First, we prove that, the scheme $\pkqc$ is \mnarpf, see \cref{cor:multi_ver-rational-fair}. The way we prove it, is by first reducing multiverifier nonadaptive rational security against sabotage to the single-verifier variant (for any public money scheme), see \cref{prop:sing-multi-rational-fair}, and then using \cref{thm:rational priv fair} on top of it. Then, we prove that, if the underlying private scheme, \prqc, is \mnauf, then multiverifier nonadaptive rational unforgeability for the scheme $\pkqc$, can be reduced to multiverifier nonadaptive rational security against sabotage. Hence, using \cref{cor:multi_ver-rational-fair}, we conclude that, if the underlying private scheme, \prqc, is \mnauf, then, \pkqc is \mnaruf, see \cref{prop: multi_ver-unforge}. Finally, by combining \cref{prop: multi_ver-unforge,cor:multi_ver-rational-fair}, we conclude the proof of \cref{cor:multi-rational-secure}.

\begin{proposition}\label{prop:sing-multi-rational-fair}
If a comparison-based public quantum coin scheme with private verification $\MS$ is (unconditionally) \narpf, then it is also (respectively, unconditionally) \mnarpf.
\end{proposition}
The proof is given on \cpageref{pf:prop:sing-multi-rational-fair}.
We would like to emphasize that the same holds in the case of public sabotage attacks using a very similar proof, even though we have not defined multiverifier security against public sabotage. 
\begin{remark}\label{remark:single-multi-rational-fair-public}
If a comparison-based public quantum coin scheme with private verification $\MS$ is (unconditionally) \nrrsp, then it is also (respectively, unconditionally) \mnrrsp\footnote{Even though we do not consider multiverifier security against public sabotage in this work (see the opening discussion in \cref{appendix:multi-ver-fair} on \cpageref{appendix:multi-ver-fair}) and hence have not defined it formally, one can easily extend single-verifier security against public sabotage to the multiverifier version, similar to how it is done in the private sabotage case (see Games~\ref{game:multi_nonadapt_rational-fair} and~\ref{game:nonadapt_rational-fair}).}. Hence, by \cref{thm:rational public fair}, $\pkqc$ in \cref{alg:ts} is \mnrrsp.
\end{remark} 

Combining \cref{thm:rational priv fair,prop:sing-multi-rational-fair}, we get the following corollary.
\begin{corollary}\label{cor:multi_ver-rational-fair}
The scheme $\pkqc$ (see \cref{alg:ts}) is \mnarpf, unconditionally if $\prqc.\verify$, the verification of the underlying private scheme $\prqc$, is a rank-$1$ projective measurement.
\end{corollary}
Our scheme $\pkqc$ (see \cref{alg:ts}) also satisfies multiverifier nonadaptive rational unforgeability in the following sense.
\begin{proposition}\label{prop: multi_ver-unforge}
The scheme $\pkqc$ (described in \cref{alg:ts}) is  \mnaruf (see \cref{definition:multi-rational_unforge} and \cref{definition:unconditional_unforgeability}), if the underlying private scheme $\prqc$ is (resp., unconditionally) \mnauf (see \cref{definition:byzantine-multi-unforge}), and $\prqc.\verify$ is a rank-$1$ projective measurement. 
%Moreover, if the scheme is unconditionally \narpf and the underlying $\prqc$ scheme is \nauuf (see \cref{definition:byzantine_unforge} and \cref{definition:unconditional_unforgeability}), then it is also \arnauf (see \cref{definition:rational_unforge})unconditionally. 
\end{proposition}
The proof is given on \cpageref{pf:prop: multi_ver-unforge}.

Combining \cref{cor:multi_ver-rational-fair,prop: multi_ver-unforge}, we get the following corollary. 
\begin{corollary}\label{cor:multi-rational-secure}
The scheme $\pkqc$ (described in \cref{alg:ts}) is (resp., unconditionally) \mnars (see \cref{definition:multi_rational_secure} and \cref{definition:unconditional_unforgeability}), provided the underlying private scheme $\prqc$ is (resp., unconditionally) \mnauf (see \cref{definition:byzantine-multi-unforge}), and $\prqc.\verify$ is a rank-$1$ projective measurement. 
\end{corollary}
Recall that, we use the private coins scheme given in~\cite{JLS18} \anote{(or the simplified version in~\cite{BS19})} and~\cite{MS10} to instantiate our construction, $\pkqc$ (see \cref{alg:ts}). In~\cite{JLS18}, the authors only discuss nonadapative unforgeability, but it can be shown that the scheme is indeed \mnauf. This follows from~\cite[Theorem 5]{JLS18}, that the authors prove. In~\cite[Theorem 5]{JLS18}, the authors show that every \emph{PRS} family satisfies the \emph{Cryptographic no-cloning Theorem with Oracle} (see~\cite[Theorem 5]{JLS18}). This means given polynomially many, suppose $n$ many copies, of a uniformly random state $\ket{\phi_k}$ chosen from a PRS family, any QPT adversary, with access to the reflection oracle $I - 2\ketbra{\phi_k}$, cannot clone it to $n+1$ copies, except with negligible fidelity. 
\begin{theorem}[{Restated from~\cite[Theorem 5]{JLS18}}]
For any PRS $\{\ket{\phi_k}\}_{k \in \mathcal{K}}$, $m\in \poly$, $m' > m$, and any QPT algorithm $\mathcal{C}$, there exists a negligible function, $\negl$ such that the $m$ to $m'$, cloning fidelity of $\mathcal{C}$,
\[\mathbb{E}_{k\in \mathcal{K}}\left\langle (\ketbra{\phi_k})^{\tensor m'},\mathcal{C}^{U_{\phi_k}}((\ketbra{\phi_k})^m) \right\rangle = \negl,\] where $U_{\phi_k} = I - 2\ketbra{\phi_k}$.
\label{thm:strong_noclon_prs}
\end{theorem}

Since the private coin is a uniformly random state from a PRS, \cref{thm:strong_noclon_prs} can be used to prove that the private coin scheme described in~\cite{JLS18}, is in fact \mnauf, using similar arguments as in the proof of~\cite[Theorem 6]{JLS18}. Hence, it is also \auf, since for private money schemes with classical private key, adaptive unforgeability is the same as multiverifier-nonadaptive-unforgeability. More precisely, we show that,
\begin{theorem}[{Adapted from~\cite{JLS18}}]
If \prs exist, then there exists a private quantum coin scheme that is \mnauf\footnote{Since the construction in~\cite{JLS18} has a fixed classical string as the private verification key, we conclude that it is also \auf (see \cref{definition:adapt_flex_unforge}), for more details see the discussion after \cref{definition:byzantine-multi-unforge} on \cpageref{definition:byzantine-multi-unforge}.} (see \cref{definition:byzantine-multi-unforge}) such that the verification algorithm is a rank-$1$ projective measurement.
\label{thm:sqaruf}
\end{theorem}
The proof is given in~\cpageref{pf:thm:sqaruf}.

In~\cite{MS10}, the authors prove \emph{black-box unforgeability} for their scheme, (see \cref{thm:qaruuf}), in which the adversary gets access to a reflection oracle around the coin state, as a black-box. As a result, the adversary has access to multiple verifications as well as the post verified state of the money. However, in the multiverifier nonadaptive model (see Game~\ref{game:multi_nonadapt_unforge}), the adversary only has access to multiple verification, and not the post verified state. Therefore, black-box unforgeability is a stronger\footnote{In the black-box unforgeability in~\cite{MS10}, the adversary is allowed to take money states only at the beginning and unlike multiverifier nonadaptive unforgeability, not given oracle access to $\bank$. However, an adversary $\adv$, which gets oracle access to $\bank$, such as in the multiverifier nonadaptive unforgeability game (Game~\ref{game:multi_nonadapt_unforge}), can be simulated by an adversary $\bdv$, which receives all the money states in the beginning. $\bdv$ simulates \adv, by taking the maximum number of coins $\adv$ takes in all possible runs. In the end if some coins remain unused, they can be submitted along with what $\adv$ submits. The unused coins should pass verification due to completeness of the scheme.} threat model than the multiverifier nonadaptive model. Hence, by \cref{thm:qaruuf}, the private scheme in~\cite{MS10} is also \mnauf.
%Similarly the money scheme in~\cite{MS10}, can be shown to be \mnauf because of the ``Complexity Theoretic No-Cloning Theorem'' (see~\cite[Theorem 2]{Aar09}, restated as~\cite[Theorem 4.2]{MS10}). Hence, we have the following results:

%(In fact, the private coin scheme constructed in~\cite{JLS18} or the simplified version in~\cite{BS19} is $\qarf$ (see  {\auf}))

\begin{theorem}\label{thm:sqaruuf}
There exists an inefficient private quantum coin scheme that is \mnauf unconditionally (see \cref{definition:byzantine-multi-unforge} and \cref{definition:unconditional_unforgeability}) such that the verification algorithm is a rank-$1$ projective measurement.
\end{theorem}

% \begin{corollary}\label{cor:multi-unconditional_secure}
% If \prs exist, then a \mnars (see \cref{definition:multi_rational_secure}) public quantum coin with comparison-based verification also exists.
% Furthermore, there exists an inefficient quantum money scheme and also a stateful quantum money scheme which are comparison-based public quantum coins with private verifications that is \mnars (see \cref{definition:multi_rational_secure}) unconditionally (see \cref{definition:unconditional_unforgeability}). 
% \end{corollary}
%The proof is almost the same as the proof of \cref{thm:unconditional_secure} given on \cpageref{pf:thm:unconditional_secure}.
\begin{comment}
Combining \cref{thm:multi-unconditional_secure} with \cref{remark:single-multi-rational-fair-public}, we get the following remark.
\begin{remark}\label{remark:multi-unconditional_secure_public}
If \prs exist, then a \mnars (see \cref{definition:multi_rational_secure}) public quantum coin with comparison-based verification also exists, which is also rationally secure against public sabotage in the multiverifier setting.
Furthermore, there exists an inefficient quantum money scheme and also a stateful quantum money scheme which are comparison-based public quantum coins with private verifications that is \mnars (see \cref{definition:multi_rational_secure}) unconditionally (see \cref{definition:unconditional_unforgeability}), and is also rationally secure against public sabotage in the multiverifier setting. 
\end{remark}
\end{comment}
\subsection{Multiverifier attack model for the user manual}\label{appendix:user_manual-multi-ver-equivalence}
The scheme, $\pkqc$, that we construct in \cref{alg:ts} is a public comparison-based quantum coin scheme with private verification. Any forging or sabotage attack on $\pkqc$, despite the user manual discussed in \cref{subsec:user_manual}, can be viewed as a multiverifier nonadaptive attack. If the adversary tries to forge or sabotage by submitting multiple times adaptively to the same user, then the honest user will use separate receiving wallets according to the user manual discussed in \cref{subsec:user_manual}. This is the same as the adversary submitting multiple honest verifiers one after another. The strategy of the adversary can only depend on what was the outcome of the previous honest users, to whom the adversary submitted. Any such attack is either a multiverifier nonadaptive forging or private sabotage attack, described in Games~\ref{game:multi_nonadapt_unforge} and~\ref{game:multi_nonadapt_rational-fair}, respectively, or a multiverifier public sabotage attacks, which we did not discuss. However since our construction is \nrrsp, it is also secure against public sabotage in the multiple verifiers setting, see \cref{remark:single-multi-rational-fair-public}. According to the user manual the adversary can go to the bank for refund at any point of the time. In Games~\ref{game:multi_nonadapt_unforge} and~\ref{game:multi_nonadapt_rational-fair}, the adversary can simulate the refund of the bank, using $\verify_\sk$ and the $\bank(\sk)$ oracles \footnote{In order to simulate, the adversary should send the coin to be refund, to the $\verify_\sk$ oracle, and if and only if $\verify_\sk$ accepts, it takes a coin using the $\pkqc.\bank$ oracle.}. Hence, in order to show that the scheme $\pkqc$ (described in \cref{alg:ts}) is secure against forging and sabotage attacks, when implemented using the user manual described in \cref{subsec:user_manual}, it is enough to show that the scheme $\pkqc$ on instantiating with some candidate private money schemes, has multiverifier security against public sabotage (by \cref{remark:single-multi-rational-fair-public}) and has the \mnars property, which we show in \cref{thm:multi-unconditional_secure}.
In Games~\ref{game:multi_nonadapt_unforge} and~\ref{game:multi_nonadapt_rational-fair}, $k$ denotes the number of verifiers the adversary attacks. We use the convention that $k$ is at most polynomially large, even for a computationally unbounded adversary. This makes sense because even though the adversary is computationally unbounded, the number of verifiers or branches, it can attack, should be polynomially bounded.

\section{Completeness and appendix proofs}\label{appendix:appendix_proofs}

First we give proof for completeness \cref{prop:completeness}.

\begin{proof}[Proof of \cref{prop:completeness}]\label{pf:prop:completeness} First, we show that every procedure in our construction (see \cref{alg:ts}) is a QPT. We assume that the underlying $\prqc$ scheme is complete and hence $\prqc.\bank(\sk,\ldots)$ is a QPT. Therefore, since $\kappa \in \log^c(\secpar)$ ($c>1$), $\pkqc.\bank$ is a QPT. Similarly $\pkqc.\verify_{bank}(\sk,\ldots)$ is a QPT since $\prqc.\verify(\sk,\ldots)$ is a QPT. It is known that $\{\Pi_{\Sym{n}}, (I - \Pi_{\Sym{n}})\}$ can be implemented efficiently~\cite{BBD+97}. Therefore, $\pkqc.\verify()$ is also a QPT.

We show by induction that (a) the wallet state before the $k^{th}$ verification of valid money states is $\ket{\cent}^{\tensor k}$ and (b) in the $k^{th}$ repeated verification of valid money,  $\Pr[\pkqc. \verify(\ket{\cent}) = 1] = 1.$

Base case ($k=1$): By the initialization of the wallet (see Line~\ref{line:init}), (a) is satisfied. Hence, the total state is $\ket{\cent}^{\tensor 2} = \ket{\mill}^{\tensor 2\kappa}$. Therefore,
$$ \Pr[\pkqc. \verify(\ket{\cent}) = 1] = \tr(\Pi_{\Sym{2\kappa}} \ketbra{\mill}^{\tensor 2\kappa}) = 1.$$

Induction step (assume for $k$ and prove for $k+1$): (a) The wallet state before the $k^{th}$ verification is, by assumptions, $\ket{\cent}^{\tensor k}$. In the $k^{th}$ verification, we verify the state $\ket{\cent}$, so the new wallet state immediately prior to the measurement is $\ket{\cent}^{\tensor k+1}$, and since the projection passes with probability 1 (by the induction hypothesis), we know that the wallet state does not change due to the projection. 

(b) By assuming the result which we proved in (a), the new wallet state is $\ket{\cent}^{\tensor k+1}$, which is invariant under the permutation of its registers. As such, it lies in the symmetric subspace and will therefore pass verification. 
Note that, as per the user manual (see \cref{subsec:notations}), only one transaction can be done with a receiving wallet initialized with one fresh coin. The proof above shows that, if all the coins received in a single transaction are valid public coins ($\ket{\cent}$), then they all will be accepted, and the transaction will be approved with certainty. Since every transaction is verified independently, using a separate receiving wallet, this is enough to prove completeness.
%\qed 
\end{proof}

Next we turn our attention towards proving rational security against sabotage of our construction $\pkqc$. 
\begin{proof}[Proof of \cref{thm:rational priv fair}]\label{pf:thm:rational priv fair}
Let $\adv$ be any computationally unbounded adversary (against single-verifier and private sabotage) who submits $m$ public coins in the security against sabotage game (Game~\ref{game:nonadapt_rational-fair}). WLOG, let the combined state of the $m$ coins be a pure state $\ket{\beta}$ and let $\omega'$ be the post measurement state of the wallet ($(m+1)$ coins). For mixed states, the proposition easily follows since every mixed state is a probabilistic ensemble of pure states.
%Let $X$ be a boolean random variable and $Y$, $Y'$ be random variables such that $X = 1$ iff $\pkqc.\Count_{\ket{\cent}}(\rho) = m$ and $Y := \prqc.\Count(\sk, \ketbra{\cent} \tensor \ketbra{\beta})/\kappa$, $Y' := \prqc.\Count(\sk, \omega')/\kappa$.
Let $X$ be a boolean random variable such that 
 \begin{equation}\label{eq:X-def}
     X = \begin{cases}
            1 &\text{if }  \pkqc.\Count_{\ket{\cent}}(\ketbra{\beta}) = m\\ 
            0 &\text{otherwise.}
         \end{cases}
 \end{equation}
 
 Let $Y$, $Y', Z$ be random variables such that 
 \begin{align}
 Y' := \frac{\prqc.\Count(\sk, \omega')}{\kappa},\\
 Y := \frac{\prqc.\Count(\sk, \ketbra{\cent} \tensor \ketbra{\beta})}{\kappa},\\
 Z := \pkqc.\Count_{bank}(\sk, \omega'),
 \end{align}where $\prqc$ is the private scheme we lift to $\pkqc$ in \cref{alg:ts}.

Note that, by the definition of $\pkqc.\Count_{bank}$ (given in \cref{alg:ts}),
\begin{equation}\label{eq:count-bank-relation}
    \mathbb{E}(Z) = \mathbb{E}(Y').
\end{equation}

Let \nom{H}{$\HTilde{m\kappa}$}{Subspace of $(m+1)\kappa$ register states such that the state of first $\kappa$ registers is $\ket{\cent}$ and the state of the last $m\kappa$ registers is some state in $\mathbb{H}^{\tensor m\kappa}$} be the subspace defined as $\HTilde{m\kappa} := \{\ket{\cent} \tensor \ket{\psi} ~|~ \ket{\psi} \in \mathbb{H}^{\tensor m\kappa}\}$ and \nom{P}{$\Pi_{\HTilde{m\kappa}}$}{projection on to $\HTilde{m\kappa}$} be the projection on to $\HTilde{m\kappa}$.

By the definition of loss, $L(\adv)$ (see \cref{eq:loss_def} after Game~\ref{game:nonadapt_rational-fair}), 
\[L(\adv)= mX + 1 - Z.\]
\begin{align} 
\implies &\mathbb{E}(L(\adv))\\ &= m\mathbb{E}(X) + 1 - \mathbb{E}(Z)\\
&= m\mathbb{E}(X) + 1 - \mathbb{E}(Y') &\text{by \cref{eq:count-bank-relation}}\\
&= m\Pr[\pkqc.\Count_{\ket{\cent}}(\ketbra{\cent}\tensor\ketbra{\beta}) = m] - \mathbb{E}(Y') + 1 &\text{see definition in \cref{eq:X-def}}. 
\label{eq:expected-loss}
\end{align}

% Let $\Count_{(m+1)\kappa}$ be an Hermitian operator defined as $\Count_{(m+1)\kappa} = \sum_{1\leq i \leq (m+1)\kappa} I \tensor \cdots $

 %if $\pkqc.\Count_{\ket{\cent}}(\ket{\beta}) = m$ and $1 - \frac{Y'}{\kappa}$ otherwise.
 As it name suggests, $\prqc.\Count$ simply counts how many registers with quantum state, $\ket{\cent}$ are present. This is indeed invariant under the permutation of the registers. We now prove that the symmetric subspace measurement ($\pkqc.\verify$) commutes with $\prqc.\Count(\sk, \ldots)$. Note that for any mixed state $\rho$ of $(m+1)\kappa$ registers,
 \begin{equation}\label{eq:count_private_def}\mathbb{E}(\prqc.\Count(\sk,\ket{\psi})) =\tr\left(\Count_{(m+1)\kappa}\rho\right),\end{equation}  and \[\Pi_{\Sym{(m+1)\kappa}} = \frac{1}{(m+1)\kappa!}\sum_{\pi \in {S}_{(m+1)\kappa}} Perm_{(m+1)\kappa}(\pi),\] where for every $n \in \NN$, \nom{Co}{$\Count_{n}$}{hamiltonian implementing $\prqc.\Count$, $\sum_{j \in [n]}  \Count_{(j,n)}$} is defined as  \[\Count_{n} := \sum_{j \in [n]}  \Count_{(j,n)},\]
  and for  every $n \in \NN$ and $j \in [n]$, \nom{Co}{$\Count_{(j, n)}$}{Hamiltonian counting whether the $j^{th}$ register is $\ket{\mill} = \ket{\phi_0}$} is defined as \[\Count_{(j, n)} := I \tensor \underbrace{\ketbra{\mill}}_j \tensor I,\] and for every $n \in \NN$ and every permutation $\pi \in {S}_n$\nomenclature[S]{$S_n$}{Symmetric group over $n$ objects}, the projector \nom{P}{$Perm_{n}(\pi)$}{ for every permutation $\pi$ of $n$ registers, it is the permutation operator $\sum_{\vec{i}\in \ZZ_d^n}\ket{\phi_{\pi^{-1}(i_1)}\ldots \phi_{\pi^{-1}(i_{n)}}}\bra{\phi_{i_1},\ldots \phi_{i_{n}}}$, (same as $P_{d}(\pi)$ in the notation of~\cite{Har13})} is defined as\[Perm_{n}(\pi) := \sum_{\vec{i} \in \ZZ_d^n}\ket{\phi_{\pi^{-1}(i_1)}\ldots \phi_{\pi^{-1}(i_{n})}}\bra{\phi_{i_1},\ldots \phi_{i_{n}}},\]
 where $\{\ket{\phi_j}\}$ is the basis for $\mathbb{H}$ defined in \cref{item:product_basis_definition} in \cref{subsec:notations}.
 \cref{eq:count_private_def} follows from the observation that for any mixed state $\rho$ of $(m+1)\kappa$ registers, the probability that the $j^{th}$ register of $\rho$ pass $\prqc.\verify(sk, )$ is $\tr(\Count_{(j, n)}\rho)$.

\paragraph*{}
For $\pi \in {S}_{(m+1)\kappa}$ and $j \in [(m+1)\kappa]$, \[Perm_{(m+1)\kappa}(\pi) \Count_{(j,(m+1)\kappa)} =  \Count_{(\pi^{-1}(j),(m+1)\kappa)} (Perm_{(m+1)\kappa}(\pi)).\]
Hence, for $j \in [(m+1)\kappa]$, 
\[\Pi_{\Sym{(m+1)\kappa}}\Count_{(j,(m+1)\kappa)} =  \Count_{(\pi^{-1}(j),(m+1)\kappa)} \Pi_{\Sym{(m+1)\kappa}}.\]

Therefore, \begin{align}
& \Pi_{\Sym{(m+1)\kappa}}\left(\Count_{(m+1)\kappa}\right)\\
&= \sum_{j \in [(m+1)\kappa]}\Pi_{\Sym{(m+1)\kappa}} \Count_{(j,(m+1)\kappa)}\\
&=  \sum_{j \in [(m+1)\kappa]}\Count_{(\pi^{-1}(j),(m+1)\kappa)} \Pi_{\Sym{(m+1)\kappa}}\\ 
&= \left(\Count_{(m+1)\kappa}\right)\Pi_{\Sym{(m+1)\kappa}} .\\
\end{align}
%Since the symmetric subspace measurement (public verification) commutes with $\prqc.\Count()$, 
Therefore, the operator commutes with the projection $\Pi_{\Sym{(m+1)\kappa}}$ and hence with the $I - \Pi_{\Sym{(m+1)\kappa}}$. Using this commutation property, it can be shown that in general, if $\widetilde{\omega}$ and $\omega'$ are the states of the wallet along with the $m$ received coins, before and after symmetric subspace measurement, respectively,
\begin{equation}
    \tr\left(\Count_{(m+1)\kappa}\omega'\right)=\tr\left(\Count_{(m+1)\kappa}\widetilde{\omega}\right).
\end{equation}
Hence, in our case,
\begin{equation}\label{eq:pre-commutation}
    \tr\left(\Count_{(m+1)\kappa}\omega'\right)==\tr\left(\Count_{(m+1)\kappa}(\ketbra{\cent}\tensor\ketbra{\beta})\right).
\end{equation}

Hence, in our case, 
\begin{align}
    \mathbb{E}(Y') &= \mathbb{E}(\prqc.\Count(\sk, \omega')/\kappa)\\ &= \tr\left(\Count_{(m+1)\kappa}\omega'\right)\\ &=\tr\left(\Count_{(m+1)\kappa}(\ketbra{\cent}\tensor\ketbra{\beta})\right) &\text{By \cref{eq:pre-commutation}}\\ 
    &= \mathbb{E}(\prqc.\Count(\sk, (\ketbra{\cent}\tensor\ketbra{\beta}))/\kappa) = \mathbb{E}(Y).\label{eq: commutation}
\end{align}

The proof crucially uses the commutation property in \cref{eq: commutation} which states that the private $\prqc.\Count()$ commutes with the public verification (symmetric subspace measurement). This is a property of our construction and does not follow from the definition itself. Hence, the above proposition might fail to hold in other constructions.

%Hence, $L(\adv) = m + 1 - \frac{Y}{\kappa} $ if $\pkqc.\Count_{\ket{\cent}}(\ket{\beta}) = m$ and $1 - \frac{Y}{\kappa} $ otherwise.

Therefore by \cref{eq:expected-loss}, 
\begin{align}
\mathbb{E}(L(\adv)) &= m\Pr[\pkqc.\Count_{\ket{\cent}}(\ketbra{\cent}\tensor\ketbra{\beta})=m] - \mathbb{E}(Y') + 1\\
&= m\Pr[\pkqc.\Count_{\ket{\cent}}(\ketbra{\cent}\tensor\ketbra{\beta})=m] - \mathbb{E}(Y) + 1\\
&= m (\bra{\cent}\tensor\bra{\beta}\Pi_{\Sym{(m+1)\kappa}}\ket{\cent}\tensor\ket{\beta}) \\
 - &\bra{\cent}\tensor\bra{\beta}\frac{1}{\kappa} \Count_{(m+1)\kappa} \ket{\cent}\tensor\ket{\beta} +1\\ 
&= \bra{\cent}\tensor\bra{\beta}(m \Pi_{\Sym{(m+1)\kappa}} - \frac{1}{\kappa} \Count_{(m+1)\kappa} + I_{(m+1)\kappa})\ket{\cent}\tensor\ket{\beta}\\
&= \bra{\cent}\tensor\bra{\beta}\Qoperator\ket{\cent}\tensor \ket{\beta}\\
&= \bra{\cent}\tensor\bra{\beta}\Pi_{\HTilde{m\kappa}} \Qoperator \Pi_{\HTilde{m\kappa}}\ket{\cent}\tensor \ket{\beta} &\text{since $\ket{\cent}\tensor \ket{\beta} \in \HTilde{m\kappa}$}\\
&\leq \lambda_{\text{max}}(\Pi_{\HTilde{m\kappa}} \Qoperator \Pi_{\HTilde{m\kappa}}),\\
\end{align}%Hence, $\mathbb{E}(L(\adv)) \leq \max_{\ket{\beta}} \bra{\cent}\tensor\bra{\beta}(m \Pi_{\Sym{(m+1)\kappa}} - \Count_{(m+1)\kappa})\ket{\cent}\tensor \ket{\beta}.$
where \nom{Q}{$\Qoperator$}{$ m \Pi_{\Sym{(m+1)\kappa}} - \frac{1}{\kappa} \Count_{(m+1)\kappa} + I_{(m+1)\kappa}$} is defined as $\Qoperator := m \Pi_{\Sym{(m+1)\kappa}} - \frac{1}{\kappa} \Count_{(m+1)\kappa} + I_{(m+1)\kappa}.$

%Therefore,
%$\mathbb{E}(L(\adv)) \leq \max_{\ket{\gamma} \in S} \bra{\gamma}(Q)\ket{\gamma}$ which % $\max_{\ket{\gamma} \in S} \bra{\gamma}(Q)\ket{\gamma}$ 
 %is also the largest eigenvalue of $\Pi_{\HTilde{m\kappa}} Q \Pi_{\HTilde{m\kappa}}$ where $\Pi_{\HTilde{m\kappa}}$ is the projection operator onto the subspace $S$.
Hence, it is enough to show that the largest eigenvalue of $\Pi_{\HTilde{m\kappa}} \Qoperator \Pi_{\HTilde{m\kappa}}$ is negligible. We now show that the largest eigenvalue of $\Pi_{\HTilde{m\kappa}} \Qoperator  \Pi_{\HTilde{m\kappa}}$ is indeed negligible.\\

%For all $\vec{j}=(j_1,\ldots,j_d)$ such that $\sum_{i}j_i = m\kappa$ let, $$\ket{\beta'_{(j_1,j_2,\ldots, j_d),m\kappa}} \equiv \frac{1}{\sqrt{\binom{m\kappa}{\vec{j}}}}\sum_{\vec{i}: T(\vec{i}) = \vec{j}} \ket{\phi_{i_1}\ldots\phi_{m\kappa}}.$$ Let $V^{Sym}_{m\kappa} = \{\ket{\beta'_{(i_1, \ldots,i_d), m\kappa}}\}_{\vec{i}: \sum_{j} i_j = m\kappa}$ be a basis for $\Sym{m\kappa}$.
Recall the orthogonal set $\BasisSymTilde{m\kappa}$ defined as 
$\BasisSymTilde{m\kappa} = \left\{\ket{\BasisSymTilde{m\kappa}_{\vec{j}}}\right\}_{\vec{j} \in \mathcal{I}_{d, m\kappa}}$ (see Notations  \cref{eq:basis_def} and \cref{eq:basis_def} in \cref{subsec:notations}). By a similar argument as in the proof of \cref{lemma:eigenbasis}, we can show that \[\left(Span\left(\BasisSymTilde{m\kappa}\right)\right)^{\perp} \subset \ker(\Pi_{\HTilde{m\kappa}} \Pi_{\Sym{(m+1)\kappa}} \Pi_{\HTilde{m\kappa}}).\]% spans $\ker(\Pi_{\HTilde{m\kappa}} \Pi_{\Sym{(m+1)\kappa}} \Pi_{\HTilde{m\kappa}})^{\perp}$.
Note that $\Pi_{\HTilde{m\kappa}} \Count_{(m+1)\kappa}\Pi_{\HTilde{m\kappa}}$ and $\Pi_{\HTilde{m\kappa}}$ have non-negative eigenvalues. Therefore $\Pi_{\HTilde{m\kappa}} \Qoperator \Pi_{\HTilde{m\kappa}}$, which can be written as \[m(\Pi_{\HTilde{m\kappa}} \Pi_{\Sym{(m+1)\kappa}} \Pi_{\HTilde{m\kappa}}) - (\Pi_{\HTilde{m\kappa}} + \Pi_{\HTilde{m\kappa}} \Count_{(m+1)\kappa}\Pi_{\HTilde{m\kappa}}),\] has all its positive eigenvalues contained in the span of $\BasisSymTilde{m\kappa}$. 
  
Moreover a simple calculation (similar to what wee did in \cref{eq:eigen_vec} in the proof of \cref{lemma:eigenbasis}) shows that $\BasisSymTilde{m\kappa}$ is a set of eigenvectors of $\Pi_{\HTilde{m\kappa}}\Pi_{\Sym{(m+1)\kappa}} \Pi_{\HTilde{m\kappa}}$. %$W'$ forms a set of eigenvectors for $\Count_{(m+1)\kappa}$. 
For every $\ket{\BasisSymTilde{m\kappa}_{\vec{j}}} \in \BasisSymTilde{m\kappa}$,
\[\Pi_{\HTilde{m\kappa}} \Pi_{\Sym{(m+1)\kappa}} \Pi_{\HTilde{m\kappa}} (\ket{\BasisSymTilde{m\kappa}_{\vec{j}}}) = \frac{\binom{m\kappa}{{(j_0,j_1,\ldots, j_{d-1})}}}{\binom{(m+1)\kappa}{(j_0 + \kappa,j_1,\ldots, j_{d-1})}} (\ket{\BasisSymTilde{m\kappa}_{\vec{j}}}).\]

Clearly, $\BasisSymTilde{m\kappa}$ forms a set of eigenvectors for $\Pi_{\HTilde{m\kappa}}$ as well as for $\Count_{(m+1)\kappa}$. 
 Hence, $\BasisSymTilde{m\kappa}$ is a set of eigenvectors of $\Pi_{\HTilde{m\kappa}} \Qoperator \Pi_{\HTilde{m\kappa}}$. Since $\BasisSymTilde{m\kappa}$ spans the positive eigenvalues of $\Pi_{\HTilde{m\kappa}} \Qoperator \Pi_{\HTilde{m\kappa}}$, its maximum eigenvalue is contained in $\BasisSymTilde{m\kappa}$ (We need not care about the negative eigenvalues). A further investigation shows that for every $\ket{\BasisSymTilde{m\kappa}_{\vec{j}}} \in \BasisSymTilde{m\kappa}$, the corresponding eigenvalue is
 \begin{align}
     & m\cdot \frac{\binom{m\kappa}{{(j_0,j_1,\ldots, j_{d-1})}}}{\binom{(m+1)\kappa}{(j_0 + \kappa,j_1,\ldots, j_{d-1})}} - \frac{j_0 + \kappa}{\kappa} + 1\\
     &= m\cdot \frac{\binom{m\kappa}{{(j_0,j_1,\ldots, j_{d-1})}}}{\binom{(m+1)\kappa}{(j_0 + \kappa,j_1,\ldots, j_{d_1})}} - \frac{j_0}{\kappa}\\
     &= m\cdot \frac{\binom{m\kappa}{j_0}}{\binom{(m+1)\kappa}{j_0 + \kappa}} - \frac{j_0}{\kappa},\\
     &\leq \frac{1}{(m+1)^{\kappa - 1}}&\text{by a similar argument, used in \cref{eq:rational_unforge_utility}}\\
     &\leq \frac{1}{2^{\kappa - 1}}\\
     &\leq \frac{2}{\secpar^{\log^{c-1}(\secpar)}} &\text{since $\kappa = (\log(\secpar))^c$, $c>1$.}\\
     &= \negl &\text{since, $c>1$.}
 \end{align}
Therefore, the largest eigenvalue of $\Pi_{\HTilde{m\kappa}} \Qoperator \Pi_{\HTilde{m\kappa}}$, which is attained by some eigenvector in $\BasisSymTilde{m\kappa}$, is also negligible.
 Observe that, the term for the eigenvalue of $\Pi_{\HTilde{m\kappa}} \Qoperator \Pi_{\HTilde{m\kappa}}$ (which is the same as the loss of the verifier) is very similar (up to a negligible factor) to the term for the expected utility in  \cref{eq:rational_unforge_utility} of the adversary in the proof of  \cref{prop: rational-unforge} on \cpageref{pf:prop: rational-unforge}. 
%Using elementary calculus, one can show that the term above is maximized for $\vec{j} = (0,m\kappa,0\ldots,0)$.
%Hence, the maximum eigenvalue for $\Pi_S Q \Pi_S$ is $\frac{m}{\binom{(m+1)\kappa}{\kappa}}$ which is negligible for $m,\kappa \in \poly$.
We did not assume or require any computational assumptions on $\adv$, and therefore, the scheme $\pkqc$, is unconditionally \narpf.\footnote{We do not even require any bound on $m$ and $n$.}
%\qed 
\end{proof}

\begin{proof}[Proof of \cref{thm:nonadaptive secure}]\label{pf:thm:nonadaptive secure}
Combining \cref{prop: rational-unforge,thm:rational priv fair,thm:rational public fair}, we conclude that $\pkqc$ is \arnauf (resp. unconditionally \arnauf) (see \cref{definition:nonadapt_all-or-nothing_unforge,definition:unconditional_unforgeability}) and \narf (resp. unconditionally) (see \cref{definition:rational_secure_against_sabotage}). Hence, it is \nars (see \cref{definition:rational_secure}) (resp. unconditionally \nars) if the underlying $\prqc$ scheme is \auf (resp. \nauuf) (see \cref{definition:adapt_flex_unforge}) such that $\prqc.\verify$ is a rank-$1$ projective measurement.
%\qed 
\end{proof} 
% \begin{proof}[Proof of \cref{cor:unconditional_secure}]\label{pf:cor:unconditional_secure}
% %Note that $\qaruf$ implies \nauf. 
% On instantiating $\pkqc$ (see \cref{alg:ts}) with the private schemes in~\cite{JLS18} (or the simplified version in~\cite{BS19}) and~\cite{MS10}, and using \cref{thm:nonadaptive secure} with \cref{thm:qaruf,thm:qaruuf} we get the first two parts of the corollary, respectively. In order to get a stateful public coin scheme with private verification which is \nars unconditionally, we instantiate our scheme with the private coin scheme in~\cite{AMR20}. Since the private coin scheme is indistinguishable from the private scheme in~\cite{MS10}, we get an unconditionally \nars public coin scheme with private verification, on instantiating with~\cite{AMR20}. The last argument is elaborately discussed in the proof of \cref{thm:unconditional_secure} given in \cpageref{pf:thm:unconditional_secure}.
% %\qed 
% \end{proof}

\begin{proof}[Proof of \cref{thm:rational public fair}]\label{pf:thm:rational public fair}
We will first show that for any $\adv$ and $\forall m,k\in \NN$ and $n_1,n_2,\ldots,n_k\in \NN$  such that $\sum_{i=1}^k n_i =m$,
\begin{align}\label{eq:kto1}
    &\EE(\lanna(\nrsp^{\adv,\pkqc}_{\secpar, k, m, (n_1,\ldots,n_k)}))\\ &\leq \EE(\lanna(\nrsp^{\adv,\pkqc}_{\secpar, 1, m, (m+1)}).
\end{align}
Then in the second part, we will show that for every $m$ and $\adv$ and $\kappa$ of our choice,
\begin{equation}
    \EE(\lanna(\nrsp^{\adv,\pkqc}_{\secpar, 1, m, (m+1)}) \leq \negl,
\end{equation} 
for some negligible function $\negl$, which concludes the proof of the theorem. 
The first part of the proof is as follows.
Let $\adv$ be an arbitrary computationally unbounded adversary and without loss of generality, assume she submits a pure state. Let the $m\kappa$ register state submitted by the $\adv$ in Game~\ref{game:nonadapt_respendability} be $\ket{\beta}$. Recall that we initialize the wallet with a fresh public coin, the state of the wallet before receiving the adversary's coins is $\ket{\cent}$. 
% Hence, the parameter $refund$ is always $1$ in the respendability game $\nrsp^{\adv,\pkqc}$ (see Game~\ref{game:nonadapt_respendability}).
Let $\omega'$ be the post verified state of the wallet after receiving the $m$ coins from $\adv$.
Let $X,Y,Y'$ and $Z$ be random variables defined as,
\begin{equation}\label{eq:X-def-rsp}
     X = \begin{cases}
            1 &\text{if }  \pkqc.\Count_{\ket{\cent}}(\ketbra{\beta}) = m\\ 
            0 &\text{otherwise.}
         \end{cases}
 \end{equation}
\begin{align}\label{eq:boolean-variables}
 Y' := \frac{\prqc.\Count(\sk, \omega')}{\kappa},\\
 Y := \frac{\prqc.\Count(\sk, \ketbra{\cent} \tensor \ketbra{\beta})}{\kappa},\\
 Z := \pkqc.\Count_{bank}(\sk, \omega'),
 \end{align}where $\prqc$ is the private scheme we lift to $\pkqc$ in \cref{alg:ts}.
 By the definition of $\pkqc.\Count_{bank}$ (given in \cref{alg:ts}), $\forall b\in \{0,1\}$
\begin{equation}\label{eq:count-bank-relation-resp}
    \mathbb{E}(Z|X=b) = \mathbb{E}(Y'|X=b).
\end{equation}
Hence, $\EE(Z)=\EE(Y')$.

Let $X'$ be the random variable defined as 
\begin{equation}\label{eq:X'-def-rsp}
     X' = \begin{cases}
            1 &\text{if }  \pkqc.\Count_{\ket{\cent}}(\omega') = m+1\\ 
            0 &\text{otherwise.}
         \end{cases}
 \end{equation}
Consider $\nrsp^{\adv,\pkqc}_{\secpar, 1, m, (m+1)}$ in Game~\ref{game:nonadapt_respendability}. This is the case where all the received coins are spent in one transaction.
Therefore by \cref{eq:loss-respendable-def},
\[\lanna(\nrsp^{\adv,\pkqc}_{\secpar, 1, m, (m+1)}) = mX + 1 - X[X'|X=0](m+1) - [Z|X=0](1-X). \]
Since $X$ is a boolean random variable,
\begin{align}\label{eq:exp-loss_respendable}
    \EE&(\lanna(\nrsp^{\adv,\pkqc}_{\secpar, 1, m, (m)}))\\
    =m\EE&(X) + 1 - \EE(X)\EE(X'|X=1)(m+1) - \EE(1-X)\EE(Z|X=0) \\
    =m\Pr&[X=1] + 1 - \Pr[X=1]\cdot\Pr[X'=1|X=1](m+1) - \Pr[X=0]\EE(Z|X=0).\\
    =m\Pr&[X=1] + 1 - \Pr[X=1]\cdot\Pr[X'=1|X=1](m+1) - \Pr[X=0]\EE(Y'|X=0) &\text{By \cref{eq:count-bank-relation-resp}}.
\end{align}
Now consider the game $\nrsp^{\adv,\pkqc}_{\secpar, k, m (n_1,\ldots,n_k)}$ in Game~\ref{game:nonadapt_respendability}, where $\sum_{i=1}^k n_i = m+1$.
Let the post verified state $\omega'$ be grouped into $k$ chunks, $\omega'_1,\ldots,\omega'_k$.
Let $X'_1,\ldots,X'_k$ be random variables such that $\forall i\in[k]$,
\begin{equation}\label{eq:Y_i_def}
    X'_i = \begin{cases}
            1 &\text{if }  \pkqc.\Count_{\ket{\cent}}(\omega'_i) = n_i\\ 
            0 &\text{otherwise.}
         \end{cases}
\end{equation}
By \cref{eq:loss-respendable-def},
\[\lanna(\nrsp^{\adv,\pkqc}_{\secpar, k, m, (n_1,\ldots,n_k)})= 1 + mX - X\sum_{i=1}^k [X'_i|X=1] n_i - [Z|X=0](1-X).\]
Hence by similar arguments above,
\begin{align}\label{eq:exp-loss_respendable_k}
    &\EE(\lanna(\nrsp^{\adv,\pkqc}_{\secpar, k, m, (n_1,\ldots,n_k)}))\\
&=m\Pr[X=1] + 1 - \Pr[X=1]\cdot\sum_{i=1}^k\Pr[X'_i=1|X=1]n_i - \Pr[X=0]\EE(Y'|X=0).
\end{align}
Therefore by \cref{eq:exp-loss_respendable},
\begin{align}
    &\EE(\lanna(\nrsp^{\adv,\pkqc}_{\secpar, 1, m, (m)})) \\
    -& \EE(\lanna(\nrsp^{\adv,\pkqc}_{\secpar, k, m, (n_1,\ldots,n_k)})) \\
    &= \Pr[X=1]\left(\Pr[X'=1|X=1](m+1) - \sum_{i=1}^k\Pr[X'_i=1|X=1]n_i\right)\\
    &=\Pr[X=1]\sum_{i=1}^kn_i\left(\Pr[X'=1|X=1] - \Pr[X'_i=1|X=1]\right).
\end{align} 
Hence in order to prove \cref{eq:kto1}, we just need to show that for every $i\in [k]$,
\[\Pr[X'=1|X=1] \leq \Pr[X'_i=1|X=1].\]
This would follow because of the symmetric subspace of $j+1$ registers is a subspace of the symmetric subspace of $j$ registers.  
\begin{equation}\label{eq:symmetric-subs-prop}
    \Sym{i+1} \subset \mathcal{H}\tensor\Sym{i}, \forall i\in \NN.
\end{equation}
Fix $i\in[k]$. Let $\omega''$ be the state $\omega'$ conditioned on $X=1$, i.e., $(m+1)\kappa$-register state which is the post-verified state of the wallet after the symmetric subspace verification of $\ket{\beta}$ conditioned on the fact that $X=1$. Similarly let, $\omega''_i$ be the corresponding $n_i\kappa$ register state, conditioned on $X=1$. Due to the property of the symmetric subspace mentioned in \cref{eq:symmetric-subs-prop}, $\forall k\in\NN$ and $k'<k$, 
\begin{equation}\label{eq:sym_subspace_prop}
\Pi_{\Sym{k}}=\Pi_{\Sym{k}}(I\tensor\Pi_{\Sym{k'}})
\end{equation}
\begin{align}
    \Pr&[X'=1|X=1]\\ &=\|\Pi_{\Sym{(m+2)\kappa}}((\ketbra{\mill})^{\tensor \kappa}\tensor \omega'')\| \\
    &= \|\Pi_{\Sym{(m+2)\kappa}}\left(\left(I\tensor \Pi_{\Sym{(n_i + 1)\kappa}}\right) (\ketbra{\cent}\tensor \omega'')\right)\| &\text{By \cref{eq:sym_subspace_prop}}\\
    &\leq \|(I\tensor \Pi_{\Sym{(n_i+1)\kappa + 1}}) \left(\ketbra{\cent}\tensor \omega''\right)\| &\text{Since $\Pi_{\Sym{(m+2)\kappa}}$ is a projection.}\\
    &=  \Pr[X'_i = 1|X=1].
\end{align}
In the above equations, $I\tensor \Pi_{\Sym{(n_i+1)\kappa}}$ is the projection which acts as the symmetric subspace projector on the registers of the fresh coin and $\omega''_i$.

For the second part of the proof we rewrite the expected loss term in \cref{eq:exp-loss_respendable} as the following,
\begin{align}
    &\EE(\lanna(\nrsp^{\adv,\pkqc}_{\secpar, 1, m, (m)}))\\
    &= m\Pr[X=1] + 1 - (m+1)\Pr[X=1]\Pr[X'=1|X=1] - \Pr[X=0]\cdot\EE(Y'|X=0)\\
    &= m\Pr[X=1] + 1 - (m+1)\Pr[X=1]\Pr[X'=1|X=1] - \EE(Y') + \Pr[X=1]\cdot \EE[Y'|X=1]\\
    &= m\Pr[X=1] + 1 - (m+1)\Pr[X=1]\Pr[X'=1|X=1] - \EE(Y) + \Pr[X=1]\cdot \EE[Y'|X=1] \\
    &= (m\Pr[X=1] + 1) - (m+1)\Pr[X=1,X'=1] - \EE(Y) + \Pr[X=1]\cdot \EE[Y'|X=1]
    \label{eq:formula_loss_respendable}
\end{align} 
The second last equality holds due to the commutation property that we proved earlier (see \cref{eq: commutation}).
Next we analyze each of the four terms individually and then put the analysis together to conclude the proof.

By definition of $X$ in \cref{eq:X-def-rsp},
\[\Pr[X=1]=\tr(\Pi_{\Sym{(m+1)\kappa}}((\ketbra{\mill})^{\tensor \kappa}\tensor\ketbra{\beta})).\]
The post-verified state after the symmetric subspace verification of $\ket{\beta}$ given $X=1$, i.e., all the $m$ coins passed verification,   is
\[\frac{\Pi_{\Sym{(m+1)\kappa}}((\ketbra{\mill})^{\tensor \kappa}\tensor\ketbra{\beta})\Pi_{\Sym{(m+1)\kappa}}^\dagger}{\tr(\Pi_{\Sym{(m+1)\kappa}}((\ketbra{\mill})^{\tensor \kappa}\tensor\ketbra{\beta}))}=\frac{\Pi_{\Sym{(m+1)\kappa}}((\ketbra{\mill})^{\tensor \kappa}\tensor\ketbra{\beta})\Pi_{\Sym{(m+1)\kappa}}^\dagger}{\Pr[X=1]}.\]
Hence,
\begin{align}
    &\Pr[X'=1|X=1]\\
    &=\tr\left( \Pi_{\Sym{(m+2)\kappa}}\left(\frac{(\ketbra{\mill})^{\tensor \kappa}\tensor\Pi_{\Sym{(m+1)\kappa}}((\ketbra{\mill})^{\tensor \kappa}\tensor\ketbra{\beta})\Pi_{\Sym{(m+1)\kappa}}^\dagger}{\Pr[X=1]}\right)\Pi_{\Sym{(m+2)\kappa}}\right)\\
    &= \tr \left( \frac{\Pi_{\Sym{(m+2)\kappa}}\left((\ketbra{\mill})^{\tensor \kappa}\tensor(\ketbra{\mill})^{\tensor \kappa}\tensor\ketbra{\beta}\right)\Pi_{\Sym{(m+2)\kappa}}}{\Pr[X=1]}\right)\\
    &=\tr\left( \frac{\Pi_{\Sym{(m+2)\kappa}}\left((\ketbra{\mill})^{\tensor \kappa}\tensor(\ketbra{\mill})^{\tensor \kappa}\tensor\ketbra{\beta}\right)}{\Pr[X=1]}\right)\\
    \implies &\Pr[X=1,X'=1]=\tr(\Pi_{\Sym{(m+2)\kappa}}\left((\ketbra{\mill})^{\tensor 2\kappa}\tensor\ketbra{\beta}\right)). \label{eq:X'-X}
\end{align}
The second equality holds due to \cref{eq:symmetric-subs-prop}.
Expressing the state $\ket{\beta}$ in the basis extended from the symmetric subspace basis $\BasisSym{m\kappa}$ (see \cref{eq:basis_def} in \cref{subsec:notations}), we get,
\[\ket{\beta}= \sum_{\vec{j}\in \mathcal{I}_{d,m\kappa}} a_{\vec{j}}\ket{\BasisSym{m\kappa}_{\vec{j}}} + \ket{Sym^{\perp}},\]
where $\ket{Sym^{\perp}}\in (\Sym{m\kappa})^\perp,$ and $\sum_{\vec{j}\in  \mathcal{I}_{d,m\kappa}}|a_j|^2\leq 1$.
We would encourage the reader to go through the notations in \cref{subsec:notations}.
A simple calculation similar to the proof of \cref{lemma:eigenbasis} (see \cref{item:lambda_max_P} on \cpageref{pf:lemma:eigenbasis}) shows that for every $\vec{j}\in \mathcal{I}_{d,m\kappa}$,
\begin{align}
    &\Pi_{\Sym{(m+2)\kappa}}\left((\ket{\mill})^{\tensor 2\kappa}\tensor\ket{\BasisSym{m\kappa}_{\vec{j}}}\right)\\
    &=\sqrt{\frac{\binom{m\kappa}{j_0}}{\binom{(m+2)\kappa}{j_0+2\kappa}}}\ket{\BasisSym{(m+2)\kappa}_{(j_0+2\kappa,j_1,\ldots,j_{d-1})}} \label{eq:basis_probab}
\end{align}  
Let,
\begin{equation}\label{eq:basis_prob2_symbol}
    P_{\vec{j},m\kappa,2} := \frac{\binom{m\kappa}{j_0}}{\binom{(m+2)\kappa}{j_0+2\kappa}}.
\end{equation}
Moreover since $\ket{Sym^{\perp}}\in (\Sym{m\kappa})^\perp$, 
\[\Pi_{\Sym{(m+2)\kappa}}\left((\ket{\mill})^{\tensor 2\kappa}\tensor\ket{Sym^{\perp}}\right)=0.\]
Hence, plugging them back in \cref{eq:X'-X}, we get,
\begin{align}
    \Pr[X=1,X'=1]&=\tr\left(\Pi_{\Sym{(m+2)\kappa}}\left((\ketbra{\mill})^{\tensor 2\kappa}\tensor\ketbra{\beta}\right)\Pi_{\Sym{(m+2)\kappa}}^\dagger\right)\\
    &=\left\|\sum_{\vec{j}\in \mathcal{I}_{d,m\kappa}}a_{\vec{j}}\Pi_{\Sym{(m+2)\kappa}}(\ket{\mill}^{\tensor 2\kappa}\tensor\ket{\beta})\right\|^2\\
    &=\left\|\sum_{\vec{j}\in \mathcal{I}_{d,m\kappa}}a_{\vec{j}}\sqrt{\frac{\binom{m\kappa}{j_0}}{\binom{(m+2)\kappa}{j_0+2\kappa}}}\ket{\BasisSym{(m+2)\kappa}_{(j_0+2\kappa,j_1,\ldots,j_{d-1})}}\right\|^2&\text{ By \cref{eq:basis_probab}}\\
    &=\sum_{\vec{j}\in \mathcal{I}_{d,m\kappa}}|a_{\vec{j}}|^2 \frac{\binom{m\kappa}{j_0}}{\binom{(m+2)\kappa}{j_0+2\kappa}}\\
    &=\sum_{\vec{j}\in \mathcal{I}_{d,m\kappa}}|a_{\vec{j}}|^2 \frac{\binom{m\kappa}{j_0}}{\binom{(m+2)\kappa}{j_0+2\kappa}}\\
    &=\sum_{\vec{j}\in \mathcal{I}_{d,m\kappa}}|a_{\vec{j}}|^2 P_{\vec{j},m\kappa,2}
    \label{eq:prob_2}
\end{align}

By a very similar argument we get,
\begin{align}
    &\Pr[X=1]\\
    &=\|\Pi_{\Sym{(m+1)\kappa}}\left((\ketbra{\mill})^{\tensor \kappa}\tensor\ketbra{\beta}\right)\|^2\\
    &=\sum_{\vec{j}\in \mathcal{I}_{d,m\kappa}}|a_{\vec{j}}|^2 P_{\vec{j},m\kappa}
    \label{eq:prob_1}
\end{align}
where \[P_{\vec{j},m\kappa}:=\frac{\binom{m\kappa}{j_0}}{\binom{(m+1)\kappa}{j_0+\kappa}}\quad \forall \vec{j} \in \mathcal{I}_{d,m\kappa}.\] 
Next we analyze the quantity, $\Pr[X=1]\EE(Y'|X=1)$.

Let $\omega''$ be the post verified state of the wallet after the symmetric subspace verification of $\ket{\beta}$ (the state submitted by the adversary) conditioned on $X=1$, i.e. all coins passed verification.
Hence,
\begin{align}\label{eq:def-omega-double}
    \omega''&= \frac{\Pi_{\Sym{(m+1)\kappa}}((\ketbra{\mill})^{\tensor \kappa}\tensor\ketbra{\beta})\Pi_{\Sym{(m+1)\kappa}}^\dagger}{\tr(\Pi_{\Sym{(m+1)\kappa}}((\ketbra{\mill})^{\tensor \kappa}\tensor\ketbra{\beta}))}\\
    &=\frac{\Pi_{\Sym{(m+1)\kappa}}((\ketbra{\mill})^{\tensor \kappa}\tensor\ketbra{\beta})\Pi_{\Sym{(m+1)\kappa}}^\dagger}{\Pr[X=1]}. 
\end{align}

$\forall m\in \NN$, define the following operator
\begin{equation}\label{eq:counter-def}
    Counter_m=\sum_{i=1}^{m\kappa} I_{i-1}\tensor \ketbra{\mill} \tensor I_{m\kappa-i}.
\end{equation}
It is easy to see that $\forall m\in \NN, \vec{j} \in \mathcal{I}_{d,m\kappa}$,
\begin{equation}\label{eq:basis-counter}
    Counter_m\ket{\BasisSym{m\kappa}_{\vec{j}}}= \left(\prqc.\Count(\sk,\ketbra{\BasisSym{m\kappa}_{\vec{j}}})\right)\ket{\BasisSym{m\kappa}_{\vec{j}}} =j_0\ket{\BasisSym{m\kappa}_{\vec{j}}}.
\end{equation}

Clearly, this represents the hamiltonian for $\prqc.\Count$ in the private scheme $\prqc$.
Hence by \cref{eq:def-omega-double}
\begin{align}
    &\EE(Y'|X=1)\\
    =&\frac{\prqc.\Count(\sk,\omega'')}{\kappa}\\
    =&\tr\left(\frac{Counter_{m+1}}{\kappa} \omega'' \right)\\
    =&\tr\left(\frac{Counter_{m+1}}{\kappa}\cdot \frac{\Pi_{\Sym{(m+1)\kappa}}((\ketbra{\mill})^{\tensor \kappa}\tensor\ketbra{\beta})\Pi_{\Sym{(m+1)\kappa}}^\dagger}{\Pr[X=1]}\right) \\
    =&\frac{(\bra{\mill}^{\tensor\kappa}\tensor \bra{\beta})\Pi_{\Sym{(m+1)\kappa}}^\dagger \frac{Counter_{m+1}}{\kappa} \Pi_{\Sym{(m+1)\kappa}}(\ket{\mill}^{\tensor\kappa}\tensor\ket{\beta})}{\Pr[X=1]}\\
    =&\frac{\left(\sum_{\vec{j}\in \mathcal{I}_{d,m\kappa}}\overline{a_{\vec{j}}}\sqrt{\frac{\binom{m\kappa}{j_0}}{\binom{(m+1)\kappa}{j_0+\kappa}}}\bra{\BasisSym{(m+1)\kappa}_{(j_0+\kappa,j_1,\ldots,j_{d-1})}}\right)\frac{Counter_{m+1}}{\kappa}\left(\sum_{\vec{j}\in \mathcal{I}_{d,m\kappa}}a_{\vec{j}}\sqrt{\frac{\binom{m\kappa}{j_0}}{\binom{(m+1)\kappa}{j_0+\kappa}}}\ket{\BasisSym{(m+1)\kappa}_{(j_0+\kappa,j_1,\ldots,j_{d-1})}}\right)}{\Pr[X=1]}\\
    =&\frac{\left(\sum_{\vec{j}\in \mathcal{I}_{d,m\kappa}}\overline{a_{\vec{j}}}\sqrt{\frac{\binom{m\kappa}{j_0}}{\binom{(m+1)\kappa}{j_0+\kappa}}}\bra{\BasisSym{(m+1)\kappa}_{(j_0+\kappa,j_1,\ldots,j_{d-1})}}\right)\left(\sum_{\vec{j}\in \mathcal{I}_{d,m\kappa}}a_{\vec{j}}\sqrt{\frac{\binom{m\kappa}{j_0}}{\binom{(m+1)\kappa}{j_0+\kappa}}}\frac{Counter_{m+1}}{\kappa}\left(\ket{\BasisSym{(m+1)\kappa}_{(j_0+\kappa,j_1,\ldots,j_{d-1})}}\right)\right)}{\Pr[X=1]}\\
    =&\frac{\left(\sum_{\vec{j}\in \mathcal{I}_{d,m\kappa}}\overline{a_{\vec{j}}}\sqrt{\frac{\binom{m\kappa}{j_0}}{\binom{(m+1)\kappa}{j_0+\kappa}}}\bra{\BasisSym{(m+1)\kappa}_{(j_0+\kappa,j_1,\ldots,j_{d-1})}}\right)\left(\sum_{\vec{j}\in \mathcal{I}_{d,m\kappa}}a_{\vec{j}}\sqrt{\frac{\binom{m\kappa}{j_0}}{\binom{(m+1)\kappa}{j_0+\kappa}}}(\frac{j_0+\kappa}{\kappa})\left(\ket{\BasisSym{(m+1)\kappa}_{(j_0+\kappa,j_1,\ldots,j_{d-1})}}\right)\right)}{\Pr[X=1]}\\
    \implies & \EE(Y'|X=1)\Pr[X=1] =\sum_{\vec{j}\in \mathcal{I}_{d,m\kappa}}|a_{\vec{j}}|^2\frac{\binom{m\kappa}{j_0}}{\binom{(m+1)\kappa}{j_0+\kappa}}\frac{j_0 + \kappa}{\kappa}=\sum_{\vec{j}\in \mathcal{I}_{d,m\kappa}}|a_{\vec{j}}|^2P_{\vec{j},m\kappa}\frac{j_0+\kappa}{\kappa}.
    \label{eq:exp_ref_giv_succ} 
\end{align}
The last equality follows from the definition of $P_{\vec{j},m\kappa}$. In the above derivation, we have also used \cref{eq:basis-counter} in the second last equality, \cref{eq:def-omega-double} in third equality, and the expression for $\ket{\beta}$ in the symmetric subspace basis, in the fifth equality.

Lastly, we analyze the quantity $\EE(Y)$.
First, we express the submitted state $\ket{\beta}$ in the standard product basis extended from the state $\ket{\mill}$ (see \cref{item:product_basis_definition} in \cref{subsec:notations}).
\begin{equation}
    \ket{\beta}=\sum_{\vec{i}\in \ZZ_d^{m\kappa}}\tilde{a}_{\vec{i}} \ket{\phi_{i_1},\ldots \phi_{i_{m\kappa}}}
\end{equation}
Note that, $\forall \vec{i}\in \ZZ_d^{m\kappa}$, and $\forall \vec{j} \in \mathcal{I}_{d,m\kappa}$,
\begin{equation}\label{eq:bases-relation}
    \braket{\BasisSym{m\kappa}_{\vec{j}}}{\phi_{i_1},\ldots \phi_{i_{m\kappa}}} =\begin{cases}
    =\frac{1}{\sqrt{\binom{m\kappa}{\vec{j}}}} &\text{if } T(\vec{i})=\vec{j}\\
    =0 &\text{Otherwise.}    
    \end{cases}
\end{equation}
Hence, $\forall \vec{j} \in \mathcal{I}_{d,m\kappa}$,
\begin{align} 
|a_{\vec{j}}| = \left|\braket{\beta}{\BasisSym{m\kappa}_{\vec{j}}}\right|&= \left|\sum_{\vec{i}:T(\vec{i})=\vec{j}} \frac{\tilde{a}_i}{\sqrt{\binom{m\kappa}{\vec{j}}}}\right|\\
&=\left|\frac{\sum_{\vec{i}:T(\vec{i})=\vec{j}} \tilde{a}_i}{\sqrt{\binom{m\kappa}{\vec{j}}}}\right|\\
&\leq \left(\sum_{\vec{i}:T(\vec{i})=\vec{j}} |\tilde{a}_i|^2\right)^{1/2}\cdot\left(\sum_{\vec{i}:T(\vec{i})=\vec{j} }\frac{1}{\binom{m\kappa}{\vec{j}}}\right)^{1/2} &\text{By Cauchy-Schwartz}\\
&=\left(\sum_{\vec{i}:T(\vec{i})=\vec{j}} |\tilde{a}_i|^2\right)^{1/2}&\text{Since $\# \vec{i}$ with $T(\vec{i})=j$ is $\binom{m\kappa}{\vec{j}}$. } \\
\implies &|a_{\vec{j}}|^2\leq  \sum_{\vec{i}:T(\vec{i})=\vec{j}} |\tilde{a}_i|^2.
\label{eq:basis-change} 
\end{align}

By definition (see \cref{eq:boolean-variables})
\begin{align}
    \EE(Y)&=\frac{\EE(\prqc.\Count(\sk,\ketbra{\mill}^\kappa\tensor\ketbra{\beta}))}{\kappa}\\
    &=\frac{\sum_{\vec{i}\in \ZZ_d^{m\kappa}}|\tilde{a}_{\vec{i}}|^2(\kappa+(T(\vec{i})_0))}{\kappa}\\
    &=\frac{\sum_{\vec{j}\in \mathcal{I}_{d,m\kappa}} \sum_{\vec{i}:T(\vec{i})=\vec{j}} |\tilde{a}_{\vec{i}}|^2(j_0+\kappa)}{\kappa}\\
    &= \frac{\sum_{\vec{j}\in \mathcal{I}_{d,m\kappa}} (j_0+\kappa)\sum_{\vec{i}:T(\vec{i})=\vec{j}} |\tilde{a}_{\vec{i}}|^2}{\kappa}\\
    &\geq \frac{\sum_{\vec{j}\in \mathcal{I}_{d,m\kappa}} (j_0+\kappa) |a_{\vec{j}}|^2}{\kappa} &\text{By \cref{eq:basis-change}}
    \label{eq:exp_Y}
\end{align}
The last inequality uses 
Plugging \cref{eq:prob_1}, \cref{eq:prob_2}, \cref{eq:exp_ref_giv_succ} and \cref{eq:exp_Y} in the formula of \newline$\EE(\lanna(\nrsp^{\adv,\pkqc}_{\secpar, 1, m, (m)}))$ in \cref{eq:formula_loss_respendable}, we get
\begin{align}
    &\EE(\lanna(\nrsp^{\adv,\pkqc}_{\secpar, 1, m, (m)}))\\
    &\leq m\left(\sum_{\vec{j}\in \mathcal{I}_{d,m\kappa}}|a_{\vec{j}}|^2P_{\vec{j},m\kappa}\right) + 1 -(m+1)\left(\sum_{\vec{j}\in \mathcal{I}_{d,m\kappa}}|a_{\vec{j}}|^2P_{\vec{j},m\kappa}\right)\\
    &+ \left(\sum_{\vec{j}\in  \mathcal{I}_{d,m\kappa}}|a_{\vec{j}}|^2P_{\vec{j},m\kappa}\frac{j_0+\kappa}{\kappa}\right) - \left(\sum_{\vec{j}\in \mathcal{I}_{d,m\kappa}} |a_{\vec{j}}|^2\frac{(j_0+\kappa)}{\kappa}\right)\\
    &=\sum_{\vec{j}\in  \mathcal{I}_{d,m\kappa}}|a_{\vec{j}}|^2\left(mP_{\vec{j}, m\kappa}+1 -(m+1)P_{\vec{j}, m\kappa, 2} + P_{\vec{j}, m\kappa}\frac{(j_0 + \kappa)}{\kappa} -\frac{(j_0 + \kappa)}{\kappa}\right) \\
    &=\sum_{\vec{j}\in  \mathcal{I}_{d,m\kappa}}|a_{\vec{j}}|^2\left((m+1)(P_{\vec{j}, m\kappa}-P_{\vec{j}, m\kappa, 2}) - \frac{j_0}{\kappa}(1-P_{\vec{j}, m\kappa})\right)\\
    &=\sum_{\vec{j}\in  \mathcal{I}_{d,m\kappa}}|a_{\vec{j}}|^2\left((m+1)\left(\frac{\binom{m\kappa}{j_0}}{\binom{(m+1)\kappa}{j_0 + \kappa}} - \frac{\binom{m\kappa}{j_0}}{\binom{(m+2)\kappa}{j_0 + 2\kappa}} \right) - \left(1-\frac{\binom{m\kappa}{j_0}}{\binom{(m+1)\kappa}{j_0 + \kappa}}\right)\frac{j_0}{\kappa}\right)\\
    &\leq \max_{\vec{j}\in  \mathcal{I}_{d,m\kappa}}\left((m+1)\left(\frac{\binom{m\kappa}{j_0}}{\binom{(m+1)\kappa}{j_0 + \kappa}} - \frac{\binom{m\kappa}{j_0}}{\binom{(m+2)\kappa}{j_0 + 2\kappa}} \right) - \left(1-\frac{\binom{m\kappa}{j_0}}{\binom{(m+1)\kappa}{j_0 + \kappa}}\right)\frac{j_0}{\kappa}\right)&\text{since}\sum_{\vec{j}\in  \mathcal{I}_{d,m\kappa}} |a_{\vec{j}}|^2\leq1\\
    &=\max_{0\leq j_0 \leq m\kappa}\left((m+1)\left(\frac{\binom{m\kappa}{j_0}}{\binom{(m+1)\kappa}{j_0 + \kappa}} - \frac{\binom{m\kappa}{j_0}}{\binom{(m+2)\kappa}{j_0 + 2\kappa}} \right) - \left(1-\frac{\binom{m\kappa}{j_0}}{\binom{(m+1)\kappa}{j_0 + \kappa}}\right)\frac{j_0}{\kappa}\right)\\
\end{align}
For $0\leq j_0< \frac{\kappa}{2}$, the above expression is negligible since \[\frac{\binom{m\kappa}{j_0}}{\binom{(m+1)\kappa}{j_0 + \kappa}} =\frac{\prod_{r=1}^\kappa(j_0+r)}{\prod_{r=1}^\kappa(m\kappa+r)}<\frac{\prod_{r=1}^\kappa(\frac{\kappa}{2}+r)}{\prod_{r=1}^\kappa(m\kappa+r)}\leq  \left(\frac{3\kappa/2}{(m+1)\kappa}\right)^\kappa=\left(\frac{3}{2(m+1)}\right)^\kappa,\] which is negligible  (at most $\left(\frac{3}{4}\right)^\kappa$) for all $m\geq1$ and $\kappa$ of our choice. For $j_0 = m\kappa$ the expression evaluates to $0$. For the remaining case, i.e., $\frac{\kappa}{2} \leq j_0 \leq m\kappa-1$, we used mathematica and found out that for any $m$ and $\kappa\geq4$, the expression is negative. 
% %
% % FindMaxValue[{(m + 1)*(Binomial[m*k, j]/Binomial[(m + 1)*k, j + k] - 
%             Binomial[m*k, j]/Binomial[(m + 2)*k, j + 2*k]) - 
%       (1 - Binomial[m*k, j]/Binomial[(m + 1)*k, j + k])*(j/k), 
%   k/2 <= j <= m*k - 1 && k >= 10}, {m, j, k}]
%-0.0129873 
Hence the expression above, which is an upper bound for the expected loss, is negligible for every $0\leq j_0\leq m\kappa$, for large enough $\kappa$ and for any $m\in\NN$. Therefore, for any $m\in \NN$ and $\adv$, and $\kappa$ of our choice (poly-logarithmic) the expected loss, $\EE(\lanna(\nrsp^{\adv,\pkqc}_{\secpar, 1, m, (m)}))$ is negligible, which finishes the second part of the proof. Since the adversary $\adv$ was arbitrary and potentially computationally unbounded, we also conclude that $\pkqc$ is $\nrrsp$ unconditionally. 
%\qed 
\end{proof}
\begin{proof}[Proof of \cref{prop:sing-multi-rational-fair}]\label{pf:prop:sing-multi-rational-fair}
We show a reduction from the multiverifier rational security against private sabotage to single-verifier security against private sabotage for any comparison public money scheme $\MS$ with private verification. The main intuition is that for every private sabotage attack against multiple verifiers, suppose $k$ many, we can construct a single-verifier private sabotage attack, by choosing one of the multiple verifiers, uniformly at random, and simulating the other $k-1$ verifiers and \adv. We can show that the single-verifier attack would be QPT if the multiverifier attack is QPT. Moreover, the expected loss of such a single-verifier attack is same as the multiverifier attack, up to a fraction of $k$. Recall that, $k$ is polynomial even against computationally unbounded adversaries (see Game~\ref{game:multi_nonadapt_rational-fair}). If the scheme, $\MS$ is \narpf, then for any QPT multiverifier attack, the expected loss of the corresponding single-verifier attack, described above, must be negligible. Hence, the expected loss due to the multiverifier private sabotage attack must also be negligible, since $k$ is polynomial. Similarly if $\MS$ is unconditionally \narpf, then it is also \mnarpf unconditionally.
%In particular, for every multiverifier sabotage adversary $\adv$ that targets $k$ verifiers (polynomially many), and incurs an expected loss $z$, we will construct a single-verifier adversary, $\widetilde{\adv}$, which incurs an expected loss, $\frac{z}{k}$. 
%Since the scheme $\pkqc$ is \narpf (see \cref{thm:rational priv fair}), $\frac{z}{k}$ must be negligible. Recall that, $k$ is polynomial even against computationally unbounded adversaries. Hence, $z$ must be negligible and hence, we conclude that the scheme is also \mnarpf.

Let $\adv$ be any multiverifier private sabotage adversary that targets $k$ verifiers. Suppose, it submits the state $\sigma_j$ to the $j^{th}$ verifier. Let $L_j(\adv)$ (see \cref{eq:sing_loss_def}) be the loss incurred on the $j^{th}$ verifier due to \adv. By \cref{eq:multi_rational_loss_def}, the combined loss due to \adv,
\[L_{\text{multi-ver}}(\adv) = \sum_{j=1}^k L_j(\adv).\]

We construct a single-verifier adversary $\widetilde{\adv}$ as follows. $\widetilde{\adv}$ picks $i\in_R [k]$, uniformly at random. It simulates $\adv$ using the $\bank(\sk)$, and $\verify_\sk$ oracles. For all $j \in [k]\setminus\{i\}$, instead of submitting $\sigma_j$ to the honest verifiers, it simulates the verifier's action or the public verification $\pkqc.\Count_{\ket{\cent}}$ on $\sigma_j$. This can be done by getting a fresh coin $\cent$ from the $\pkqc.\bank$ oracle and then simulating  $\pkqc.\Count_{\ket{\cent}}(\sigma_j)$ (comparison based verification) using it. It is essential to know the outcomes of the public verifications, the verifiers' actions, since $\adv$'s strategy may depend on these outcomes. It submits $\sigma_i$ to the verifier. Clearly, the loss of the verifier incurred due to $\widetilde{\adv}$ (see \cref{eq:loss_def}), $L(\widetilde{\adv})$ should be the same as $L_i(\adv)$, the loss of the $i^{th}$ verifier, incurred due to \adv. Therefore,
\begin{align}
    \mathbb{E}(L(\widetilde{\adv}))&=\frac{1}{k}\sum{i=1}^k\mathbb{E}(L_i(\adv))\\
    &=\frac{1}{k}\mathbb{E}(\sum_{i=1}^k L_i(\adv))\\
    &=\frac{1}{k}\mathbb{E}(L_{\text{multi-ver}}(\adv)).\\
\end{align}
Note that if $\adv$ is QPT, then $\widetilde{\adv}$ is also QPT. 
 If the scheme $\MS$ is \narpf (resp. unconditionally \narpf), the it is also \mnarpf.

%\qed 
\end{proof}

\begin{proof}[Proof of \cref{prop: multi_ver-unforge}]\label{pf:prop: multi_ver-unforge}
Let $\adv$ be any nonadaptive rational forger (QPT or not depending on whether the underlying private scheme $\prqc$, is unconditionally secure or not), which attacks $k$ verifiers. Recall that, $k$ is polynomial even against computationally unbounded adversaries (see Game~\ref{game:multi_nonadapt_unforge}). We will construct a private sabotage adversary $\widetilde{\adv}$ which basically just simulates $\adv$ (hence, is a QPA if $\adv$ is a QPA), such that if the underlying private scheme \prqc, is \mnauf (resp. unconditionally \mnauf), then for all such QPA (resp. computationally unbounded) $\adv$, there exists a negligible function $\widetilde{\negl}$, such that the following holds,
\begin{equation}\label{eq:loss-utility_relation}
    \mathbb{E}(L_{\text{multi-ver}}(\widetilde{\adv})) \geq \mathbb{E}(\uanmna(\adv)) + \widetilde{\negl},
\end{equation}
where $L_{\text{multi-ver}}(\widetilde{\adv})$ and $\uanmna(\adv)$ is the corresponding multiverifier loss and multiverifier utility as defined in \cref{eq:multi_rational_loss_def,eq:multi_rational_utility_def} respectively.
Since, $\pkqc$ is unconditionally \mnarpf (see \cref{cor:multi_ver-rational-fair}), $\mathbb{E}(L_{\text{multi-ver}}(\widetilde{\adv}))$ must be negligible. Hence, by \cref{eq:loss-utility_relation}, we conclude that $\mathbb{E}(\uanmna(\adv))$ must be negligible too. Therefore, if the scheme $\prqc$ is \mnauf (respectively, unconditionally \mnauf), then the scheme $\pkqc$ (that we lift in \cref{alg:ts}) is \mnaruf.

Suppose, $\adv$ submits $m_j$ alleged public coins to the $j^{th}$ verifier, and the combined state of the coins submitted (over $m_j\kappa$ registers) is $\sigma_j$. Let $\widetilde{N}$ be the random variable, which denotes the number of times $\verify_\sk$ (same as $\verify_{bank}$) accepted the alleged coins submitted by \adv. Let $N$ denote the random variable which denotes the number of times $\pkqc.\bank$ was queried by \adv. Let the post verification state of the verifier's wallet, after the public verification of the $m_j$ alleged coins, submitted by $\adv$, be $\omega'_j$. Let $N_j$ be the random variable denoting $\pkqc.\Count_{bank}(\sk, \omega'_j)$, for all $j \in [k]$.
Note that by definition of $\pkqc.\Count_{\ket{\cent}}$ in \cref{alg:ts},
\begin{equation}\label{eq:expected_count}
    \mathbb{E}(N_j) = \mathbb{E}({\pkqc.\Count_{bank}(\sk, \omega'_j)}{\kappa}) = \mathbb{E}({\prqc.\Count(\sk, \omega'_j)}{\kappa}).
\end{equation}
For every $j \in [k]$, let $X_j$ be the random variable such that
\begin{equation}
X_j = \begin{cases}
        1 &\text{if $\pkqc.\Count_{\ket{\cent}}(\sigma_j) = m_j$}\\
        0 &\text{otherwise.}
      \end{cases}
\end{equation}
We construct a private sabotage adversary $\widetilde{\adv}$ which attacks $k$ verifiers, as follows. $\widetilde{\adv}$ simulates $\adv$ and prepares the states $\sigma_1,\ldots,\sigma_k$ to be submitted to the $k$ verifiers in the same order as $\adv$ did. Note that, $\widetilde{\adv}$ has access to all the oracles as $\adv$ has, and hence can simulate \adv.

According to \cref{eq:multi_rational_utility_def}, the utility of $\adv$ is
\begin{equation}\label{eq:multi-utility}
    \uanmna(\adv) = \sum_{j=1}^k (m_jX_j + \widetilde{N} - N).
\end{equation}

By definition of multiverifier loss (see \cref{eq:multi_rational_loss_def}) and since every receiving wallet is initialized to $\ket{\cent}$ (see user manual in \cref{subsec:user_manual}), the loss incurred by $\widetilde{\adv}$ is 
\begin{align}
    L_{\text{multi-ver}}(\widetilde{\adv}) &= \sum_{j=1}^k(m_j X_j - N_j + 1).
\end{align}

Next, using multiverifier nonadaptive unforgeability of the underlying private scheme $\prqc$, we derive a relation between $N_j$'s (refunds obtained form $\omega'_j$), $\widetilde{N}$ (the number of coins accepted by $\pkqc.\verify_{bank}$) and $N$ (the number of times $\pkqc.\bank$ was called).

For the derivation, we construct a multiverifier nonadaptive forger $\bdv$ against the private scheme $\prqc$, which simulates $\adv$, as follows. The $\bank$ oracle is simulated by taking $\kappa$ private coins using $\prqc.\bank$. The $k$ verifiers targeted by \adv, can be simulated by getting $k$ additional public coins or $\kappa k$ private coins using the $\prqc.\bank$ oracle. In order to simulate every call to $\verify_\sk$ (same as $\verify_{bank}$) \bdv, submits the money to be queried, to a bank branch. The branch outputs $\prqc.\Count$ value on the submitted alleged coins, using which the $\bdv$ simulates $\pkqc.\Count_{bank}$, according to the definition of $\pkqc.\Count_{bank}$ in \cref{alg:ts}. Finally, the combined state of $\omega'_1, \ldots,\omega'_k$ is submitted to the last branch. Suppose, $\verify_\sk$ was called $r$ many times by $\adv$, then $\bdv$ submits to $r+1$ bank branches. Note that if $\adv$ is QPA, then $\bdv$ is also a QPA.

Let the money states queried by \adv, same as the ones submitted to the bank branches by \bdv, be 
\[\widetilde{\sigma_1},\ldots,\widetilde{\sigma_{r}} .\]
For every $i\in [r]$, let $Y_i$'s be the random variable defined as
\begin{equation}
    Y_i := \begin{cases} 
                1 &\text{if $\pkqc.\Count_{bank}(\widetilde{\sigma_{i}}) = 1$}\\
                0 &\text{otherwise.}
           \end{cases}
\end{equation}
By definition, 
\begin{equation}\label{eq:sum_random_var_def} 
\widetilde{N} = \sum_{i=1}^r Y_i.
\end{equation}
The sum of the expected utilities $U_i(\bdv)$ (as defined in \cref{eq:sing_rational_utility_def}), due to the first $r$ submissions, is 
\begin{align}
\sum_{i=1}^{r}\mathbb{E}(U_i(\bdv)) &= \sum_{i=1}^{r} \prqc.\Count(\sk, \widetilde{\sigma_i})\\
&= \sum_{i=1}^{r} \kappa \cdot \pkqc.\Count_{bank}(\sk, \widetilde{\sigma_i}) &\text{(by definition of $\Count_{bank}$ in \cref{alg:ts})}\\
&= \sum_{i=1}^{r} \kappa \cdot \mathbb{E}(Y_i) = \mathbb{E}(\kappa\widetilde{N}).&\text{(by \cref{eq:sum_random_var_def})}
\end{align}

The expected utility due to the last submission is clearly, \[\mathbb{E}(\sum_{j=1}^k \prqc.\Count(\omega'_j)).\] The number of private coins that $\bdv$ takes from $\prqc.\bank$ is clearly, $\kappa (N + k) $.
Hence, the net expected utility of \bdv, against $\prqc$ \[ \mathbb{E}(\uanmna(\bdv)) = \mathbb{E}(\sum_{j=1}^k \prqc.\Count(\omega'_j) + \kappa \widetilde{N} -  \kappa N -\kappa k).\] Since $\prqc$ is \mnauf, there exists a negligible function $\negl$ such that
\begin{align}
    &\mathbb{E}(\uanmna(\bdv)) = \mathbb{E}(\sum_{j=1}^k \prqc.\Count(\omega'_j) + \kappa (\widetilde{N} - N - k)) \leq \negl.\\
    &\implies \mathbb{E}(\sum_{j=1}^k \prqc.\Count(\omega'_j)) \leq \mathbb{E}(\kappa(N+k-\widetilde{N})) + \negl\\
    &\implies \mathbb{E}(\sum_{j=1}^k N_j \kappa) \leq \kappa(\mathbb{E}(N-\widetilde{N}) + k) + \negl &\text{by \cref{eq:expected_count}}.\\
    &\implies \mathbb{E}(\sum_{j=1}^k N_j) \leq \mathbb{E}(N-\widetilde{N}) + k + \widetilde{\negl},
    \label{eq:multi-priv-bound}
\end{align} where $\widetilde{\negl} = \frac{\negl}{\kappa}$. 
Note that, if $\prqc$ is unconditionally \mnauf, the above equation holds even if $\adv$ is computationally unbounded. Next, we use this relation to prove \cref{eq:loss-utility_relation}, which concludes the proof. 
The expected loss due to $\widetilde{\adv}$ is 
\begin{align}
    &\mathbb{E}(L_{\text{multi-ver}}(\widetilde{\adv})) \\&= \mathbb{E}(\sum_{j=1}^k m_j X_j) - \mathbb{E}(\sum_{j=1}^k N_j) + k\\
    &\geq \mathbb{E}(\sum_{j=1}^k m_j X_j) - (\mathbb{E}(N -\widetilde{N}) + k) + \widetilde{\negl} + k &\text{by \cref{eq:multi-priv-bound}}\\
    &= \mathbb{E}(\sum_{j=1}^k m_j X_j) - \mathbb{E}(N -\widetilde{N}) +\widetilde{\negl}\\
    &= \mathbb{E}(\uanmna(\adv)) + \widetilde{\negl} &\text{by \cref{eq:multi-utility}}. &%\qedhere
\end{align}
%\qed 
\end{proof}

\begin{proof}[Proof of \cref{thm:sqaruf}]\label{pf:thm:sqaruf}
We will show that the private coin scheme described in~\cite{JLS18}, is \mnauf, using \cref{thm:strong_noclon_prs}. %Since, it is a private money scheme with classical private key, the previous statement implies it is also \auf (see \cref{definition:adapt_flex_unforge}).
Recall that in~\cite{JLS18}, with respect to a PRS $\{\ket{\phi_k}\}_{k\in \mathcal{K}}$, $\keygen$ randomly choses $k \in \mathcal{K}$, and $k$ is fixed as the secret key. The $\bank$ oracle returns the state $\ket{\phi_k}$, i.e., the private coin is a uniformly random sample from the PRS. The $\verify$ algorithm performs a projective measurement $\{\ketbra{\phi_k}, I - \ketbra{\phi_k}\}$, and accepts iff the outcome is $\ketbra{\phi_k}$. 
Suppose there exists a multiverifier QPT forger \adv, in Game~\ref{game:multi_nonadapt_unforge}, against the scheme which attacks $k$ many branches. Let $n$ be the maximum number of coins it asks for in all possible runs. Since, $\adv$ is polynomial, $n$ must be polynomial. We will construct an oracle-based cloner \bdv, which receives $n$ copies of a uniformly random state $\ket{\phi_k}$ chosen from the PRS family $\{\ket{\phi_k}\}_{k \in \mathcal{K}}$, has oracle access to $U_{\phi_k} = I - 2\ketbra{\phi_k}$, and submits $n+1$ alleged copies. On receiving $n$ copies of $\ket{\phi_k}$, a (uniform) randomly chosen state from the PRS, $\bdv$ simulates $\adv$ in Game~\ref{game:multi_nonadapt_unforge}. Every query to $\bank$ oracle, is simulated by giving a copy of $\ket{\phi_k}$. Every time, $\adv$ returns alleged coins to be submitted to a bank branch, $\bdv$ simulates $\verify$, using the reflection oracle, $U_{\phi_k}$. $\bdv$ stores the coins which passed verification, and hence are in the state $\ket{\phi_k}$, and discards the coins that did not pass verification. Suppose, $\adv$ made $j$ many queries to $\bank$ in the multiverifier game, Game~\ref{game:multi_nonadapt_unforge}, and $j'$ be the number of coins that passed verification in total. Hence, the multiverifier utility of $\adv$ (see \cref{definition:byzantine-multi-unforge}), \[\uflmna(\adv) = j'-j.\] If $j' > j$, i.e., $\adv$ succeeds in forging ($\uflmna(\adv) > 0$), $\bdv$ submits the first $n+1$ registers out of the $n-j+j'$ registers in its possession ($n-j'$ unused registers and $j'$ registers that passed verification), each of which are in the state $\ket{\phi_k}$. Else, $\bdv$ aborts. Hence, $\bdv$ clones with fidelity $1$ if and only if $\adv$ succeeds in cloning. Therefore, the fidelity with which $\bdv$ clones,
\[\mathbb{E}_{k\in \mathcal{K}}\left\langle (\ketbra{\phi_k})^{\tensor (n+1)},\mathcal{C}^{U_{\phi_k}}((\ketbra{\phi_k})^n) \right\rangle = \Pr[\uflmna(\adv) > 0].\]
Hence, by \cref{thm:strong_noclon_prs}, there must exist a negligible function $\negl$, such that
\[\Pr[\uflmna(\adv) > 0] = \negl.\]
%\qed 
\end{proof}

%proof of private untraceability that we commented out for qip2020 submissions
\begin{proof}[Proof of \cref{thm:unconditional_secure}]\label{pf:thm:unconditional_secure}
Combining \cref{prop: rational-unforge,prop:completeness,thm:rational priv fair,thm:rational public fair} %\cref{prop:completeness,prop: rational-unforge} 
with \cref{thm:qaruf,thm:qaruuf}, we get the result as follows. 

By \cref{thm:qaruf,thm:qaruuf}, the private coin schemes constructed in~\cite{JLS18} \anote{(or the simplified version in~\cite{BS19})} and~\cite{MS10} have a rank-$1$ projective measurement as the private verification and are \nauf and \nauuf, respectively. Hence, our construction $\pkqc$ instantiated with any of these two schemes will be \arnauf or unconditionally \arnauf, respectively by \cref{prop: rational-unforge}, and would also be \narf, i.e., \narpf and \nrrsp  by  \cref{thm:rational priv fair} and \cref{thm:rational public fair} respectively.

The private coin scheme in~\cite{AMR20} is statistically indistinguishable from the quantum money scheme in~\cite{MS10}, i.e., given oracle access to mint and verification procedures, no adversary can distinguish between the two schemes, even if it is computationally unbounded. This holds because by~\cite[Theorem 7]{AMR20} giving oracle access to the mint and verification procedures of the stateful private quantum coin scheme, is statistically indistinguishable from doing the following: Initially, a Haar random state is sampled, minting produces copies of the sampled state and verification is done by doing a rank-$1$ projective measurement on to the sampled state. This is precisely the private coin scheme in~\cite{MS10}. 

Note that, any adversary (forger or sabotage adversary) has access only to the minting and verification procedures of $\pkqc$, $\pkqc.\bank$ and $\pkqc.\verify$, respectively (see Game~\ref{game:nonadapt_unforge_strongest}, and Game~\ref{game:nonadapt_rational-fair}). These procedures are built only using the private minting and verification procedures, $\prqc.\bank$ and $\prqc.\verify$ (see \cref{alg:ts}). Hence, the resulting stateful public coin scheme, obtained by instantiating our construction $\pkqc$ with the stateful private coin scheme in~\cite{AMR20}, should satisfy the same sabotage and unforgeability properties as that of the one obtained by instantiating with~\cite{MS10}. %The private untraceability holds irrespective of the private scheme.
Hence, all the required security properties hold for our construction $\pkqc$ instantiated with the stateful private coin scheme in~\cite{AMR20}.
%\qed 
\end{proof}

\begin{proof}[Proof of \cref{thm:multi-unconditional_secure}]\label{pf:thm:multi-unconditional_secure}
    The proof follows by combining \cref{cor:multi-rational-secure} with \cref{thm:sqaruf,thm:sqaruuf} using the same template as in the proof of \cref{thm:unconditional_secure}, given on \cpageref{pf:thm:unconditional_secure}.
\end{proof}

\fi

\ifnum\masterthesis=1
    \bibliography{almost_public_quantum_coins_crypto}
\fi
\ifnum\masterthesis=1
\includepdf{Hebrew_Title_Page.pdf}
\fi
\ifnum\shownomenclature=1
    % \section{Nomenclature}
    % \label{sec:nomenclature} 
    %There is a label which is created automatically in the macros of the following form: \label{sec:nomenclature}
    \renewcommand{\nomname}{Nomenclature}
    %a fix to remove the by default nomenclature paragraph heading.
    %\nom is a new command declared in macros.tex which takes 3 arguments: the first argument to sort, the second for the name and third for the description and prints the second argument there.
    \printnomenclature[1in] %Removed the "Nomenclature" title, since this is already the heading of the appendix
\fi 
\end{document}

%%%%%%%%%%%%%%%

%%%%%%%%%%%%%